\documentclass[12pt,fleqn]{article}
\usepackage{mathtools}
\usepackage[margin=0.7in]{geometry}
\usepackage{amsmath,bm}
\usepackage{amsfonts}
\usepackage{amsthm}
\usepackage{comment}
\usepackage{multirow}
\usepackage{graphicx}
\usepackage{tikz}
\usepackage{lscape}
\usetikzlibrary{decorations}
\usetikzlibrary{decorations.markings,arrows}

\newtheorem{Thm}{Theorem}
\newtheorem{Prop}[Thm]{Proposition}
\newtheorem{Lemma}[Thm]{Lemma}
\newtheorem{Remark}[Thm]{Remark}
    \newenvironment{remark}{\begin{Remark}\rm}{\end{Remark}}
\theoremstyle{definition}

    \def\tr{{\rm tr \,}}
    \def\Re{{\rm Re \,}}
    \def\Im{{\rm Im \,}}

    \def\bigO{{\cal O}}

    \def\m{{(m)}}
        \def\P2n{{\rm P}_{{\rm II}}^{(n)}}

\title{
Sine-kernel determinant on two large intervals} 
\author{B. Fahs and I. Krasovsky}
\date{}

\begin{document}
\maketitle

\begin{abstract}
We consider the probability of two large gaps (intervals without eigenvalues) in the bulk scaling limit of the Gaussian Unitary Ensemble of random matrices. We determine the multiplicative constant in the asymptotics. 
We also provide the full explicit asymptotics (up to decreasing terms) for the transition between one and two large gaps.  
\end{abstract}

\section{Introduction}

Let $K_s$ be the (trace class) operator on $L^2(A)$, $A\subset\mathbb R$, with kernel  
$K_s(x,y)=\frac{\sin s(x-y)}{\pi (x-y)}$. Consider the Fredholm determinant
\begin{equation}\label{1}
P_s(A)=\det(I-K_s)_A.
\end{equation}
In this paper, we take $A$ to be the union of two intervals, and complete the 
description of  the asymptotics of $P_s(A)$ as $s\to\infty$, including the transition when the intervals 
merge into one.

The determinant \eqref{1} is the probability that the set $\frac{s}{\pi} A=\{\frac{s}{\pi}x: x\in A\}$  
contains no eigenvalues
of the Gaussian Unitary Ensemble of random matrices in the bulk scaling limit where the average distance between eigenvalues is 1. (Thus, if $A$ is a union of intervals, they are called gaps.)
Similar statements hold in other contexts: the sine-process with
kernel  $K_s(x,y)$ is the simplest, and one of the most common and well-studied determinantal processes
appearing in random matrix theory, random partitions, etc.

In the case where $A$ is an interval, which we can assume without loss\footnote{$P_s(A)$ is invariant under translations of $A$, and rescaling results only in the appearance of a prefactor of $s$.} to be $(-1,1)$,
the asymptotics of the logarithm of \eqref{1} have the form:
\begin{equation}\label{1gap}
\log P_s((-1,1))=-\frac{s^2}{2}-\frac{1}{4}\log s+c_0+\mathcal O(s^{-1}),\qquad s\to \infty,
\end{equation}
where
\begin{equation}\label{const}
c_0=\frac{1}{12}\log 2+3\zeta'(-1).
\end{equation}
Here $\zeta'(z)$ is the derivative of Riemann's zeta function. 

The leading term $-\frac{s^2}{2}$ was found by Dyson in 1962 in one of his fundamental papers on random matrix theory \cite{Dyson62}. Dyson used Coulomb gas arguments.
The terms $-\frac{s^2}{2}-\frac{1}{4}\log s$ were computed by des Cloizeaux and Mehta \cite{CM} in 1973
who used the fact that eigenfunctions of $K_s$ are spheroidal functions. The constant \eqref{const}, known as the Widom-Dyson constant, was identified
by Dyson \cite{Dyson} in 1976 who used the inverse scattering techniques and the earlier work of Widom
\cite{W1} on Toeplitz determinants. The works \cite{Dyson62}, \cite{CM}, and \cite{Dyson} are not fully rigorous.
The first rigorous confirmation of the main term, i.e. the fact that $\log P_s((-1,1))=-\frac{s^2}{2}(1+o(1))$, was given by Widom \cite{W2}  in 1994.
%who also found and proved the main term for $P_s(A)$ in the case of $A$ being a union of several intervals. 
The full asymptotic expansion \eqref{1gap}, apart from the expression \eqref{const} for $c_0$, was 
proved by Deift, Its, and Zhou in a landmark work \cite{DIZ} in 1997, where the multi-interval case was also addressed.
The authors of  \cite{DIZ} used Riemann-Hilbert
techniques to determine asymptotics of the logarithmic derivative $\frac{d}{ds}\log P_s(A)$, where $A$
is one (or a union of several) interval(s).  The asymptotics for $P_s(A)$ were then obtained in \cite{DIZ}
by integrating the  logarithmic derivative with respect to $s$.
The reason the expression for $c_0$ was not established in \cite{DIZ}
is that there is no initial integration point $s=s_0$ where $P_s(A)$ would be known explicitly.
In \cite{K04}, the author was able to justify the value of $c_0$ in \eqref{const} by using a different differential identity for associated Toeplitz determinants and again the result of Widom \cite{W1}.
An alternative proof of \eqref{const}
was given in \cite{DIKZ}, which was based on another
differential identity for Toeplitz determinants. In \cite{DIKZ}, the result of \cite{W1} was 
also rederived this way.
Both \cite{K04}  and \cite{DIKZ} relied on Riemann-Hilbert techniques. Yet another proof of \eqref{const} was given by Ehrhardt \cite{Ehrhardt} who used a very different 
approach of operator theory.

If $A$ is a union of several intervals, it was shown by Widom in \cite{W3} that
\begin{equation}
\frac{d}{ds}\log P_s(A)=-C_1s+C_2(x)+o(1),\qquad s\to\infty,
\end{equation}
where $C_1>0$ and $C_2(x)$ is a bounded oscillatory function.
The constant $C_1$ can be computed explicitly, but $C_2(x)$ is an implicit solution of a Jacobi inversion problem. This result was extended and made more explicit by Deift, Its, and Zhou in \cite{DIZ}. We will now present the 
solution of \cite{DIZ}
 in the case when $A$ is the union of two intervals, which is relevant for the present work.

As above, we assume without loss that 
\[
A=(-1,v_1)\cup (v_2,1),\qquad -1<v_1<v_2<1.
\] 
Let $p(z)=(z^2-1)(z-v_1)(z-v_2)$, and consider the two-sheeted Riemann surface $\Sigma$ of the function $p(z)^{1/2}$. On the first sheet $p(z)^{1/2}/z^2\to 1$ as $z\to\infty$,
while on the second, $p(z)^{1/2}/z^2\to -1$ as $z\to\infty$. The sheets are glued at the cuts $(-1,v_1)$, $(v_2,1)$.
Each point $z\in\overline{\mathbb C}\setminus((-1,v_1)\cup(v_2,1))$ (including infinity) has two images on $\Sigma$. The Riemann surface $\Sigma$ is topologically a torus.

\begin{figure}
	\begin{center}
		\begin{tikzpicture}
		\node [above] at (0,0){$-1$};
		\node [above] at (2,0){$v_1$};
		\node [above] at (5,0) {$v_2$};
		\node [above] at (7,0) {$1$};

		\draw[black,fill=black]  (0,0) circle [radius=0.04];
		\draw[black,fill=black]  (2,0) circle [radius=0.04];
		\draw[black,fill=black]  (5,0) circle [radius=0.04];
		\draw[black,fill=black]  (7,0) circle [radius=0.04];

		\node [below] at (6.5,-1) {$A_1$};
		\node [below] at  (0.6,-1){$A_0$};
		\node [below] at (3.7,-0.88){$B_1$};

		\draw  (5,0)--(7,0);
		\draw  (0,0)--(2,0);

		\draw[dashed,decoration={markings, mark=at position 0.5 with {\arrow[thick]{>}}},
		postaction={decorate}] (1,0) to [out=90,in=180] (3.5,1.5) to [out=0,in=90]  (6,0);
		
		\draw[decoration={markings, mark=at position 0.5 with {\arrow[thick]{<}}},
		postaction={decorate}] (1,0) to [out=270,in=180] (3.5,-1.5) to [out=0,in=270]  (6,0);

		\draw[decoration={markings, mark=at position 0.5 with {\arrow[thick]{>}}},
		postaction={decorate}] (-1,0) to [out=90,in=180] (1,1) to [out=0,in=90]  (3,0);
		\draw[decoration={markings, mark=at position 0.5 with {\arrow[thick]{<}}},
		postaction={decorate}] (-1,0) to [out=270,in=180] (1,-1) to [out=0,in=270]  (3,0);     
		
		\draw[decoration={markings, mark=at position 0.5 with {\arrow[thick]{>}}},
		postaction={decorate}] (4,0) to [out=90,in=180] (6,1) to [out=0,in=90]  (8,0);
		\draw[decoration={markings, mark=at position 0.5 with {\arrow[thick]{<}}},
		postaction={decorate}] (4,0) to [out=270,in=180] (6,-1) to [out=0,in=270]  (8,0);

		\end{tikzpicture} 
		\caption{Cycles on the Riemann surface $\Sigma$}\label{Fig1}\end{center}
\end{figure}
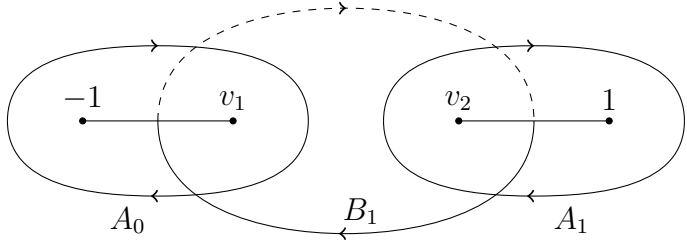

Let the elliptic integrals $I_j=I_j(v_1,v_2)>0$, $J_j=J_j(v_1,v_2)>0$ be given by
\begin{equation}\label{IJ}
I_j=\int_{v_2}^{1}\frac{x^j dx}{\sqrt{|p(x)|}}=
\frac{i}{2}\int_{A_1}\frac{x^j dx}{p(x)^{1/2}},\qquad 
J_j=\int_{v_1}^{v_2}\frac{x^j dx}{\sqrt{|p(x)|}}=
\frac{1}{2}\int_{B_1}\frac{x^j dx}{p(x)^{1/2}}
,\qquad j=0,1,2,
\end{equation}
where the loops (cycles) $A_1$, $B_1$ are shown in Figure \ref{Fig1}. The loops $A_0$, $A_1$ lie on the first sheet,
and the loop $B_1$ passes from one to the other: the part of it denoted by a solid line is on the first sheet, the other is on the second.

Let
\begin{equation}\label{defpsi}
\psi(z)=\frac{q(z)}{p(z)^{1/2}},\qquad q(z)=(z-x_1)(z-x_2),
\end{equation}
where the constants  $x_1\in (-1,v_1)$ and $x_2\in (v_2,1)$ are defined by the conditions
\begin{equation}\label{int0}
\int_{A_j}\psi(z) dz=0,\qquad j=0,1.
\end{equation}
It follows that 
 \begin{align}
x_1+x_2&=\frac{v_1+v_2}{2}\label{x1x2eqn1},
\\  x_1x_2&=
\left(-I_2+\frac{v_1+v_2}{2}I_1\right)\frac{1}{I_0},
\label{x1x2eqn2}
\end{align}
which gives an explicit expression for $q(z)$ in terms of elliptic integrals.

Note that \eqref{int0} implies that $\psi(z)$ has no residue at infinity. More precisely, we obtain as $z\to\infty$ on the first sheet
\begin{equation}\label{G0}
\psi(z)=1+\frac{G_0}{z^2}+\mathcal O(z^{-3}),\qquad G_0=
-\frac{I_2-\frac{v_1+v_2}{2}I_1}{I_0}+\frac{1}{2}+\frac{(v_2-v_1)^2}{8}.
\end{equation}
As shown in \cite{DIZ}, $G_0>0$.

Denote the holomorphic differential 
\begin{equation}\label{def:omega}
\bm{{\omega}}=i\frac{dz}{2I_0 p(z)^{1/2}}.
\end{equation}
 
Clearly, it is normalized:
\begin{equation}\label{omega cond}
\int_{A_1}\bm{{\omega}}=-\int_{A_0}\bm{{\omega}}=1.
\end{equation}

Let 
\begin{equation}\label{def:Omega}
\tau=\int_{B_1}\bm{{\omega}}=i\frac{J_0}{I_0},\qquad
\Omega=-\frac{1}{2\pi}\int_{B_1}\psi(x)dx=\frac{1}{\pi}\int_{v_1}^{v_2}\psi(x)dx=\frac{1}{I_0},
\end{equation}
where the integration $\int_{v_1}^{v_2}\psi(x)dx$ is taken on the first sheet, and where the last equation for $\Omega$ 
follows by Riemann's period relations (Lemma 3.45 in \cite{DIZ} for $n=1$).
Recall the definition \eqref{def:theta} in Appendix \ref{App1} of the third Jacobian $\theta$-function $\theta_3(z;\tau)$.
Deift, Its, and Zhou found in \cite{DIZ} that
\begin{equation}\label{FormDIZ}
\log P_s((-1,v_1)\cup (v_2,1))=-s^2G_0+\widehat G_1 \log s+\log \theta_3(s\Omega;\tau)+c_1+
\mathcal O(s^{-1}),
\qquad s\to\infty,
\end{equation}
with $G_0$ as in \eqref{G0}, and $\tau$, $\Omega$ as given in \eqref{def:Omega}. 
Constants $\widehat G_1$, $c_1$ are independent of $s$. The constant $\widehat G_1$
is written in \cite{DIZ} in terms of a limit of an integral of a combination of $\theta$-functions. The constant
$c_1$ remained undetermined (for the same reason as given above in the case of one interval).

The main result of the present paper is the expression for $c_1$, which completes the
description of the asymptotics \eqref{FormDIZ}. We also find that the original expression for $\widehat G_1$
in \cite{DIZ} can be simplified, and we obtain that $\widehat G_1=-1/2$ (see Appendix \ref{App2}).
We also determine this coefficient $-1/2$ of $\log s$ in a different way, as a direct result of our 
computation of \eqref{FormDIZ} which also produced $c_1$. We describe this computation in more detail below in the introduction.

Thus, we obtain
\begin{Thm}\label{Thm}
The asymptotics \eqref{FormDIZ} hold with
\begin{equation}\label{constant}
\widehat G_1=-\frac{1}{2},\qquad
c_1=-\frac{1}{2}\log\frac{I_0}{\pi}-\frac{1}{8}\sum_{y\in\{-1,v_1,v_2,1\}}\log |q(y)|+2c_0,\qquad
c_0=\frac{1}{12}\log 2+3\zeta'(-1).
\end{equation}
\end{Thm}

\begin{remark}
Using a connection between the elliptic integral $I_0$ and $\theta_3(0)$, equation \eqref{idth3} below, and substituting
$\widehat G_1$, $c_1$ into \eqref{FormDIZ}, we can write\footnote{
Perhaps, the corresponding formula for the logarithm of the probability of $n+1$ gaps $A=\cup_{j=0}^n(a_j,b_j)$ is
\begin{multline}
\log \det (I-K_s)_A=-\alpha s^2 -\frac{n+1}{4}\log s+\log \frac{\theta(sV)}{\theta(0)}\\+\frac{1}{4}\sum_{0\leq j<k\leq n} \log (b_k-b_j)(a_k-a_j) -\frac{1}{8}\sum_{j=0}^n\log|q(a_j)q(b_j)|+(n+1)c_0+\mathcal O(s^{-1}),
\qquad s\to\infty.
\end{multline}
The coefficient $\alpha$ here was determined in \cite{DIZ}, and $V$, $q$, and the multivariable $\theta$-function
are in the notation of \cite{DIZ}.} 
\begin{multline}
\log P_s((-1,v_1)\cup (v_2,1))=-s^2 G_0 -\frac{1}{2}\log s+\log \frac{\theta_3(s\Omega;\tau)}{\theta_3(0;\tau)}\\+\frac{1}{4}\log (1-v_1)(1+v_2) -\frac{1}{8}\sum_{y\in\{-1,v_1,v_2,1\}}\log |q(y)|+2c_0+\mathcal O(s^{-1}),
\qquad s\to\infty.
\end{multline}
\end{remark}

\begin{remark}
The elliptic integrals $I_j$, $J_j$ can be reduced to the complete ones.  In particular, 
in the symmetric case of 
$-v_1=v_2=v$, \eqref{FormDIZ} becomes (by a straightforward use of \eqref{I02} in Appendix \ref{App1})
\begin{equation}
\begin{aligned}
\log P_s((-1,-v)\cup (v,1))=-s^2\left(\frac{1+v^2}{2}-\frac{E(v')}{K(v')}\right)
-\frac{1}{2}\log\frac{s}{\pi}+
\log \theta_3\left(\frac{s}{K(v')};2i\frac{K(v)}{K(v')}\right)\\
-\frac{1}{4}\log[(K(v')-E(v'))(E(v')-v^2K(v'))]+2c_0+
\mathcal O(s^{-1}),
\end{aligned}
\end{equation}
where $v'=\sqrt{1-v^2}$, and $K(z)$, $E(z)$ are the complete elliptic integrals of first and second kind, respectively,
see \eqref{EllInt}.
\end{remark}

%Let $A=(-1,v_1)\cup (v_2,1)$. As $s\to \infty$,
%\begin{multline}\nonumber
%\log \det (I-K_s)=s^2\left(\frac{I_2-\frac{v_2-v_1}{2}I_1}{I_0}-\frac{1}{2}-\frac{(v_2-v_1)^2}{8}\right)-\frac{1}{2}\log s\\+\log \theta(s\Omega;\tau)-\frac{1}{2}\log(I_0/\pi)-\frac{1}{8}\sum_{y\in\{-1,v_1,v_2,1\}}\log |(y-x_1)(y-x_2)|
%+\frac{1}{6}\log 2+6\zeta'(-1)+o(1).
%\end{multline}\end{Thm}

\bigskip

The asymptotics \eqref{FormDIZ} with the coefficients given by \eqref{G0}, \eqref{def:Omega}, \eqref{constant}
can be extended (with a worse error term) to various double scaling regimes
where $v_1$, $v_2$ are allowed to approach each other or the endpoints $\pm1 $ at a sufficiently slow rate as 
$s\to\infty$: Theorems \ref{Thmtau0}, \ref{Thmtau1} below.
In Section \ref{sec-tau0}, we prove

\begin{Thm} (Extension to slowly merging gaps)\label{Thmtau0}
For a fixed $\epsilon>0$, let $-1+\epsilon\le v_1 < v_2\le 1-\epsilon$ be such that $2\nu\equiv v_2-v_1>s^{-5/4}$.
Then the asymptotics \eqref{FormDIZ} hold with the error term $\mathcal O(s^{-1/9})$.
% uniformly in $v_1$, $v_2$.
In particular, if $s\nu\to 0$ as $s\to\infty$, the expansion of the terms in \eqref{FormDIZ} gives
\begin{equation}\label{limThm1tau0}
\begin{aligned} \log P_s((-1,v_1)\cup (v_2,1))
&=
s^2\left(-\frac{1}{2}+\frac{|\alpha \beta|}{\log (\gamma \nu)^{-1}}\right)
-\frac{1}{2}\log s+\frac{1}{4}\log \log (\gamma \nu)^{-1}
-\langle \omega_0 \rangle^2 \log (\gamma \nu)^{-1}
\\& +\log \left(1+ (\gamma\nu)^{1-2|\langle \omega_0 \rangle |}\right)-\frac{1}{8} \log \left|\alpha \beta\right|
+2 c_0
+o(1),
\end{aligned}\end{equation}
where
$-\alpha=1+\frac{v_2+v_1}{2}>0$, $\beta=1-\frac{v_2+v_1}{2}>0$,
$\gamma=\frac{1}{8}(\beta^{-1}+|\alpha|^{-1})$, 
\begin{equation}\label{omegaintro}
\omega_0=\frac{s\sqrt{|\alpha \beta|}}{\log(\gamma\nu)^{-1}}>0,
\end{equation}
and $\langle x\rangle\in (-1/2,1/2]$ denotes the difference between $x$ and the integer nearest to it.

\end{Thm}

\begin{remark}
The rate $-5/4$ can be somewhat decreased with an appropriate change of the error term.
\end{remark}

\begin{remark}
Using the translational invariance of $\det (I-K_s)$, we see by the shift of variable $x\to x-\frac{v_1+v_2}{2}$ that
\[
P_s((-1,v_1)\cup (v_2,1))=P_s((\alpha,-\nu)\cup (\nu,\beta)).
\]
\end{remark}

Thus Theorem \ref{Thmtau0} provides the asymptotics for  $P_s((-1,v_1)\cup (v_2,1))$ in the case when
$|v_1-v_2|>s^{-5/4}$. 
In recent work \cite{FKduke}, we obtained the asymptotics of $P_s((-1,v_1)\cup (v_2,1))=
P_s((\alpha,-\nu)\cup (\nu,\beta))$ in the case
of two gaps merging into one, i.e. 
where $v_1$, $v_2$ are scaled with $s$ in such a way
that $|v_1-v_2|\le 1/(s \log^2 s)$ while being bounded away from $\pm 1$.
 We also showed implicitly that the asymptotics we obtained in that case
uniformly connect to those of fixed $v_1<v_2$. Theorem \ref{Thmtau0} provides  an explicit matching:
More precisely, we showed in \cite{FKduke} that\footnote{In \cite{FKduke}, $\beta-\alpha$ was arbitrary, but 
by a rescaling argument we can assume without loss that $\beta-\alpha=2$, which is the assumption
in the present work.}

\begin{Thm} (Splitting of the gap $(-1,1)$ \cite{FKduke})\label{Dukethm}
 As $s\to \infty$, uniformly for $\nu=\frac{v_2-v_1}{2} \in (0, \nu_0)$, where $ s\nu_0\log \nu_0^{-1} \to 0$, 
 \begin{equation} \label{duke1}
 \begin{aligned}
 \log P_s ((-1,v_1)\cup (v_2,1))&=
 -\frac{s^2}{2}
 +s\sqrt{|\alpha \beta|}\left(\omega_0-\frac{\langle\omega_0\rangle^2}{\omega_0}\right) -\frac{1}{4}\log s +c_0+
 \log \left(\frac{2^{2k^2-k}}{\pi^k}\frac{G(k+1)^4}{G(2k+1)}\right)\\
 &+\log \left(1+2\pi \kappa_{k-1}^2 (\gamma \nu)^{1+2\langle\omega_0\rangle}\right)
+ \log \left(1+(2\pi \kappa_k^2)^{-1}(\gamma \nu)^{1-2\langle\omega_0\rangle}\right)\\
&+\mathcal O \left(\max \left\{s\nu_0\log \nu_0^{-1},\frac{1}{\log \nu^{-1}_0},\frac{1}{s} \right\}\right)
,\qquad k=\omega_0-
\langle \omega_0 \rangle,
\end{aligned}\end{equation}  
where 
$G$ is the Barnes G-function, and where $\kappa_j$ is the leading coefficient of the Legendre polynomial of degree $j$ orthonormal on the interval $[-2,2]$,
given by 
\begin{equation} \label{def Leg coeff}
\kappa_j =4^{-j-1/2}\sqrt{2j+1}\frac{(2j)!}{j!^2},  \quad  j =1,2,\dots, \qquad \kappa_0=1/2, \qquad \kappa_{-1}=0. 
\end{equation} 
The rest of notation in \eqref{duke1} is from Theorem \ref{Thmtau0}.

 As $s\to \infty$, uniformly for $\nu\in(\nu_1,\nu_0)$, where
$ s\nu_0\log \nu_0^{-1}\to0,\,\, \frac{s}{\log \nu_1^{-1}}\to \infty$ (i.e., $k\to \infty$), formula \eqref{duke1} reduces to 
\begin{equation}\label{limThmduke}
\begin{aligned} \log P_s((-1,v_1)\cup (v_2,1))
=
s^2\left(-\frac{1}{2}+\frac{|\alpha \beta|}{\log (\gamma \nu)^{-1}}\right)
-\frac{1}{2}\log s+\frac{1}{4}\log \log (\gamma \nu)^{-1}
-\langle \omega_0 \rangle^2 \log (\gamma \nu)^{-1}
\\ +\log \left(1+ (\gamma\nu)^{1-2|\langle \omega_0 \rangle |}\right)-\frac{1}{8} \log \left|\alpha \beta\right|
+2 c_0+
 \mathcal O \left(\max \left\{s\nu_0\log \nu_0^{-1},\frac{1}{\log \nu^{-1}_0},\frac{\log \nu^{-1}_1}{s} \right\}\right).
\end{aligned}\end{equation}

\end{Thm}
 
Thus we see that the asymptotic regime of Theorem \ref{Thmtau0} overlaps with that of  Theorem \ref{Dukethm}
(for example, $\nu=s^{-6/5}$  belongs to both regimes), and 
comparing \eqref{limThm1tau0} with \eqref{limThmduke} we see an explicit matching. 
Taken together, these theorems describe the asymptotics for two large gaps and one large gap 
(note that \eqref{duke1} reduces to \eqref{1gap} when $\nu\to 0$ sufficiently rapidly)
as well as the transition between them.

\bigskip

Our strategy to prove Theorem \ref{Thm} relies on connecting the asymptotics for fixed $v_1<v_2$
with another double-scaling regime, namely the one where $v_1$ approaches $-1$, and $v_2$ approaches $1$.
In this regime the scaled gaps, $s(-1,v_1)$, $s(v_2,1)$,
although still growing with $s$, become small in comparison with the separation between them,
and we show that in that case $P_s((-1,v_1)\cup (v_2,1))$ splits to the main orders into the product
of $P_s(-1,v_1)$ and $P_s(v_2,1)$. The advantage is that for each of the separate gaps we can use
an appropriately rescaled asymptotics \eqref{1gap} which contains the constant $c_0$.
More precisely, we prove in Section \ref{secsep} by elementary arguments 
the following

\begin{Lemma}(Separation of gaps)\label{Lemmasep}
Let 
\[
A_s=\left(-1,-1+\frac{2t}{s}\right)\cup \left(1-\frac{2t}{s},1\right),
\qquad t=\frac{1}{2}(\log s)^{1/4}.
\]
Then
\begin{equation}
\log \det(I-K_s)_{A_s}= -t^2-\frac{1}{2}\log t+2c_0+\mathcal O(1/t),\qquad t\to\infty.
\end{equation}
\end{Lemma}

\begin{remark}\label{Lemmasepremark}
The rate of increase of $t$, $t=\frac{1}{2}(\log s)^{1/4}$, can be replaced with a slower rate of growth with $s$,
and the statement will still hold (cf. Theorem \ref{Thmtau1} below).
\end{remark}

Now we describe the steps of the proof of Theorem \ref{Thm}. First, we obtain in Section \ref{secDI} an identity (equation \eqref{diffid new} of  Lemma \ref{diffFred})
for the derivative $\frac{\partial}{\partial v_2}\log P_s((-1,v_1)\cup (v_2,1))$ in terms of 
a certain Riemann-Hilbert (RH) problem, the $\Phi$-RH problem. The fact that we use a differential identity with respect to one of the edges ($v_2$) of the gaps is crucial in allowing us to determine the constant $c_1$. 

%To derive our differential identity, we first consider an $n$-dimensional Toeplitz determinant with symbol $1$ supported on two arcs and obtain a differential identity for it. We then take a double scaling limit as the edges of the arcs approach each other at the rate 
%$s/n$, which yields the identity for 
%$\frac{\partial}{\partial v_2}\log P_s((-1,v_1)\cup (v_2,1))$ in terms of the $\Phi$-RH problem.

We then give in  Section \ref{SecR} the asymptotic solution of the $\Phi$-RH problem as $s\to\infty$ with $v_1$, $v_2$ fixed. This problem is very similar to that solved in \cite{DIZ}, and its solution involves the Jacobian $\theta$-functions (we give a collection of various useful properties of $\theta$-functions in the Appendix \ref{App1} below).  
In Section \ref{secextend}, we show that the solution of the $\Phi$-RH problem can be extended to the double-scaling range where $v_2$ is allowed to approach $1$ at such a rate that $(1-v_2)s\to\infty$ (by symmetry, also
$v_1$ is allowed to approach $-1$  so that $(1+v_1)s\to\infty$). It is this extension which eventually provides a connection with Lemma \ref{Lemmasep}.  

In Section \ref{secPrelim},
we then substitute the solution
into our differential identity (see \eqref{exact}, \eqref{notexact2}).
In Proposition \ref{PropD}, we
characterize the main asymptotic terms (equation \eqref{DD}) in the differential identity 
using averaging with respect to fast oscillations.

A large part of our work, Sections \ref{SecLead}, \ref{SecFluc}, \ref{SecConst}, is to bring the expression \eqref{DD} to an explicit form. This relies, apart from the use of standard formulae,
on (specific to our setting) identities for $\theta$-functions obtained in Lemma \ref{ThmThetaids} of Section \ref{secthid}. 
As a result, we obtain an explicit form \eqref{newD} for the non-small part \eqref{DD}  of the right-hand side of the differential identity \eqref{diffid new}. 

We then, by Proposition \ref{PropD}, 
integrate the resulting identity with respect to $v_2$ from the point when $v_2=-v_1$ is close to $1$
to a fixed $v_2=-v_1$, and then, with $v_1$ fixed, over $v_2$, so that at one of the integration limits we can use the result of Lemma \ref{Lemmasep}.  This proves Theorem \ref{Thm}. Thus
the part $2c_0$ of the constant $c_1$ in \eqref{constant} comes from Lemma \ref{Lemmasep}, while the rest of
$c_1$ comes from the integration. 

As a byproduct of our proof we also obtain the following extension of the asymptotics \eqref{FormDIZ}.
\begin{Thm} (Extension to separation of gaps)\label{Thmtau1}
For a fixed $\epsilon>0$, let $-1< v_1 < v_2<1$ be such that $v_2-v_1\ge \epsilon$,
$(1-v_2)s\to\infty$, $(1+v_1)s\to\infty$.
Then the asymptotics \eqref{FormDIZ} hold with the error term
$\mathcal O \left(\max \left\{\frac{1}{(1-v_2)s},\frac{1}{(1+v_1)s}\right\}\right)$. 
%uniformly in $v_1$, $v_2$.
\end{Thm}

Note that
the use of Toeplitz determinants in \cite{K04}, \cite{DIKZ} was essential to determine the constant $c_0$ in the asymptotics for one gap. In this paper, however, we use Lemma \ref{Lemmasep} which, in turn, relies on the already known constant $c_0$. 
%Here we use Toeplitz determinants to obtain a differential identity for convenience, but also to provide a 
%setup for future analysis of the probability of several gaps in the Circular Unitary Ensemble.

\section{Separation of gaps: proof of Lemma \ref{Lemmasep}}\label{secsep}
For $w>2$ let
\[
A^{(w)}=A_1^{(w)}\cup A_2^{(w)},\qquad
A_1^{(w)}=(-w,-w+1),\qquad A_2^{(w)}=(w-1,w).
\]
With $t$ as in Lemma \ref{Lemmasep} and $v=s/(2t)$, we have
\begin{equation}
\det(I-K_s)_{A_s}=\det(I-K_{2t})_{A^{(v)}}.
\end{equation}

By \eqref{1gap} and translational invariance, as $t\to\infty$,
\[
\det(I-K_{2t})_{A_1^{(v)}}=\det(I-K_{2t})_{A_2^{(v)}}=\det(I-K_{t})_{(-1,1)}=
e^{c_0}t^{-1/4}e^{-t^2/2}(1+\mathcal O(1/t)).
\]
Therefore, upon setting $u=2t$, $w=v$, we obtain
Lemma \ref{Lemmasep} as a direct consequence of the following lemma we now prove.

\begin{Lemma}\label{seplemma2}
Let $u,w>2$. There exist absolute constants $C_3,C_4>0$ such that 
\begin{equation}\label{decor1}
\left|\det(I-K_u)_{A^{(w)}}-\det(I-K_u)_{A_1^{(w)}}\det(I-K_u)_{A_2^{(w)}}\right|
\leq \frac{C_3}{w} e^{C_4 u^2}.
\end{equation}
\end{Lemma}

We start with 
\begin{Prop}\label{propos}
Let $m\in\{0,1,\dots\}$ and $B$ be an $m+1\times m+1$ matrix satisfying $|B_{jk}|\leq u$ for all $j,k=1,\dots,m+1$. Let $\widehat X$ be a set of indices $j,k$ such that $|B_{jk}|<1/w$ for all $(j,k)\in \widehat X$ and set
\[
\widehat B_{jk}=\begin{cases}
B_{jk}&{\rm if }\,\,(j,k)\not \in \widehat X,\\
0&{\rm if} \, \, (j,k)\in \widehat X.
\end{cases}
\]
Then
\begin{equation}\label{Hadamard}
|\det B-\det \widehat B|\leq \frac{1}{w} (C_1 u)^m \sqrt{m!}
\end{equation}
for a sufficiently large absolute constant $C_1>0$.
\end{Prop}

\begin{proof}
Let $B^{(0)}=B$ and
\begin{equation}
B^{(\ell)}_{jk}=\begin{cases}
B_{jk}&{\rm if }\,\,(j,k)\not \in \widehat X,\\
0&{\rm if} \, \, (j,k)\in \widehat X\, {\rm and}\, j\le\ell
\end{cases}\qquad \ell=1,\dots,m+1.
\end{equation}
In particular,
$
\widehat B=B^{(m+1)}
$.

Expanding $B$ and $B^{(1)}$ in the first row we have
\begin{equation}\label{B1}
|\det B- \det B^{(1)}|\leq \frac{1}{w} \sum_{k=1}^{m+1} |\det B^{(0)(1k)}|,
\end{equation}
where $B^{(0)(jk)}$ is the $m\times m$ matrix obtained by removing the $j$'th row and the $k$'th column from $B=B^{(0)}$. Similarly, for any $\ell=1,2,\dots,m$, expanding in the $\ell+1$ row,
we have
\begin{equation}\label{Bl}
|\det B^{(\ell)}- \det B^{(\ell+1)}|\leq \frac{1}{w} \sum_{k=1}^{m+1} |\det B^{(\ell)(\ell+1\, k)}|,
\end{equation}
Inequalities (\ref{B1}), (\ref{Bl}) imply
\begin{equation}
|\det B- \det \widehat B|\leq \frac{1}{w} \sum_{\ell=0}^{m} \sum_{k=1}^{m+1} |\det B^{(\ell)(\ell+1\, k)}|.
\end{equation}

Hadamard's inequality yields
\begin{equation}
|\det B^{(\ell)(\ell+1\, k)}|\leq u^{m} m^{m/2},
\end{equation}
and so
\begin{equation}
|\det B-\det \widehat B|\leq\frac{1}{w} (m+1)^2 u^m m^{m/2}\leq \frac{1}{w} (C_1 u)^m 
\sqrt{m!}
\end{equation}
for some $C_1>0$.
\end{proof}

\noindent{\it Proof of Lemma \ref{seplemma2}.}
Let
\begin{equation}\widehat K_u(x,y)=\begin{cases}
K_u(x,y)&{\rm if}\,\, x,y\in A_1^{(w)}\, {\rm or}\, x,y\in A_2^{(w)}\\
0&{\rm otherwise}.
\end{cases}
\end{equation}

If we set
\begin{equation}B=\det(K_u(x_j,y_k))_{j,k=1}^{m+1},\qquad \widehat B=\det(\widehat K_u(x_j,y_k))_{j,k=1}^{m+1},
\end{equation}
with $x_j,y_k\in A^{(w)}$, then $B, \widehat B$ satisfy the conditions of Proposition 
\ref{propos} for some $\widehat X$.
By \eqref{Hadamard} and the definition of the Fredholm determinant, 
we have for sufficiently large absolute constants $C_j>0$ 
\begin{equation}\label{dethatKs1}
\begin{aligned}
&|\det(I-K_u)_{A^{(v)}}-\det(I-\widehat K_u)_{A^{(v)}}|\\
&\leq\sum_{m=0}^{\infty} \frac{1}{(m+1)!}
\int_{A^{(w)}}dx_1\cdots \int_{A^{(w)}}dx_{m+1}
\left|\det(K_u(x_i,y_j))_{i,j=1}^{m+1}-
\det(\widehat K_u(x_i,y_j))_{i,j=1}^{m+1}\right|\\
&\leq \frac{1}{w}\sum_{m=0}^{\infty} \frac{(C_2 u)^m} {\sqrt{m!}}
\leq  \frac{1}{w}\sqrt{\sum_{m=0}^{\infty} \frac{(C_2 u)^{2m}(m+1)^2} {m!}}
\sqrt{\sum_{m=0}^{\infty} \frac{1}{(m+1)^2}}
\leq \frac{C_3}{w} e^{C_4 u^2}.
\end{aligned}
\end{equation}

The reason for introducing $\widehat K$ is that the corresponding Fredholm determinant
splits into the product of the determinants over $A_1^{(w)}$ and $A_2^{(w)}$.
Indeed,
\begin{multline}
\det(I-\widehat K_u)_{A^{(w)}}=I+\sum_{m=1}^\infty \sum_{k=0}^m \frac{(-1)^m}{(m-k)!k!}\\ \times \int_{\substack{x_1,\dots,x_k \in A_1^{(w)}\\ x_{k+1},\dots,x_m\in A_2^{(w)}}}\det K_u(x_i-x_j)_{i,j=1}^k\det K_u(x_i-x_j)_{i,j=k+1}^m dx_1\dots dx_m\\
=\det(I-K_u)_{A_1^{(w)}}\det(I-K_u)_{A_2^{(w)}}.
\end{multline}
Combining this with the estimate \eqref{dethatKs1} proves the lemma.
$\Box$

\section{Differential Identity}\label{secDI}
Consider the following Riemann-Hilbert problem for a $2\times 2$ matrix valued
function $\Phi(w)$. Let $\Gamma_{\Phi}$ be the contour shown in Figure \ref{ContourPhi},
where as usual the $+$ side of the contour is on the left w.r.t. the direction shown by the arrow, and the $-$ side is on the right.

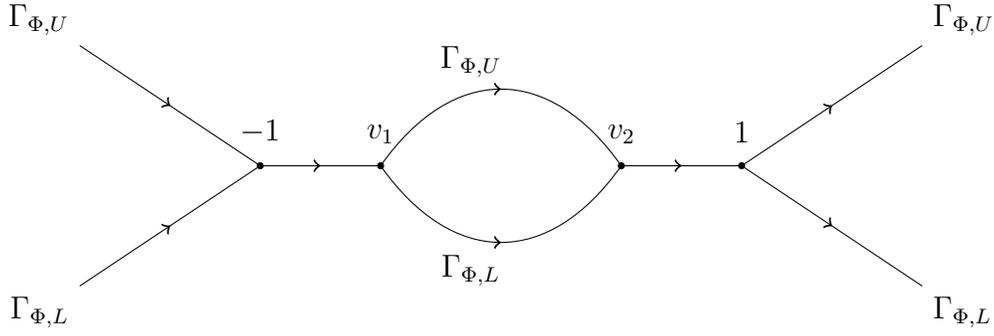
\begin{figure}\begin{center}
\begin{tikzpicture}[scale=0.8]
\draw [decoration={markings, mark=at position 0.5 with {\arrow[thick]{>}}},
        postaction={decorate}] (3,0) -- (5,0);
\draw [decoration={markings, mark=at position 0.5 with {\arrow[thick]{>}}},
        postaction={decorate}] (9,0) -- (11,0);
\draw [decoration={markings, mark=at position 0.5 with {\arrow[thick]{>}}},
        postaction={decorate}] (0,2) -- (3,0); 
\draw [decoration={markings, mark=at position 0.5 with {\arrow[thick]{>}}},
        postaction={decorate}] (11,0) -- (14,2); 
\draw [decoration={markings, mark=at position 0.5 with {\arrow[thick]{>}}},
        postaction={decorate}] (11,0) -- (14,-2); 
\draw  [decoration={markings, mark=at position 0.5 with {\arrow[thick]{>}}},
        postaction={decorate}] (5,0) .. controls (6.2,-1.7) and (7.8,-1.7) .. (9,0);     
\draw  [decoration={markings, mark=at position 0.5 with {\arrow[thick]{>}}},
        postaction={decorate}] (0,-2) -- (3,0);
\draw  [decoration={markings, mark=at position 0.5 with {\arrow[thick]{>}}},
        postaction={decorate}] (5,0) .. controls (6.2,1.7) and (7.8,1.7) .. (9,0);  
\node [above left] at (0,2) {$\Gamma_{\Phi,U}$};
\node [above right] at (14,2) {$\Gamma_{\Phi,U}$};
\node [above ] at (6.5,1.3) {$\Gamma_{\Phi,U}$};
\node [below left] at (0,-2) {$\Gamma_{\Phi,L}$};
\node [below right] at (14,-2) {$\Gamma_{\Phi,L}$};
\node [below ] at (6.5,-1.3) {$\Gamma_{\Phi,L}$};
\node [above ] at (9,0.2) {$v_2$};
\node [above ] at (11,0.2) {$1$};
\node [above ] at (5,0.2) {$v_1$}; 
\node [above ] at (3,0.2) {$-1$};
\draw[fill] (9,0) circle [radius=0.05];
\draw[fill] (3,0) circle [radius=0.05]; 
\draw[fill] (5,0) circle [radius=0.05];   
\draw[fill] (11,0) circle [radius=0.05];  
%\node [left] at (-1,0) {$\mathbb R$};    
\end{tikzpicture}
\caption{The jump contour $\Gamma_{\Phi}$}\label{ContourPhi}
\end{center}
\end{figure}
\subsubsection*{RH problem for $\Phi$}
\begin{itemize}
\item[(a)] $\Phi$ is analytic for $w \in \mathbb C \setminus \Gamma_{\Phi}$.

\item[(b)] $\Phi$ has $L^2$ boundary values $\Phi_+(w)$, $\Phi_-(w)$ as the point 
$w\in\Gamma_{\Phi}$ is approached nontangentially from the $+$ side, $-$ side, respectively. 
These values are related by the jump condition $\Phi_+(w)=\Phi_-(w)J_{\Phi}(w)$, where 
\begin{equation}
J_{\Phi}(w)=
\begin{cases}
\begin{pmatrix}
0&-1\\1&0
\end{pmatrix}&\textrm{for }w\in I= (-1,v_1)\cup (v_2,1),\\
\begin{pmatrix}
1&0\\1&1
\end{pmatrix}
&\textrm{for } w\in \Gamma_{\Phi, \rm L},\\
\begin{pmatrix}
1&-1\\0&1
\end{pmatrix}
&\textrm{for } w\in \Gamma_{\Phi, \rm U}.
\end{cases}
\end{equation}
\item[(c)] As $w \to \infty$,
\begin{equation}\label{44}
\Phi(w)=\left(I+\mathcal O\left(\frac{1}{w}\right)\right)
\begin{pmatrix} e^{isw} & 0 \cr 0 & e^{-isw}
\end{pmatrix}.
\end{equation}
\end{itemize}

\noindent{\it Remarks}

\noindent
1) As usual, we write for brevity
\[
\begin{pmatrix} e^{isw} & 0 \cr 0 & e^{-isw}
\end{pmatrix}=
e^{isw \sigma_3},\qquad
\sigma_3=\begin{pmatrix} 1 & 0 \cr 0 & -1\end{pmatrix}.
\]

\noindent
2) By general theory, see, e.g., \cite{Dbook}, if this problem has a solution $\Phi(w)$, then the solution is unique.
Although it is unclear a-priori that it has a solution, we will show this in Section \ref{SecR} for large $s$, which is the case we need.

The rest of this section will be devoted to 2 different proofs of the following

\begin{Lemma} (Differential identity)\label{diffFred}
The Fredholm determinant (\ref{1}) satisfies:
\begin{equation}\label{diffid new}
\frac{\partial}{\partial v_2}\det(I-K_s)_{(-1,v_1)\cup (v_2,1)}=\mathcal F_s(v_1,v_2)\equiv\frac{i}{2\pi}\left[\Phi_+^{-1}(v_2)\Phi_+'(v_2)\right]_{12},
\end{equation}
where $\Phi'(z)=\frac{d}{dz}\Phi(z)$ and
$\Phi_+^{-1}(v_2)\Phi_+'(v_2)=\lim_{\epsilon\downarrow 0}
\Phi^{-1}(v_2+i\epsilon)\Phi'(v_2+i\epsilon)$.
Moreover, if $-v_1=v_2=v$,
\begin{equation}\label{diffidsym}
\frac{\partial}{\partial v}\det(I-K_s)_{(-1,-v)\cup(v,1)}=2\mathcal F_s(-v,v). 
\end{equation}

\end{Lemma}

\subsection{First proof of Lemma \ref{diffFred}} 
The proof of identities of type \eqref{diffid new} using the theory of integrable operators is standard
\cite{IIKS, DIZ, BBIK, BD}. We give an outline.
First, we write the kernel of the (integrable)  operator $K_s$ in the form
\begin{equation}\label{kernel fh}
K_s(x,y)=\frac{\vec \lambda^T(x)\vec \mu(y)}{x-y}=\frac{\sum_{j=1}^2 \lambda_j(x)\mu_j(y)}{x-y}
, \qquad \vec \lambda(z)=\begin{pmatrix}e^{isz}\\
-e^{-isz}\end{pmatrix},\qquad \vec \mu=\frac{1}{2\pi i}\begin{pmatrix}e^{-isz}\\
e^{isz}\end{pmatrix}.
\end{equation}
Note that $\sum_{j=1}^2 \lambda_j(z)\mu_j(z)=0$. 
The resolvent of the operator $K_s$,
\[
(I-K_s)^{-1}=I+R_s,
\]
has the property \cite[Lemma 2.8]{DIZ} that the kernel of $R_s$ is of the form
\begin{equation}\label{resolvent}
R_s(x,y)=\frac{\vec \Lambda^T(x)\vec M(y)}{x-y}, \qquad  \Lambda_j=(I-K_s)^{-1}\lambda_j,\qquad 
M_j=(I-K_s^T)^{-1}\mu_j,\qquad j=1,2,
\end{equation}
and moreover, $\sum_{j=1}^2 \Lambda_j(z)M_j(z)=0$. 
The functions $\Lambda(z)$ and $M(z)$ for $z\in A$ can be written as \cite[Lemma 2.12]{DIZ}
\begin{equation}\label{FG}
\vec \Lambda(z)=\widehat m_+(z)\vec \lambda(z), \qquad \vec M(z)=(\widehat m_+^{-1}(z))^T\vec \mu(z),
\end{equation}
where $\widehat m(z)$ is the $2\times 2$ matrix valued function which 
solves the following RHP (this is the $m$-RHP of \cite{DIZ} up to a slight modification: $\lambda_2$, $\mu_2$ are replaced
by $-\lambda_2$, $-\mu_2$, respectively):
\subsubsection*{RH problem for $\widehat m$}
\begin{itemize}
\item[(a)] $\widehat m(z)$ is analytic in $\mathbb C\setminus \overline{A}$.
\item[(b)] $\widehat m(z)$ has $L^2$ boundary values related by the condition
$\widehat m_+(x)=\widehat m_-(x)J_m(x)$ for $x\in A$, with
\begin{equation}
J_m(x)=I-2\pi i \vec \lambda(x)\vec \mu^T(x).
\end{equation}
\item[(c)] $\widehat m(z)= I+\bigO(z^{-1})$ as $z\to \infty$.
\end{itemize}

This problem is reduced to a constant jump problem by the transformation
\begin{equation}
\widehat\psi(z)=\widehat m(z)e^{isz\sigma_3}.
\end{equation}
Indeed so defined $\widehat\psi(z)$ satisfies
\subsubsection*{RH problem for $\widehat\psi(z)$}
\begin{itemize}
\item[(a)] $\widehat\psi(z)$ is analytic in $\mathbb C\setminus \overline{A}$.
\item[(b)] $\widehat\psi(z)$ has $L^2$ boundary values related by the condition
$\widehat\psi_+(x)=\widehat\psi_-(x)
\begin{pmatrix}
0 & -1\\
1 & 2
\end{pmatrix}$ for $x\in A$. 
\item[(c)] $\widehat\psi(z)= \left(I+\bigO(z^{-1})\right) e^{isz\sigma_3} $ as $z\to \infty$.
\end{itemize}

It is now straightforward to verify that the solution to the $\Phi$-RH problem is written in terms of $\widehat\psi(z)$
as follows: $\Phi(z)=\widehat\psi(z)\begin{pmatrix}
1 & -1\\
0 & 1
\end{pmatrix}$ above $\Gamma_{\Phi, U}$ (see Figure \ref{ContourPhi});
$\Phi(z)=\widehat\psi(z)\begin{pmatrix}
1 & 0\\
-1 & 1
\end{pmatrix}$ below $\Gamma_{\Phi, L}$; and  $\Phi(z)=\widehat\psi(z)$ inside the lenses in Figure \ref{ContourPhi}.

Writing $\widehat m$ in terms of $\Phi$ in \eqref{FG}, we obtain
\begin{equation}\label{FG2}
\vec \Lambda(z)=\begin{pmatrix}-\Phi_{12,+}(z)\\
-\Phi_{22,+}(z)\end{pmatrix},\qquad
\vec M(z)=\frac{1}{2\pi i}\begin{pmatrix}\Phi_{22,+}(z)\\
-\Phi_{12,+}(z)\end{pmatrix},\qquad z\in A.
\end{equation}

Now the logarithmic derivative of the determinant
\begin{multline}\label{identityFredholmresolvent}
\frac{\partial}{\partial v_2}\log\det(I-K_s)_{(-1,v_1)\cup (v_2,1)}
=-\tr\left((I-K_s)^{-1}{\partial K_s\over \partial v_2}\right)=
((I-K_s)^{-1}K_s)(v_2,v_2)\\
=((I-K_s)^{-1}(K_s-I+I))(v_2,v_2)=R_s(v_2,v_2)=
-(\Lambda_1(v_2)M_1'(v_2)+\Lambda_2(v_2)M_2'(v_2)).
\end{multline}
Substituting here \eqref{FG2}, we obtain \eqref{diffid new}. The identity \eqref{diffidsym} is obtained similarly.

\subsection{Differential identity for Toeplitz determinants} 
For the second proof of Lemma \ref{diffFred},
we will first represent the Fredholm determinant $\det (I-K_s)_A$ in terms of a special Toeplitz determinant and then obtain (\ref{diffid new}) as a limit of the corresponding differential identity for Toeplitz determinants. This way of proving Lemma \ref{diffFred} has a potential advantage of future applications to computing probabilities in the Circular Unitary Ensemble of random matrix theory, and to the theory of orthogonal polynomials.

Let $J=J_1\cup J_2$ be the union of two disjoint arcs $J_1$ and $J_2$ on the unit circle $C$. 
We parametrize the endpoints of $J_1$ by $a_1=e^{i\phi_1},a_2=e^{i\phi_2}$  and the endpoints of $J_2$ by $b=e^{i\phi_0}$, $\bar{b}=e^{-i\phi_0}$,
see Figure \ref{J crit}. Let $f$ be the indicator function of the set $J$:
\begin{equation}\nonumber
f(z)=\begin{cases} 1 & \textrm{for $z \in J$,}\\ 0& \textrm{for $z \notin J$.} \end{cases} \end{equation}
Consider the $n$-dimensional Toeplitz determinant with symbol $f$ on the unit circle $C$:
\begin{equation}\nonumber
D_n(f)=\det(f_{j-k})_{j,k=0}^{n-1}, \qquad f_j=\int_C f(z) z^{-j}\frac{dz}{2\pi i z}=
\int_J z^{-j}\frac{dz}{2\pi i z},
\end{equation}
where the integration is in the counterclockwise direction.

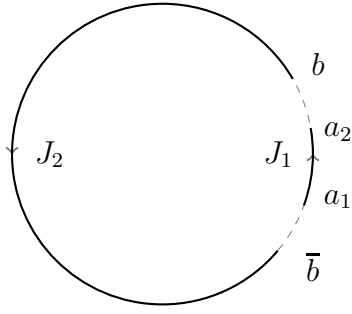
\begin{figure} 
\begin{center}
\begin{tikzpicture}
\node [right] at (2,0.3) {$a_2$};
\node [right] at (2,-0.6) {$a_1$};
\node [left] at (2.3,1.2) {$b$};
\node  at (2,-1.5) {$\overline{b}$};
\node  at (-1.5,0) {$J_2$};
\node [right] at (1.2,0) {$J_1$};
%\node  at (1.5,0.7) {$\Sigma_1$};
%\node  at (1.3,-1) {$\Sigma_2$};
\draw [decoration={markings, mark=at position 0 with {\arrow[thick]{>}}},
        postaction={decorate},decoration={markings, mark=at position 0.5 with {\arrow[thick]{>}}},
        postaction={decorate},dashed, gray] (0,0) circle[radius=2]; 
\draw[thick](2,0) arc [radius=2,start angle=0, end angle=10];
\draw[thick](2,0) arc [radius=2,start angle=360, end angle=340];
\draw[thick](-2,0) arc [radius=2,start angle=180, end angle=30];
\draw[thick](-2,0) arc [radius=2,start angle=180, end angle=320];
\end{tikzpicture}
\caption{Arc $J_1$ on the right and $J_2$ on the left.}\label{J crit}
\end{center}
\end{figure}

If the end-points of the arcs vary with $n$ as follows,
$\phi_0=2s/n$ and $\phi_j=2v_js/n$ for $j=1,2$, then it is easily verified that  
\begin{equation}\label{20}
\lim_{n\to\infty}D_n(f)=\det(I-K_s)_{(-1,v_1)\cup (v_2,1)}.
\end{equation}

We will now obtain a differential identity for $D_n(f)$, and in the next subsection,
by taking $n\to\infty$ and using (\ref{20}), will prove Lemma \ref{diffFred}.

Since $f$ is nonnegative, it follows from the multiple integral representation for Toeplitz determinants that $D_j(f)>0$ for all $j=1,2,\dots$. Set $D_0(f)=1$. 
Define the polynomials $\psi_0=1/\sqrt{f_0}$,
$\psi_j$, $j=1,2,\dots$ by
\begin{equation}\nonumber
\psi_j(z)=\frac{1}{\sqrt{D_j(f)D_{j+1}(f)}}\det 
{\small \begin{pmatrix} f_0 & f_{-1} & \dots &f_{-j+1} &f_{-j}\\
f_1& f_0 &\dots&f_{-j+2}&  f_{-j+1}\\
&&\ddots &\\
f_{j-1}&f_{j-2}& \dots& f_0 &f_{-1}\\
1&z&\dots &z^{j-1}&z^j \end{pmatrix}}=\chi_jz^j+\dots,
\end{equation}
where the leading coefficient $\chi_j$ is given by
\begin{equation}
\chi_j=\sqrt{\frac{D_j(f)}{D_{j+1}(f)}}.
\end{equation}
These polynomials are orthonormal on $J$:
\begin{equation}
\int_J \psi_k(z)\overline{ \psi_j(z)} \frac{dz}{2\pi iz}= \delta_{jk},\qquad j,k=0,1,\dots
\end{equation}
For a given $n\ge 1$, define the matrix-valued function $Y=Y(z)$ in terms of the orthogonal polynomials:
\begin{equation}\label{Soln Y}
Y(z)=\begin{pmatrix}
\chi_n^{-1}\psi_n(z)&\chi_n^{-1} \int_J \frac{\psi_n(\zeta)}{\zeta-z}\frac{d\zeta}{2\pi i \zeta^n} \\
-\chi_{n-1}z^{n-1} \overline{\psi}_{n-1}(z^{-1})& -\chi _{n-1} \int_J \frac{\overline{\psi}_{n-1}(\zeta^{-1})}{\zeta-z} \frac{d\zeta}{2\pi i \zeta}
\end{pmatrix}.
\end{equation}
The function $Y$ is a unique solution to the following RH Problem:
\begin{itemize}
\item[(a)] $Y:\mathbb C \setminus J \to \mathbb C^{2\times 2}$ is analytic;
\item[(b)]$Y_+(z)=Y_-(z)\begin{pmatrix} 1&z^{-n}\\0&1 \end{pmatrix}$ for $z\in J$;
\item[(c)] $Y(z)=(I+\mathcal{O}(1/z))z^{n\sigma_3}$ as $z\to \infty$.
\end{itemize}
This fact was initially noticed in \cite{FIK} for orthogonal polynomials on the real line and extended to the case of orthogonal polynomials on the unit circle in \cite{BDJ}.
As in \cite{K04,DIK},
we will use the orthogonal polynomials to obtain a differential identity for $\log D_n(f)$ in terms of the solution to the RH problem for $Y$. Namely, we have

\begin{Prop} \label{Prop Diffid}\begin{itemize}  \item[(a)] Let $a_2=e^{i\phi_2}$. 
The Toeplitz determinant $D_n(f)$ satisfies
\begin{equation}\label{diff id}
\frac{\partial}{\partial \phi_2}\log D_n(f)=-\frac{1}{2\pi}F(a_2),
\end{equation}
where $F$ is given by
\begin{equation} \label{alternative}
F(z)=-z^{-n+1}[Y^{-1}(z)Y'(z)]_{21}.
\end{equation}
\item[(b)] Let $a_2=\overline a_1=e^{i\phi_2}$. Then 
\begin{equation}\label{diffidsymT}
\frac{d}{d\phi_2}\log D_n(f)=-\frac{1}{\pi}F(a_2).
\end{equation}
\end{itemize}
\end{Prop}

\begin{proof}
From the definition of the orthogonal polynomials it is clear that
\begin{align}\label{Toeplitz and coeffs}
D_n(f)=&\prod_{j=0}^{n-1} \chi_j^{-2}.\end{align}
The orthogonality conditions imply that, with $z=e^{i\theta}$,
\begin{align} \frac{1}{2\pi} \int_J \frac{\partial \psi_j(z)}{\partial \phi_2} \overline{\psi_j(z)}d\theta =& \frac{1}{2\pi} \int_J \frac{\partial \chi_j}{\partial\phi_2} (z^j +\textrm{polynomial of degree $j-1$}) \overline{\psi_j(z)}d \theta
=\frac{1}{\chi_j} \frac{\partial \chi_j}{\partial \phi_2} ,
\end{align}
and similarly,
\begin{equation}\label{diff tricks}
 \frac{1}{2\pi} \int_J \psi_j(z)\frac{\partial \overline{\psi_j(z)}}{\partial \phi_2}d\theta=\frac{1}{\chi_j} \frac{\partial \chi_j}{\partial \phi_2} .
\end{equation}
By \eqref{Toeplitz and coeffs}--\eqref{diff tricks} we obtain:
\begin{equation}\label{Dnsumphij}
  \frac{\partial }{\partial \phi_2} \log (D_n(f))=-2\sum_{j=0}^{n-1} \frac{\partial \chi_j}{\partial \phi_2} /\chi_j
=-\frac{1}{2\pi} \int_J \frac{\partial}{\partial\phi_2} \left( \sum_{j=0}^{n-1} |\psi_j(z)|^2 \right)d\theta.
\end{equation}
The Christoffel-Darboux formula for orthogonal polynomials on the unit circle 
(see, e.g., equation (2.8) in \cite{DIK})
states that
\begin{equation}\label{CD}
-\sum_{k=0}^{n-1}| \psi_k(z)|^2=n|\psi_n(z)|^2 -2\Re \left(z\overline{\psi_n(z)}\psi_n'(z)\right) \quad \textrm{for } z \in C.
\end{equation}
On the other hand, using the following identity (equation (2.4) in \cite{DIK})
\begin{equation}
\chi_n\overline{ \psi_n(z)}=\chi_{n-1}z^{-1}\overline{ \psi_{n-1}(z)}+\overline{ \psi_n(0)} z^{-n}\psi_n(z),
\end{equation}
and (\ref{Soln Y}), we easily verify that
\begin{equation}\label{FOP}
F(z)=-z^{-n+1}[Y^{-1}(z)Y'(z)]_{21}=
n|\psi_n(z)|^2 -2\Re \left(z\overline{\psi_n(z)}\psi_n'(z)\right)
 \quad \textrm{for } z \in C.
\end{equation}
Substitution of \eqref{CD}, \eqref{FOP} into \eqref{Dnsumphij} gives
\begin{equation} \label{eqndiffid1}
\frac{\partial}{\partial \phi_2}\log D_n(f)=  \frac{1}{2\pi} \int_J \frac{\partial}{\partial \phi_2} \left(F(z)\right) d\theta.
\end{equation}

Since by orthogonality
\[
\int_JF(z)\frac{d\theta}{2\pi}=-\int_J \sum_{k=0}^{n-1}| \psi_k(z)|^2\frac{d\theta}{2\pi}=
-n,
\]
we obtain
\begin{equation}\label{eqndiffid2}
0=\frac{\partial}{\partial\phi_2}\left(\int_{J} F(z) d\theta\right)=F(a_2)+\int_{J} \frac{\partial}{\partial\phi_2}F(z)d\theta,
\end{equation}
and proposition \ref{Prop Diffid} (a) follows from \eqref{eqndiffid1}. Part (b) is proved similarly.
\end{proof}

\subsection{Limit $n\to \infty$. Second proof of Lemma \ref{diffFred}.}
As we are eventually interested in the limit $n\to\infty$, we first reduce the $Y$ RH problem to an approximate problem for $\Phi$ which does not contain the parameter $n$, and the dependence on $n$ is in the error of approximation. 

Let 
\begin{equation}
T(z)=\begin{cases}Y(z)&|z|<1,\\
Y(z)z^{-n\sigma_3}&|z|>1.
\end{cases}
\end{equation}

\begin{figure}\begin{center}
\begin{tikzpicture}
\node [right] at (1.9,0.6) {$a_2$};
\node [right] at (2.2,0) {$\Gamma_{\widehat S}^{\textrm{Out}}$};
\node [left] at (1.8,0) {$\Gamma_{\widehat S}^{\textrm{In}}$};
\node [right] at (1.9,-0.65) {$a_1$};
\node [right] at (-1.5,0) {$\Gamma_{\widehat S}^{\textrm{In}}$};
\node [left] at (-2.5,0) {$\Gamma_{\widehat S}^{\textrm{Out}}$};
\draw [
decoration={markings, mark=at position 0 with {\arrow[thick]{>}}},
        postaction={decorate},
decoration={markings, mark=at position 16/360 with {\coordinate (p1);}},
        postaction={decorate},
decoration={markings, mark=at position 40/360  with {\coordinate (p2);}},
        postaction={decorate},
decoration={markings, mark=at position 0.5 with {\arrow[thick]{>}}},
        postaction={decorate},
decoration={markings, mark=at position325/360  with {\coordinate (p3);}},
        postaction={decorate},
decoration={markings, mark=at position 0.75 with {\arrow[thick]{>}}},
        postaction={decorate},   
decoration={markings, mark=at position344/360  with {\coordinate (p4)  ;}},
        postaction={decorate}        
] (0,0) circle[radius=2]; 
\draw [decoration={markings, mark=at position 0.5 with {\arrow[thick]{>}}},
        postaction={decorate}](p4)  to [out=120, in=240] (p1);
\draw[decoration={markings, mark=at position 0.5 with {\arrow[thick]{>}}},
        postaction={decorate}] (p4) to [out=20, in=-20] (p1);
\draw[decoration={markings, mark=at position 0.5 with {\arrow[thick]{>}}},
        postaction={decorate}] (p2) to [out=70,in=0] (0,2.5) to [out=180, in=90] (-2.5,0) to [out=-90,in=180] (0,-2.5) to [out=0,in=-70]  (p3);
\draw [decoration={markings, mark=at position 0.5 with {\arrow[thick]{>}}},
        postaction={decorate}] (p2) to [out=-170,in=0] (0,1.5) to [out=180, in=90] (-1.5,0) to [out=-90,in=180] (0,-1.5) to [out=0,in=170]  (p3);
        \node [right] at (p2) {$b$};
        \node [right] at (p3) {$\overline{b}$};
\end{tikzpicture} 
\caption{Contour $\Gamma_{\widehat S}$.}\label{Contour GammaS}
\end{center}
\end{figure}
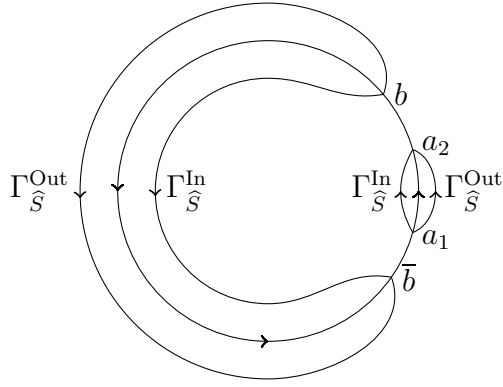

We open the lenses around $J_1$ and $J_2$, see Figure \ref{Contour GammaS}.
Denote the edges of the lenses inside  the unit disc by $\Gamma^{\rm In}_{\widehat S}$, the edges of the lenses outside the unit disc by $\Gamma^{\rm Out}_{\widehat S}$,
and let
\begin{equation}
{\widehat S}(z)=\begin{cases}
T(z)&\textrm{outside the lenses,}\\
T(z)\begin{pmatrix}1&0\\-z^{n}&1 \end{pmatrix}&\textrm{inside the lenses, for }|z|<1,\\
T(z)\begin{pmatrix}1&0\\z^{-n}&1 \end{pmatrix}&\textrm{inside the lenses, for }|z|>1.
\end{cases}\end{equation}

%\subsubsection*{RH problem for ${\widehat S}$}
Then ${\widehat S}$ satisfies the following RH problem:
\begin{itemize}
\item[(a)] ${\widehat S}$ is analytic on $\mathbb C\setminus \left(C\cup \Gamma_{\widehat S}^{\rm In}\cup \Gamma^{\rm Out}_{\widehat S}\right)$.
\item[(b)] The jumps of ${\widehat S}$ are given by ${\widehat S}_+(z)={\widehat S}_-(z)J_{\widehat S}(z)$, where
\begin{equation} \nonumber 
J_{\widehat S}(z)=\begin{cases}
\begin{pmatrix}0&1\\-1&0\end{pmatrix}&\textrm{for } z\in J,\\
\begin{pmatrix} 1&0\\ z^{-n}&1 \end{pmatrix}&\textrm{for } z\in \Gamma_{\widehat S}^{\rm Out}\\
\begin{pmatrix} 1&0\\ z^{n}&1 \end{pmatrix}&\textrm{for } z\in \Gamma_{\widehat S}^{\rm In},\\
z^{n\sigma_3}&\textrm{for } z\in C\setminus J.
\end{cases}
\end{equation}
\item[(c)] As $z\to \infty$,
\begin{equation}{\widehat S}(z)=I+\mathcal O\left(z^{-1}\right).\end{equation}
\end{itemize}

We assume that the lenses around $J_1$ and the contour part $C\setminus J$ are contained within the set
$|z-1|<1/2$.
The following function $\mathcal M$ will approximate ${\widehat S}$ for $|z-1|>1/2$:
\begin{equation}
\mathcal M(z)=\begin{cases}
\begin{pmatrix}0&1\\-1&0\end{pmatrix}&|z|<1,\\
I&|z|>1.
\end{cases} \end{equation} 
For $|z-1|<1/2$, we construct the following function $Q$. Let 
\begin{equation}
w(z)=-i\frac{n}{2s}\log z,
\end{equation}
so that $w(e^{2i t\frac{s}{n}})=t$ for any $t$, and define
\begin{equation}\label{66}
Q(z)=\begin{cases}\Phi\left(w(z);s\right)z^{- \frac{n}{2}\sigma_3}\begin{pmatrix}
0&1\\-1&0
\end{pmatrix}&|z|<1,\\
\Phi\left(w(z);s\right)z^{- \frac{n}{2}\sigma_3}&|z|>1,
\end{cases}
\end{equation}
where $\Phi$ is the solution of the $\Phi$ RH problem at the beginning of the section.

Let
\begin{equation}
\widehat R(z)=\begin{cases}
{\widehat S}\mathcal M^{-1} & \textrm{for $|z-1|>1/2$,}\\
{\widehat S}Q^{-1} & \textrm{for $|z-1|<1/2.$}
\end{cases}
\end{equation} 
Then $\widehat R$ is analytic for $\mathbb C\setminus \Gamma_{\widehat R}$, where 
\begin{equation}\widehat \Gamma_R= \{ \textrm{the edge of the lens for $|z-1|>1/2$} \} \cup \{z:|z-1|=1/2\}.
\end{equation}
We have using \eqref{44}, \eqref{66}, 
\begin{equation}\nonumber \widehat R_+^{-1}(z)\widehat R_-(z)=Q(z)\mathcal M^{-1}(z)=I+\mathcal O(s/n)\end{equation}
 uniformly on the circle $|z-1|=1/2$ oriented counterclockwise.
Furthermore, $\widehat R_+^{-1}\widehat R_--I=\mathcal O(e^{-n\epsilon})$, $\epsilon>0$,
uniformly on the edges of the lenses. Thus, by standard small norm analysis (see, e.g., \cite{Dbook}),
\begin{equation}\widehat R(z)=I+\mathcal O(s/n),\qquad \widehat R'(z)=\mathcal O(s/n),\end{equation}
uniformly for $z\in \mathbb C$.

We now express $F(a_2)$ from Proposition \ref{Prop Diffid} in terms of elements of $\Phi$.
Tracing back the transformations, we see that as $z$ approaches $a_2$ from the {\it inside}
of the unit circle and being outside the lens,
\[
Y(z)=T(z)=\widehat S(z)=\widehat R(z)Q(z)=
\widehat R(z)\Phi(w(z))z^{-(n/2)\sigma_3}
\begin{pmatrix}
0&1\\-1&0
\end{pmatrix}.
\]
Using this, we obtain
\[
\begin{aligned}
&-z^{n+1}(Y(z)^{-1}Y'(z))_{21}=
z\left(\Phi(w(z))^{-1}\frac{d}{dz}\Phi(w(z))\right)_{12}
+z\left(\Phi^{-1}\mathcal O(s/n)\Phi\right)_{12}
\\&=
z\left(\Phi(w)^{-1}\frac{d}{dw}\Phi(w)\right)_{12}\frac{dw}{dz}+
z\left(\Phi^{-1}\mathcal O(s/n)\Phi\right)_{12}
\\&=
-\frac{in}{2s}\left(\Phi(w)^{-1}\frac{d}{dw}\Phi(w)\right)_{12}+
z\left(\Phi^{-1}\mathcal O(s/n)\Phi\right)_{12}.
\end{aligned}
\]
Taking the limit $z\to a_2=\exp(i\phi_2)=\exp(i2v_2 s/n)$ along this trajectory, 
we obtain
\begin{equation}
F(a_2)=-\frac{in}{2s}\left[\Phi_+^{-1}(v_2)\Phi_+'(v_2)\right]_{12}+\mathcal O(s/n),
\end{equation}
as $s,n\to \infty$ and $s/n\to 0$.
Substituting this into \eqref{diff id}, recalling (\ref{20}), and noting
that $dv_2/d\phi_2=n/(2s)$, proves \eqref{diffid new}.
The symmetric case identity \eqref{diffidsym} follows from \eqref{diffidsymT}.
Thus we finished the proof of Lemma \ref{diffFred}.

We now solve the RH problem for $\Phi$, compute the asymptotics of the r.h.s. of (\ref{diffid new}), integrate it,
and use Lemma \ref{Lemmasep} at one of the integration limits
to obtain Theorem \ref{Thm}.

\section{Solution of the RH problem for $\Phi$}\label{secPhi}
Recall the definition of $\psi(z)$ in \eqref{defpsi}, and for 
$z\in \mathbb C\setminus \overline{(-1,v_1)\cup (v_2,1)}$ on the first sheet of the Riemann surface $\Sigma$, let
\begin{equation}\label{def:phi}
\phi(z)=\int_{1}^z\psi(\xi)d\xi.
\end{equation}
We see by \eqref{int0}  that $\phi(z)$ is a well defined function, analytic on $\mathbb C\setminus \overline{(-1,v_1)\cup (v_2,1)}$. Since $\psi_+=-\psi_-$ on $(-1,v_1)\cup (v_2,1)$, we have
\begin{equation}\label{def:Omega2}
\phi_+(z)+\phi_-(z)=\begin{cases}0&\textrm{for }z\in(v_2,1),\\
-2\pi \Omega&\textrm{for }z\in(-1,v_1),
\end{cases}\qquad \Omega=\frac{1}{\pi}\int_{v_1}^{v_2}\psi(x)dx>0.
\end{equation}

Since by \eqref{int0} $\psi(z)$ has zero residue at infinity,
$\psi(z)=1+\mathcal O(1/z^2)$ as $z\to \infty$, and we have
\begin{equation}\label{asymphi}
\phi(z)=z+\mathcal O(1),\qquad z\to \infty.
\end{equation}
Let
\begin{equation}\label{def:hatPhi}
 \mathcal S(z)= e^{is\ell\sigma_3}\Phi(z)e^{-is\phi(z)\sigma_3},
\qquad \ell=\int_1^{\infty}(\psi(x)-1)dx-1, 
 \end{equation}
then $\mathcal S$ satisfies the following RH problem.
\subsubsection*{RH Problem for $\mathcal S$}
\begin{itemize}
\item[(a)] $\mathcal S$ is analytic for $z \in \mathbb C \setminus \Gamma_{ \Phi}$,
\item[(b)] $\mathcal S$ has jumps given by $\mathcal S_+(z)=\mathcal S_-(z)J_{\mathcal S}(z)$, where 
\begin{equation}
J_{\mathcal S}(z)=
\begin{cases}
\begin{pmatrix}
0&-1\\1&0
\end{pmatrix}&\textrm{for }z\in (v_2,1),\\
\begin{pmatrix}
0&-e^{-2\pi is\Omega}\\e^{2\pi is\Omega}&0
\end{pmatrix}&\textrm{for }z\in (-1,v_1),\\
\begin{pmatrix}
1&0\\e^{-2is\phi(z)}&1
\end{pmatrix}
&\textrm{for } z\in \Gamma_{ \Phi,\rm L},\\
\begin{pmatrix}
1&-e^{2is\phi(z)}\\0&1
\end{pmatrix}
&\textrm{for } z\in \Gamma_{ \Phi,\rm U}.
\end{cases}
\end{equation}
\item[(c)] As $z \to \infty$,
\begin{equation}
\mathcal S(z)=I+\mathcal O\left(\frac{1}{z}\right).
\end{equation}
\end{itemize}

We need the conditions $\Im \phi(z)<0$, $\Im \phi(z)>0$,
to hold uniformly on $\Gamma_{\Phi,\rm L}$, $\Gamma_{\Phi,\rm U}$, respectively,
away from some fixed $\epsilon$ neighborhoods of the end-points for the corresponding jumps to be exponentially close to the identity.
Since  \eqref{asymphi} is uniform as $|z|\to \infty$, the conditions hold for $|z|>W$ for some sufficiently large but fixed $W>0$. Since $\frac{d}{dx}\phi(x)=\psi(x)>0$ for $x\in \mathbb R \setminus \overline{(-1,v_1)\cup (v_2,1)}$, the conditions hold 
on the contour as stated assuming (and we do this) 
that the angle between the parts of $\Gamma_{\Phi,\rm L}$, $\Gamma_{\Phi,\rm U}$ emanating from $\pm 1$ and the real axis  
  was chosen to be sufficiently small and the lens around $(v_1,v_2)$ was sufficiently narrow. Therefore
\begin{equation}\label{smallJhat}
J_{\mathcal S}(z)=I+\mathcal O\left( e^{-cs(1+|z|)}\right),
\end{equation}
as $s\to \infty$, for some constant $c>0$, uniformly on 
$\Gamma_{\Phi,\rm L}$, $\Gamma_{\Phi,\rm U}$
away from fixed $\epsilon$-neighborhoods of $\pm 1$, $v_1$, $v_2$.

\subsection{Outside parametrix and $\theta$-functions}\label{Inter}
Consider the following RH problem for the $2\times2$-matrix valued function
$\mathcal N(z;\omega)$ with a real parameter $\omega$,
which will give an approximate solution to the $\Phi$ RH problem away from the edge points $\pm 1$, $v_1$, $v_2$, when $\omega=s\Omega$. Later on we also construct approximate solutions
(local parametrices) on a neighborhood of each edge point, and match them to the leading
order with $\mathcal N(z;\omega)$ on the boundaries of the neighborhoods.

\subsubsection*{RH problem for $\mathcal N$}
\begin{itemize}
\item[(a)] $\mathcal N(z)$ is analytic on $\mathbb C\setminus \overline{(-1,v_1)\cup (v_2,1)}$.
\item[(b)] On $(-1,v_1)\cup (v_2,1)$, $\mathcal N$ has $L^2$ boundary values related by 
the jump conditions:
\begin{equation}\nonumber
\begin{aligned}
\mathcal N_+(z)&=\mathcal N_-(z)\begin{pmatrix}
0&-1\\1&0
\end{pmatrix}&\textrm{for }z\in (v_2,1),\\
\mathcal N_+(z)&=\mathcal N_-(z)\begin{pmatrix}
0&-e^{-2\pi i\omega}\\e^{2\pi i\omega}&0
\end{pmatrix}&\textrm{for }z\in (-1,v_1).
\end{aligned}
\end{equation}
\item[(c)] As $z\to \infty$,
\begin{equation}
\mathcal N(z)=I+\mathcal O(z^{-1}).
\end{equation}
\end{itemize}
A more general problem with jumps on $m$ intervals was solved in \cite{DIZ} in terms of 
multidimensional $\theta$-functions. We now present the solution in our case of 2 intervals:
$(-1,v_1)$, $(v_2,1)$.
Let
\begin{equation}\label{def:gamma}
\gamma(z)=\left( \frac{(z-1)(z-v_1)}{(z-v_2)(z+1)}\right)^{1/4}, \end{equation}
also with branch cuts on $(-1,v_1)\cup (v_2,1)$, such that $\gamma(z) \to 1$ as $ z \to \infty$ on the first sheet
of the Riemann surface $\Sigma$.

Recall the definition of the holomorphic differential \eqref{def:omega}.
Let $u$ be the following analytic function on $\mathbb C\setminus \left\{(-\infty,v_1]\cup [v_2,+\infty)\right\}$:
\begin{equation}\label{def:u}
u(z)=-\int_{v_2}^z\bm{{\omega}},\end{equation}
with integration taken on the first sheet.
Note that, $\mod\mathbb Z$,
\begin{equation}\label{uz0}
u(-1)=-\frac{\tau}{2}-\frac{1}{2},\qquad u(v_1)=-\frac{\tau}{2},
\qquad u(v_2)=0,\qquad u(1)=-\frac{1}{2}.
\end{equation} 
The function $u(z)$ extends to the Riemann surface $\Sigma$ and is then called the Abel map. It maps the Riemann 
surface onto the torus where $\theta$-functions are defined.

A simple calculation (see \cite{DIZ}) shows that
the function $\gamma(z)-\gamma(z)^{-1}$ has a single zero on $(v_1,v_2)$ on the first sheet, 
denote it by $\hat z$, and no zeros on the second sheet. We have
\begin{equation}
\hat{z}=\frac{v_1+v_2}{2+v_1-v_2}.\end{equation}
Similarly, the function $\gamma(z)+\gamma(z)^{-1}$ has no zeros on the first sheet and one zero on the second. 

Let
\begin{equation}\label{def:d}
d=-\frac{1-\tau}{2}-\int_{v_2}^{\hat z}\bm{{\omega}},
\end{equation}
with integration taken on the first sheet. 

Consider the third Jacobian $\theta$-function $\theta(z)=\theta_3(z;\tau)$
(see Appendix \ref{App1}).
Since $\theta((1-\tau)/2)=0$, we have 
$\theta(u(\hat z)-d)=0$.
The function $\theta(u(z)-d)=0$ has no other zeros on the Riemann surface.
The function $\theta(u(z)+d)=0$ has only one zero on the Riemann surface located on the second sheet which coincides with the only zero of $\gamma(z)+\gamma(z)^{-1}$.

By an argument in \cite{DIZ} we have 
\begin{equation}\nonumber
u(\infty)+d =m\tau\mod\mathbb Z,
\end{equation}
for some integer $m$.
Consider the integral of $\bm{{\omega}}$ along the closed contour composed of a large interval along the real axis and a semicircle in the upper half-plane. Then using (\ref{omega cond}) and the definition of $\tau$ in (\ref{def:Omega}) we obtain in the case $v_1=-v_2$ that $u(\infty)+d =0\mod\mathbb Z$ with $u(z)$ considered on the first sheet.
Therefore also in the general case of $v_1$, $v_2$, by continuity,
\begin{equation} \label{uinftyd}
u(\infty)+d=0\mod\mathbb Z.
\end{equation}

The solution to the RH problem for $\mathcal N$ is given by 
\begin{equation}\begin{aligned}\label{def:calN}
\mathcal N(z;\omega)&=\begin{pmatrix}
\frac{\gamma+\gamma^{-1}}{2}m_{11}&
-\frac{\gamma-\gamma^{-1}}{2i}m_{12}\\
 \frac{\gamma-\gamma^{-1}}{2i}m_{21}&
\frac{\gamma+\gamma^{-1}}{2} m_{22}
\end{pmatrix},
\\m(z)&=\frac{\theta(0)}{\theta(\omega)}
\times \begin{pmatrix}
\frac{\theta(u(z)+\omega+d)}{\theta (u(z)+d)}&
\frac{\theta(u(z)-\omega-d)}{\theta (u(z)-d)}\\
\frac{\theta(u(z)+\omega-d)}{\theta (u(z)-d)}&
\frac{\theta(u(z)-\omega+d)}{\theta (u(z)+d)}
\end{pmatrix}
\end{aligned}\end{equation}
with $z$ on the first sheet.
To see that $\mathcal N$ solves the RH problem for $\mathcal N$, one makes several observations. First note that $\gamma(z)$ is analytic on $\mathbb C \setminus \overline{(-1,v_1)\cup (v_2,1)}$ and 
\[ \gamma_+(z)=i\gamma_-(z), \quad z \in (-1,v_1)\cup (v_2,1).\]
Hence for $w\in (-1,v_1)\cup (v_2,1)$
\begin{equation}\label{jumpsgamma} \begin{aligned}
\left(\frac{\gamma+\gamma^{-1}}{2}\right)_+&=-
\left(\frac{\gamma-\gamma^{-1}}{2i}\right)_-;\\
\left(\frac{\gamma-\gamma^{-1}}{2i}\right)_+&=
\left(\frac{\gamma+\gamma^{-1}}{2}\right)_-.
\end{aligned} \end{equation}
Secondly, as follows from \eqref{theta3 period} and the relations 
\begin{equation} \label{jumpsu}
u_+(z)=\begin{cases}-u_-(z)\mod \mathbb Z& z \in (v_2,1),\\
-u_-(z)-\tau \mod \mathbb Z &z\in (-1,v_1), \end{cases}
\end{equation}
the matrix $m$ has the jumps:
 \begin{equation} \label{jumpsm}
\begin{aligned}
m_+(z)&=m_-(z)\begin{pmatrix}
0&1\\1&0
\end{pmatrix} \quad \textrm{for $z \in (v_2,1)$,}\\
m_+(z)&=m_-(z)\begin{pmatrix}
0& e^{-2\pi i \omega}\\
e^{2\pi i \omega} &0
\end{pmatrix} \quad \textrm{for $z \in (-1,v_1)$.}
\end{aligned}
\end{equation}
The singularities of $m$ cancel with the zeros of $\gamma \pm \gamma^{-1}$. Furthermore, 
\[ \mathcal N(z)=I+\mathcal O (z^{-1}) \]
as $z \to \infty$.

\subsection{Identities for $\theta$-functions}\label{secthid}
Our proof of Theorem \ref{1} will use several identities satisfied by $\theta$-functions.
We present these identities now. Standard definitions and properties of theta-functions that we need are summarized in Appendix \ref{App1}.
\begin{Lemma}\label{ThmThetaids}
With the coefficients of the expansion $\gamma_0$, $u_0$, $\gamma_1$, $u_1$,  given in \eqref{expgamma} below
we have:
\begin{itemize}
\item[(a)] For any\footnote{If $d\pm\omega$ is a zero of $\theta_3$, we multiply through
in \eqref{id1} before evaluating. We adopt the same convention in other formulae below.
Furthermore, it is easily seen that $\theta_3(d)\neq 0$ and $\theta_3(\omega)\neq 0$ for any $\omega\in\mathbb R$.} $\omega \in \mathbb R$,
\begin{equation}\label{id1}
\frac{\theta_3(0)^2\theta_3(d+\omega)\theta_3(d-\omega)}{\theta_3(d)^2\theta_3(\omega)^2}\left(1-\frac{\gamma_0^2u_0}{2}\left[\frac{\theta_3'(d+\omega)}{\theta_3(d+\omega)}+\frac{\theta_3'(d-\omega)}{\theta_3(d-\omega)}-2\frac{\theta_3'(d)}{\theta_3(d)}\right]
\right)=1.
\end{equation}
\item[(b)] 
\begin{equation}\label{id2} 
\frac{\theta_1'(d)}{\theta_1(d)}-\frac{\theta_3'(d)}{\theta_3(d)}=\frac{1}{\gamma_0^2 u_0}=-iI_0 (1+v_2).
\end{equation}
\item[(c)] 
\begin{equation}\label{id3} 
\left(\frac{\theta_1(d)}{\theta_3(d)}\right)'''=\frac{3}{\gamma_0^2u_0}\left(\frac{\theta_1(d)}{\theta_3(d)}\right)''-\frac{6(2\gamma_1+u_1)}{\gamma_0^2u_0^3}\left(\frac{\theta_1(d)}{\theta_3(d)}\right).
\end{equation}
\item[(d)] For $z_0\in\{-1,v_1,v_2,1\}$, 
\begin{equation}\label{idd} 
\frac{\theta_1^2(u(z_0)+d)}{\theta_3^2(u(z_0)+d)} \left( \frac{\theta_3}{\theta_1'}\right)^2 h(z_0)=-\frac{1}{I_0^2},
\qquad h(z)=(z-1)(z-v_1)+(z-v_2)(z+1).
\end{equation}
\item[(e)]
\begin{equation}\label{idth4}
\theta_4(0;\tau)^4=\theta_4^4=\frac{I_0^2}{\pi^2}2(v_2-v_1).
\end{equation}
\item[(f)]
\begin{equation}\label{idth2} 
\theta_2(0;\tau)^4=\theta_2^4=\frac{I_0^2}{\pi^2}(1+v_1)(1-v_2).
\end{equation}
\item[(g)]
\begin{equation}\label{idth3} 
\theta_3(0;\tau)^4=\theta_3^4=\frac{I_0^2}{\pi^2}(1-v_1)(1+v_2).
\end{equation}
\end{itemize}
\end{Lemma}

\begin{proof}
We begin by proving (a). By \eqref{jumpsgamma} and \eqref{jumpsm}, it follows that $\eta_1(z)$ defined by
Consider $\eta_1(z)$ defined by
\begin{multline}\nonumber 1\equiv \det \mathcal{N}(z)=\eta_1(z)=
\left(\frac{\gamma(z)+\gamma^{-1}(z)}{2}\right)^2\frac{\theta_3^2\theta_3(u(z)+\omega+d)\theta_3(u(z)-\omega+d)}{\theta_3(\omega)^2\theta_3(u(z)+d)^2}\\-
\left(\frac{\gamma(z)-\gamma^{-1}(z)}{2}\right)^2\frac{\theta_3^2\theta_3(-u(z)+\omega+d)\theta_3(-u(z)-\omega+d)}{\theta_3(\omega)^2\theta_3(-u(z)+d)^2}
\end{multline}

Of course, we know that $\eta_1(z)=\det\mathcal{N}(z)=1$ for all $z$ from the Riemann-Hilbert problem. However, it is easy to provide a direct proof: By \eqref{jumpsgamma} and \eqref{jumpsm}, and the fact that $\frac{\theta_3(\xi+\omega+d)\theta_3(\xi-\omega+d)}{\theta_3(\xi+d)^2}$ is an elliptic function of $\xi$,
the function $\eta_1(z)$ no jumps on $A$ 
and is thus meromorphic. By the fact that $\theta_3(\pm u(z)+d)$ has the same zeros as $\gamma(z)\pm \gamma(z)^{-1}$, respectively, it follows that $\eta_1(z)$ has no singularities and is an entire function. By \eqref{uinftyd},  $\eta(z)\to 1$ as $z\to \infty$, and thus  
$\eta_1(z)=1$ for all $z\in\mathbb C$ by Liouville's theorem. 

The expansion of $\eta_1(z)$ as $z\to v_2$ (using \eqref{expgamma} below) shows that
\begin{equation}
\eta_1(z)\to \frac{\theta_3^2\theta_3(d+\omega)\theta_3(d-\omega)}{\theta_3(d)^2\theta_3(\omega)^2}\left(1-\frac{\gamma_0^2u_0}{2}\left[\frac{\theta_3'(d+\omega)}{\theta_3(d+\omega)}+\frac{\theta_3'(d-\omega)}{\theta_3(d-\omega)}-2\frac{\theta_3'(d)}{\theta_3(d)}\right]
\right),
\end{equation} 
and since $\eta_1(v_2)=1$, we obtain Part (a) of the proposition.

Now consider
\begin{equation}\label{Smartfun}
\eta_2(z)=\left(\frac{\gamma(z)+\gamma^{-1}(z)}{2}\right)^2\frac{\theta_1^2(u(z)+d)}{\theta_3^2(u(z)+d)}-
\left(\frac{\gamma(z)-\gamma^{-1}(z)}{2}\right)^2\frac{\theta_1^2(-u(z)+d)}{\theta_3^2(-u(z)+d)}.
\end{equation}
By the fact that $\frac{\theta_1^2(\xi)}{\theta_3^2(\xi)}$ is an elliptic function of $\xi$, and by \eqref{jumpsgamma} and \eqref{jumpsu}, it follows that $\eta_2(z)$ is a meromorphic function, and again by cancellation of the poles from $\theta_3(\pm u(z)+d)$ by the  zeros of $\gamma(z)\pm \gamma(z)^{-1}$, it follows that $\eta_2(z)$ in fact is entire. As $z\to \infty$, $\eta_2(z)\to 0$ by \eqref{uinftyd}
since $\theta_1(0)=0$, and thus, $\eta_2(z)\equiv 0$ by Liouville's theorem. We see from the expansion of $\eta_2(z)$ in powers of $z-v_2$
as $z\to v_2$ that
\begin{equation}
\eta_2(z)\to \frac{\theta_1(d)}{\theta_3(d)}\left[\frac{\theta_1(d)}{\theta_3(d)}-\gamma_0^2u_0\left(\frac{\theta_1(d)}{\theta_3(d)}\right)'\, \right],\qquad z\to v_2.
\end{equation}
Since this limit is zero, we obtain that
\begin{equation}\label{id20}
\left(\frac{\theta_1(d)}{\theta_3(d)}\right)'=\frac{1}{\gamma_0^2 u_0}\frac{\theta_1(d)}{\theta_3(d)},
\end{equation}
which gives Part (b) of the proposition. 

To prove part (c), we consider the coefficient of the first power $z-v_2$ in the expansion of $\eta_2(z)$ as $z\to v_2$. Denote here $g(z)=\frac{\theta_1(z)}{\theta_3(z)}$, then as $z\to v_2$, 
\begin{multline}\label{expeta2}
0=\eta_2(z)-\eta_2(v_2)=4(z-v_2)\Bigg[
-\gamma_0^2\frac{u_0^3}{6}\left(g'''(d)g(d)+3g''(d)g'(d)\right)\\-u_0g'(d)g(d)\left(u_1\gamma_0^2+\gamma_0^{-2}+2\gamma_1\gamma_0^2\right)+u_0^2\left(g''(d)g(d)+g'(d)^2\right)
\Bigg]+\mathcal O((z-v_2)^2).
\end{multline}
By substituting the identity for $g'(d)$ from Part (b) of the proposition into the right hand side of \eqref{expeta2} and setting the resulting coefficient of $z-v_2$ equal to zero, we obtain Part (c).

Finally, to prove Part (d), we consider
\begin{equation}\label{eta3}
\eta_3(z)=R(z)\left[\left(\frac{\gamma(z)+\gamma^{-1}(z)}{2}\right)^2\frac{\theta_1^2(u(z)+d)}{\theta_3^2(u(z)+d)}+
\left(\frac{\gamma(z)-\gamma^{-1}(z)}{2}\right)^2\frac{\theta_1^2(-u(z)+d)}{\theta_3^2(-u(z)+d)}\right].
\end{equation}
By the same arguments as for $\eta_1$ and $\eta_2$ (and in addition by the fact that $R_+=-R_-$ on $A$), it follows that $\eta_3$ is entire. By  recalling the definition of $u$ in \eqref{def:u}, by \eqref{uinftyd}, and by the definition of $\gamma$ in \eqref{def:gamma},   we obtain
\begin{equation} \label{asymeta3}
\lim_{z\to\infty}\eta_3(z)=-\frac{1}{4I_0^2}\left(\frac{\theta_1'}{\theta_3}\right)^2+\frac{(2+v_1-v_2)^2}{16}\left(\frac{\theta_1(2d)}{\theta_3(2d)}\right)^2,
\end{equation}
so that $\eta_3(z)$ is identically equal to this constant.
Now consider the asymptotics of $\eta_2(z)$ as $z\to \infty$. We have
\begin{equation}
0\equiv\eta_2(z)=-z^{-2}\left[\frac{1}{4I_0^2}\left(\frac{\theta_1'}{\theta_3}\right)^2+\frac{(2+v_1-v_2)^2}{16}\left(\frac{\theta_1(2d)}{\theta_3(2d)}\right)^2\,\right]+\mathcal O(z^{-3}),
\end{equation}
from which we conclude that
\begin{equation}
-\frac{1}{4I_0^2}\left(\frac{\theta_1'}{\theta_3}\right)^2=\frac{(2+v_1-v_2)^2}{16}\left(\frac{\theta_1(2d)}{\theta_3(2d)}\right)^2.\end{equation}
By substituting this into \eqref{asymeta3}, we obtain
\begin{equation}\label{eta3id}
 \eta_3(z)=-\frac{1}{2I_0^2}\left(\frac{\theta_1'}{\theta_3}\right)^2
\end{equation}
for all $z\in \mathbb C$. On the other hand,
for $z_0\in \{-1,v_1,v_2,1\}$, from (\ref{eta3}) by \eqref{uz0}
and ellipticity, 
\begin{equation}
\eta_3(z_0)=\frac{1}{2}\frac{\theta_1^2(u(z_0)+d)}{\theta_3^2(u(z_0)+d)}h(z_0). 
\end{equation}
Equating this to (\ref{eta3id}) we obtain Part (d).
To show Parts (e), (f), (g), we consider the function (as usual, theta functions written without argument stand for 
their values with argument zero)
\begin{equation}\label{newnu}
\frac{\theta_3(u(z))^2}{\theta_1(u(z))^2}+I_0^2(v_2-v_1)(v_2^2-1)\frac{\theta_3^2}{\theta_1'^2}\frac{1}{z-v_2}.
\end{equation}
As before, we see that this function is identically constant. By evaluating at infinity, it is equal to $\frac{\theta_3(d)^2}{\theta_1(d)^2}$.
On the other hand, part (d) at $z_0=v_2$ gives
\begin{equation}
\frac{\theta_3(d)^2}{\theta_1(d)^2}=-I_0^2(v_2-v_1)(v_2-1)\frac{\theta_3^2}{\theta_1'^2}.
\end{equation}
Equating this constant to \eqref{newnu} we obtain the identity for all $z$:
\begin{equation}\label{newnufinal}
\frac{\theta_3(u(z))^2}{\theta_1(u(z))^2}=-I_0^2(v_2-v_1)(v_2-1)\frac{\theta_3^2}{\theta_1'^2}\frac{z+1}{z-v_2}.
\end{equation}
Evaluating it at $z=1$ (recall from \eqref{uz0} that $u(1)=1/2\mod\mathbb Z$ and recall the definition of $\theta_j(z)$ from Appendix \ref{App1}),
and using the identity $\theta_1'=\pi\theta_2\theta_3\theta_4$, we obtain Part (e). We similarly obtain Part (f) by evaluating
\eqref{newnufinal} at $z=v_1$. Finally, we obtain Part (g)
by making use of the identity $\theta_3^4=\theta_2^4+\theta_4^4$.

\end{proof}

\subsection{Local parametrices} 
Our goal in this section is to construct a function $P$ on a neighborhood of each point of the set $\mathcal T=\{-1,v_1,v_2,1\}$, with the same jump conditions as $\mathcal S$ on these neighborhoods, and with an asymptotic behavior matching that of $\mathcal N$ to the main order on the boundaries of these neighborhoods.
The first step is to recall the following model RH problem from \cite{KMVV} with an explicit solution in terms of Bessel functions.
\subsubsection*{RH problem for $\Psi$}
\begin{itemize}
\item[(a)] $\Psi: \mathbb C \setminus \Gamma_{\Psi} \to \mathbb C ^{2\times 2}$ is analytic, where $\Gamma_{\Psi} =\mathbb R^- \cup \Gamma^\pm_{\Psi}$, with $\Gamma^\pm_{\Psi}=\{ xe^{\pm \frac{ 2\pi}{3}i}: x\in \mathbb R ^+\}$, and with orientation of $\mathbb R^-$, $\Gamma^\pm_{\Psi}$ towards zero.
\item[(b)] $\Psi$ satisfies the jump conditions:
\begin{align}\nonumber
\Psi_+(\zeta)&=\Psi_-(\zeta)\begin{pmatrix}0&1\\-1&0\end{pmatrix} \quad \textrm{for $\zeta \in \mathbb R^-$,}\\
\Psi_+(\zeta)&=\Psi_-(\zeta)\begin{pmatrix}1&0\\1&1\end{pmatrix}\quad \textrm{for $\zeta \in \Gamma^\pm_{\Psi}$.}\nonumber
\end{align}
\item[(c)] As $\zeta \to \infty$, 
\begin{equation}\nonumber
\Psi (\zeta)=\left(\pi \zeta^{\frac{1}{2}}\right)^{-\frac{\sigma_3}{2}} \frac{1}{\sqrt{2}} \begin{pmatrix} 1&i \\i &1\end{pmatrix} \left( I +\frac{1}{8\sqrt \zeta} \begin{pmatrix} -1&-2i \\-2i &1\end{pmatrix}
-\frac{3}{2^7\zeta} \begin{pmatrix} 1&-4i \\ 4i &1\end{pmatrix}
+\mathcal O\left( \zeta^{-\frac{3}{2}}\right)
\right)e^{\zeta^{\frac{1}{2}}\sigma_3}.  \end{equation}
\item[(d)] As $\zeta \to 0$, $
\Psi(\zeta)=\mathcal O(\log |\zeta|) .$
\end{itemize}

 For $|\arg \zeta| <2\pi/3$, we have
\begin{equation}\label{def:Psiout}
\Psi(\zeta)=\begin{pmatrix} I_0(\zeta^{1/2})& \frac{i}{\pi}K_0(\zeta^{1/2})\\
\pi i \zeta^{1/2} I'_0(\zeta^{1/2})&-\zeta^{1/2}K_0'(\zeta^{1/2})\end{pmatrix},
\end{equation}
where $I_0$ and $K_0$ are Bessel functions, $I'_0(x)=\frac{d}{dx}I_0(x)$,
$K'_0(x)=\frac{d}{dx}K_0(x)$.
For definitions and properties of Bessel functions see, e.g, \cite{GR}. Here the principal branch of $\zeta^{1/2}$
with the cut along the negative real axis is chosen.
For the explicit expression of the solution in the sector $|\arg \zeta|>2/3$, see \cite{KMVV}.

We have the following useful asymptotics as $z\to 0$ for $I_0$:
\begin{equation}
I_0(z)=1+\frac{z^2}{4}+\frac{z^4}{64} +\mathcal O(z^6)\label{asymI0}.
\end{equation}

We denote by $U^{(p)}$ fixed open nonintersecting balls containing $p\in\mathcal T=\{-1,v_1,v_2,1\}$. Recalling $\psi$ in \eqref{defpsi}, we define $\zeta=\zeta^{(p)}$ on $U^{(p)}$ by
\begin{equation}\label{def:zeta}
\zeta^{(p)}(z) =-\left( s\int_p^z\psi(\xi)d\xi \right)^2. \end{equation}
As $z\to p$, we have the expansion
\begin{equation}\label{expzeta2}
\zeta^{(p)}(z) =(z-p)s^2\widetilde{\zeta}_0(1+o(1)),\qquad 
\widetilde{\zeta}_0=-\frac{4(p-x_1)^2(p-x_2)^2}{\prod_{q\in\mathcal T\setminus\{p\}}(p-q)}.
\end{equation}
Note that $\zeta^{(p)}(z)$ is a conformal mapping of $U^{(p)}$ onto a neighborhood of zero.
Observe also that $\widetilde{\zeta}_0>0$ for $p=v_2,-1$, and 
$\widetilde{\zeta}_0<0$ for $p=v_1,1$, and so the contours in $U^{(p)}$ are mapped from the $z$-plane to the 
$\zeta$-plane accordingly. We choose the exact form of the contours in the $z$-plane so that their images are direct lines.

Keeping in mind our conventions for the root branches, we obtain
\begin{equation}\label{branches}
\begin{aligned}
(\zeta^{(p)}(z))^{1/2}+i(\phi(z)-\phi(p))=0,&\qquad \Im z>0\\
(\zeta^{(p)}(z))^{1/2}-i(\phi(z)-\phi(p))=0,&\qquad \Im z<0\\
\end{aligned}
\end{equation}

By \eqref{int0}, \eqref{def:phi} and the definition of $\Omega$ in \eqref{def:Omega},
\begin{equation}\label{pointsphi}
\phi(v_2)=\phi(1)=0,\qquad \phi(-1)=\phi(v_1)=-\pi \Omega.\end{equation}

Let
\begin{equation}\label{defnhatPsi}
X(z)=\begin{cases} \begin{pmatrix}0&-1\\1&0\end{pmatrix} &\textrm{for }\Im z>0,\\
I&\textrm{for }\Im z<0.\end{cases}\end{equation}
 For $p=-1,v_2$, we define the local parametrix on $U^{(p)}$ by
\begin{equation}\label{PE}
\begin{aligned}
P(z)&=E(z)\Psi\left(\zeta(z)\right)X(z)e^{-is\phi(z)\sigma_3},\\
E(z)&=\mathcal N(z;s\Omega)e^{ i s \phi(p)\sigma_3}X(z)^{-1}\frac{1}{\sqrt{2}}\begin{pmatrix}1&-i\\-i&1\end{pmatrix}\left(\pi\zeta^{\frac{1}{2}}\right)^{\frac{1}{2}\sigma_3},
\end{aligned}
\end{equation}
where we have suppressed the superscript in $\zeta=\zeta^{(p)}$, and the branch cut for $\zeta^{1/4}$ is the same one as for $\zeta^{1/2}$.

Using the jump conditions, it is straightforward to verify that $E(z)$ has no jumps
in $U^{(p)}$, and since its singularity at $p$ is removable, $E(z)$ is analytic in the neighborhood $U^{(p)}$, $p=-1,v_2$.

Furthermore, it is easy to verify that $P(z)$ satisfies the same jump conditions as 
$\mathcal S(z)$ in $U^{(p)}$, $p=-1,v_2$.

Finally, using the condition (c) in the $\Psi$-RHP and (\ref{branches}), we obtain for $w\in\partial U^{(p)}$ 
\begin{equation}\label{idPN}
P(z)\mathcal N(z;s\Omega)^{-1}=I+\Delta_1(z)+\mathcal O(1/s^2),
\qquad \Delta_1(z)=\mathcal O(1/s),
\end{equation}
uniformly on the boundary as $s\to \infty$, where 
\begin{equation}\label{Delta1}\begin{aligned}
\Delta_1(z)&\equiv\Delta_1(z;s\Omega);\\
\Delta_1(z;\omega)
&=\frac{\mp 1}{8\sqrt{\zeta(z)}}\mathcal N(z;\omega)e^{is\phi(p)\sigma_3}\begin{pmatrix}
-1&- 2i\\ - 2i &1
\end{pmatrix}e^{-is\phi(p)\sigma_3}\mathcal N^{-1}(z;\omega),\qquad p=-1,v_2,
\end{aligned}
\end{equation}
where $\mp$ is taken to be $-$ on $U^{(p)}\cap\mathbb C_+$, and
 $+$ on $U^{(p)}\cap\mathbb C_-$.
Note that $\Delta_1(z)$ is meromorphic in $U^{(p)}$, $p=-1,v_2$, with the first order pole 
at $z=p$.

Similarly, for $p=v_1,1$, we
define the local parametrix on $U^{(p)}$ by
\begin{equation}\begin{aligned}
P(z)&=E(z)\sigma_3\Psi\left(\zeta(z)\right)\sigma_3 X(z) e^{-is\phi(z)\sigma_3},\\
E(z)&=\mathcal N(z;s\Omega)e^{ i s \phi(p)\sigma_3}X(z)^{-1}\frac{1}{\sqrt{2}}\begin{pmatrix}1&i\\i&1\end{pmatrix}\left(\pi\zeta^{\frac{1}{2}}\right)^{\frac{1}{2}\sigma_3}.
\end{aligned}
\end{equation}

Here $E(z)$ is analytic on $U^{(p)}$, $P(z)$ has the same jumps as $\mathcal S(z)$ in
$U^{(p)}$, $p=v_1,1$, and the same condition (\ref{idPN}) holds with
\begin{equation}\label{Delta11}\begin{aligned}
\Delta_1(z)&\equiv\Delta_1(z;s\Omega);\\
\Delta_1(z;\omega)
&=\frac{\mp 1}{8\sqrt{\zeta(z)}}\mathcal N(z;\omega)e^{is\phi(p)\sigma_3} \begin{pmatrix}
-1& 2i\\  2i &1
\end{pmatrix}e^{-is\phi(p)\sigma_3}\mathcal N^{-1}(z;\omega),\qquad p=v_1,1,
\end{aligned}
\end{equation}
where $\mp$ is taken to be $-$ on $U^{(p)}\cap\mathbb C_+$, and
 $+$ on $U^{(p)}\cap\mathbb C_-$.
As at $v_2,-1$, $\Delta_1(z)$ in (\ref{Delta11}) is meromorphic in $U^{(p)}$, $p=v_1,1$, with the first order pole at $z=p$.

\subsection{Small norm RH problem. Solution of the $\Phi$-RH problem for fixed $v_1$, $v_2$}\label{SecR}
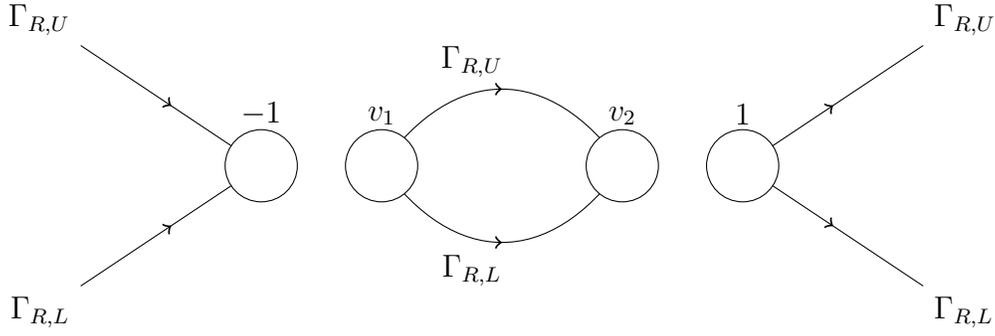
\begin{figure}\begin{center}
\begin{tikzpicture}[scale=0.8]
\draw [decoration={markings, mark=at position 0.5 with {\arrow[thick]{>}}},
        postaction={decorate}] (0,2) -- (3,0); 
\draw [decoration={markings, mark=at position 0.5 with {\arrow[thick]{>}}},
        postaction={decorate}] (11,0) -- (14,2); 
\draw [decoration={markings, mark=at position 0.5 with {\arrow[thick]{>}}},
        postaction={decorate}] (11,0) -- (14,-2); 
\draw  [decoration={markings, mark=at position 0.5 with {\arrow[thick]{>}}},
        postaction={decorate}] (5,0) .. controls (6.2,-1.7) and (7.8,-1.7) .. (9,0);     
\draw  [decoration={markings, mark=at position 0.5 with {\arrow[thick]{>}}},
        postaction={decorate}] (0,-2) -- (3,0);
\draw  [decoration={markings, mark=at position 0.5 with {\arrow[thick]{>}}},
        postaction={decorate}] (5,0) .. controls (6.2,1.7) and (7.8,1.7) .. (9,0);  
\node [above left] at (0,2) {$\Gamma_{R,U}$};
\node [above right] at (14,2) {$\Gamma_{R,U}$};
\node [above ] at (6.5,1.3) {$\Gamma_{R,U}$};
\node [below left] at (0,-2) {$\Gamma_{R,L}$};
\node [below right] at (14,-2) {$\Gamma_{R,L}$};
\node [below ] at (6.5,-1.3) {$\Gamma_{R,L}$};
\node [above ] at (9,0.5) {$v_2$};
\node [above ] at (11,0.5) {$1$};
\node [above ] at (5,0.5) {$v_1$}; 
\draw[black,fill=white] (9,0) circle [radius=0.6];
\draw[black,fill=white] (3,0) circle [radius=0.6]; 
\draw[black,fill=white] (5,0) circle [radius=0.6];   
\draw[black,fill=white] (11,0) circle [radius=0.6];  
\node [above ] at (3,0.5) {$-1$};
%\node [left] at (-1,0) {$\mathbb R$};    
\end{tikzpicture}
\caption{The jump contour $\Gamma_{R}$}\label{ContR}
\end{center}
\end{figure}
Let 
\begin{equation}\label{defnR}
R(z)=\begin{cases}
\mathcal S(z)\mathcal N^{-1}(z;s\Omega) & \textrm{for $z \in \mathbb C \setminus \left(\cup_{p\in \mathcal T} U^{(p)} \right)$,}\\
\mathcal S (z)P^{-1}(z)& \textrm{for $z \in \cup_{p\in \mathcal T} U^{(p)} $.}
\end{cases}
\end{equation}
Then $R(z)$ is analytic for $z\in \mathbb C\setminus \Gamma_R$, where $\Gamma_R$ is as in Figure \ref{ContR}. 
We have
\begin{equation}
R_+(z)=R_-(z) J_R(z),\qquad J_R(z)=\begin{cases} P(z)\mathcal N^{-1}(z)& \textrm{for $z \in \cup_{p\in \mathcal T} 
\partial U^{(p)}$,}\\
\mathcal N(z)J_S(z)\mathcal N^{-1}(z)&  \textrm{for $z \in \Gamma_R \setminus \left(\cup_{p\in \mathcal T} \partial U^{(p)} \right)$.}
\end{cases}
\end{equation}

By \eqref{smallJhat} and \eqref{idPN}, it follows that
\begin{equation}
J_R(z)=I+\mathcal O(s^{-1}/(|z^2|+1))),
\end{equation}
 as $s\to \infty$, uniformly for $z\in \Gamma_R$, and by the definition of $\mathcal S$ and $\mathcal N$, we have
\begin{equation}
R(z)=I+\mathcal O(z^{-1}),
\end{equation}
as $z\to \infty$. By standard small norm analysis, it follows that there is a solution to the RH problem for $R$ for $s$ sufficiently large, and that
\begin{equation}
R(z)=I+\mathcal O(1/s),\end{equation}
uniformly on compact sets 
as $s\to \infty$. In particular, the existence of $R$ also implies the existence of a solution to the RH problem for $\mathcal S$ for sufficiently large $s$, which also implies the existence of a solution to the RH problem for $\Phi$.
As usual, we expand $R$ in the powers of the small parameter, $1/s$ in our case, to write
\begin{equation}\label{expandR}
R(z)=I+R_1(z)+\mathcal O(1/s^2),
\end{equation}
where $R_1$ solves the following RH problem. 
$R_1(z)$ is analytic outside the \textit{clockwise} oriented boundaries 
$\partial U^{(p)}$ of the neighborhoods $U^{(p)}$,
\[
R_{1+}(z)=R_{1-}(z)+\Delta_1(z),\qquad z\in\cup_{p\in \mathcal T}\partial U^{(p)},
\] 
and $R_1(z)\to 0$ as $z\to\infty$. The solution to this problem is given by
\begin{equation}\label{R1}
R_1(z)=\frac{1}{2\pi i}\int_{\cup_{p\in \mathcal T}\partial U^{(p)}}
\frac{\Delta_1(x;s\Omega)}{x-z}dx,\qquad z\in\mathbb C\setminus \cup_{p\in \mathcal T}\partial U^{(p)},
\end{equation}
where the integrals are taken with \textit{clockwise} orientation.

Taking \eqref{expandR} (one can obtain further terms in that expansion in a standard way) and
tracing back the transformations $R\rightarrow\mathcal S\rightarrow \Phi$, 
we obtain an asymptotic solution of the $\Phi$-RH problem.

\subsection{Extension of the solution to the regimes $v_2\to 1$, $s(1-v_2)\to \infty$;
$v_1\to -1$, $s(1+v_1)\to \infty$}\label{secextend}

In our solution of the previous section, the end-points $-1<v_1<v_2<1$ were fixed. In this section, we show that 
the solution can be extended to the regime where $v_2$ not only can 
be fixed but can also approach $1$ (and $v_1$ approach $-1$) sufficiently slowly as $s\to\infty$.
This will be needed for the proof of Theorem \ref{Thm} below.

More precisely,  we fix $\epsilon>0$ and assume 
\begin{equation}\label{partcass}
1-v_2\le 1+v_1,\qquad v_2-v_1\ge\epsilon,\qquad  s(1-v_2)\to \infty.
\end{equation}

We let $U^{(v_2)}$ and $U^{(1)}$ have radius equal to  $c(1-v_2)$, and similarly $U^{(v_1)}$ and $U^{(-1)}$ have radius equal to $c(1+v_1)$, for some fixed and sufficiently small $c>0$. Note that the neighborhoods can now contract with growing $s$.

 As $v_2\to 1$, $I_j\to \frac{\pi}{\sqrt{2(1-v_1)}}$, for $j=0,1,2$, and computing an additional term in the expansion we find by \eqref{x1x2eqn2} that
\begin{equation}\nonumber
x_1x_2=\frac{v_1-v_2}{2}-\frac{(1-v_2)(1+v_1)}{4}+\mathcal O\left((1-v_2)^2\right),\qquad v_2\to 1
\end{equation} 
uniformly in the regime \eqref{partcass}.  By \eqref{x1x2eqn1}, 
\begin{equation}\label{limx1x2}
x_1=\frac{v_1-1}{2}+\mathcal O\left((1-v_2)^2\right), \qquad x_2=\frac{1+v_2}{2}+\mathcal O\left((1-v_2)^2\right).
\end{equation}

By \eqref{limx1x2} and \eqref{def:zeta}, 
\begin{equation}
\label{asymzetav21}
%|\phi(z)|=\left|\int_{v_2}^z\psi(x)dx\right|=\mathcal O(1-v_2),\qquad
\frac{1}{\sqrt{\zeta(z)}}=\mathcal O\left(\frac{1}{s(1-v_2)}\right),
\end{equation}  
uniformly in the regime \eqref{partcass} and also uniformly for $z\in \partial U^{(p)}$, $p\in \mathcal T=\{-1,v_1,v_2,1\}
$.

Next we will show that $\mathcal N$ and $\mathcal N^{-1}$ are uniformly bounded on $\partial U^{(p)}$ for $p\in \mathcal T$. 
As $v_2\to 1$ (under conditions \eqref{partcass}),  we see from \eqref{def:gamma} that
%\begin{equation}\begin{aligned}
%\gamma(z)&=\left(\frac{(1-v_1)(z-1)}{2(z-v_2)}\right)^{1/4}+\mathcal O(1-v_2),&&z\in \partial U^{(v_2)}\cup U^{(1)},\\
%\gamma(z)&=\left(\frac{z-v_1}{z+1}\right)^{1/4}+\mathcal O(1-v_2),&&z\in \partial U^{(v_1)}\cup U^{(-1)},
%\end{aligned}
%\end{equation}
% and in particular 
both $\gamma(z)$ and $\gamma^{-1}(z)$ are uniformly bounded also on $\partial U^{(p)}$ for $p\in \mathcal T$.
 
 We now consider $\theta$-functions, and start with $\tau$.
For $J_0$, we have
\begin{equation}
J_0=\int_{v_1}^{v_2-\sqrt{1-v_2}}\frac{dx\left(1+\mathcal O \left(\sqrt{1-v_2}\right)\right)}{(1-x)\sqrt{(x+1)(x-v_1)}}+\int_{v_2-\sqrt{1-v_2}}^{v_2}\frac{dx\left(1+\mathcal O \left(\sqrt{1-v_2}\right)\right)}{\sqrt{2(1-v_1)(1-x)(v_2-x)}},
\end{equation}
as $v_2\to 1$, and since
\begin{equation} \label{simpleint}
\frac{d}{dz}\log \left(\frac{\sqrt{z^2-1}+(it+\sqrt{1-t^2})z+i}{\sqrt{z^2-1}+\left(it-\sqrt{1-t^2}\right)z+i}\right)=\frac{\sqrt{1-t^2}}{\left(zt+1\right)\sqrt{z^2-1}},
\end{equation}
for any parameter $t$,
it follows that
\begin{equation}
J_0
=\frac{1}{\sqrt{2(1-v_1)}}\left[5\log 2+\log(1-v_2)^{-1}+\log \frac{1-v_1}{1+v_1}\right]\left(1+\mathcal O\left(\sqrt{1-v_2}\right)\right).
\end{equation}
Thus, since $I_0=\frac{\pi}{\sqrt{2(1-v_1)}}(1+\mathcal O(1-v_2))$,
\begin{equation}\label{limtau}
 \tau=i\frac{J_0}{I_0}=
 \frac{i}{\pi}\left[5\log 2+\log\frac{1}{1-v_2}+\log \frac{1-v_1}{1+v_1}\right]\left(1+\mathcal O\left(\sqrt{1-v_2}\right)\right),\qquad v_2\to 1
 \end{equation}
in the regime \eqref{partcass}, so that we have $-i\tau \to +\infty$. 

As $-i\tau \to +\infty$, $\frac{\theta_3}{\theta_3(\omega)}\to 1$ for any $\omega \in \mathbb R$.
We also observe that as $-i\tau \to +\infty$,
the fraction
\begin{equation}\label{fracthetav21}
\frac{\theta(\xi+\omega;\tau)}{\theta(\xi;\tau)}\end{equation}
is bounded uniformly under conditions \eqref{partcass} and over all $\omega\in [0,1)$,
for $\xi$ bounded away from the zero of the $\theta$-function 
$\frac{1+\tau}{2}$ modulo the lattice, and the same holds for derivatives of \eqref{fracthetav21} with respect to $\xi$, $\omega,$ and $\tau$. We now show that $\xi=u(z)\pm d$ remains bounded away from $\frac{1+\tau}{2}$ modulo the lattice for $z\in \partial U^{(p)}$, $p\in \mathcal T$.

We have by \eqref{uinftyd}, \eqref{def:u}, \eqref{uz0},
\begin{equation}\label{uinftyd1}
d=-u(\infty)=-\tau/2+1/2+\frac{i}{2I_0}\int_{-\infty}^{-1}\frac{dx}{\sqrt{p(x)}}\mod \mathbb Z. \end{equation}
As $v_2\to 1$,
\begin{equation} \int_{-\infty}^{-1}\frac{dx}{\sqrt{|p(x)|}}=\int_{-\infty}^{-1}\frac{dx(1+\mathcal O(1-v_2))}{(1-x)\sqrt{(-1-x)(v_1-x)}},
\end{equation}
and by using \eqref{simpleint}
\begin{equation}\label{uinftyint}
\frac{1}{2I_0}\int_{-\infty}^{-1}\frac{dx}{\sqrt{p(x)}}=
\frac{\sqrt{2(1-v_1)}}{(3-v_1)\pi\sqrt{1-\left(\frac{1+v_1}{3-v_1}\right)^2}}\log \left(\frac{1+\frac{1+v_1}{3-v_1}+\sqrt{1-\left(\frac{1+v_1}{3-v_1}\right)^2}}{1+\frac{1+v_1}{3-v_1}-\sqrt{1-\left(\frac{1+v_1}{3-v_1}\right)^2}}\right)(1+\mathcal O(1-v_2)),
\end{equation}
as $v_2\to 1$ in the regime \eqref{partcass}.
 We also have in the same regime by the definition \eqref{def:u} of $u(z)$, 
\begin{equation}\label{uint}
\begin{aligned}
u(z)&=-\frac{i}{2\pi} \int_{v_2}^z\frac{dz}{((z-v_2)(z-1))^{1/2}}\left(1+\mathcal O(1-v_2)\right), && z\in \partial U^{(v_2)}\cup \partial U^{(1)},\\
u(z)&=-\frac{\tau}{2}-\frac{i\sqrt{1-v_1}}{\sqrt{2}\pi} \int_{v_1}^z\frac{dz}{((z+1)(z-v_1))^{1/2}(z-1)}\left(1+\mathcal O(1-v_2)\right), && z\in \partial U^{(v_1)}\cup \partial U^{(-1)}.
\end{aligned}
\end{equation}

We note that \eqref{uinftyint} is bounded below by a fixed positive constant $c_1>0$
under conditions \eqref{partcass} and is uniformly to the main order $\frac{1}{2\pi}\log(1+v_1)^{-1}$,
which is less or equal to $|\tau|/4$, since $\tau\sim \frac{i}{\pi}(\log(1-v_2)^{-1}+\log(1+v_1)^{-1})$.
By \eqref{uint}, provided $c$ is sufficiently small (where we recall that the radii of $U^{(v_2)}$ and $U^{(1)}$ are equal to $c(1-v_2)$, and the radii of $U^{(v_1)}$ and $U^{(-1)}$ are equal to $c(1+v_1)$), $c_1/2<|\Im (u(z)-d-\tau/2+1/2)|<\tau/3$ for $z\in \overline{ U^{(v_2)}}$, and as a consequence $u(z)-d$ is bounded away from $\tau/2+1/2$ modulo the lattice. Similarly, it is straightforward to verify that $u(z)\pm d$ is bounded away from $\tau/2+1/2$ on $ \overline{ U^{(p)}}$, for $p\in \mathcal T$.
By the boundedness of \eqref{fracthetav21}, it follows that $m_{ij}(z;\omega)$ and $\frac{\partial m_{ij}(z;\omega)}{\partial \omega}$ are uniformly bounded for $i,j\in \{1,2\}$ and for $z\in \overline{U^{(p)}}$, with $p\in \mathcal T$, and for future reference we note that by the boundedness of the derivatives of \eqref{fracthetav21} with respect to $\xi, \omega, \tau$, 
\begin{equation}\label{diffmij}
\frac{\partial }{\partial v_2}\frac{\partial m_{ij}(z;\omega)}{\partial \omega},\,\,
\frac{\partial m_{ij}(z;\omega)}{\partial v_2}=\mathcal O\left(\max\left\{ \left|\frac{\partial d}{\partial v_2}\right|,\left|\frac{\partial u(z) }{\partial v_2}\right|, \left| \frac{\partial  \tau}{\partial v_2}\right|\right\}\right),
\end{equation}
as $v_2\to 1$ in the regime \eqref{partcass}, for $z\in  \overline{U^{(p)}}$.

 Combining the statements about boundedness of  $m$ and $\gamma$ and $\gamma^{-1}$, it follows that $\mathcal N(z)$ and $\mathcal N(z)^{-1}$ are uniformly bounded for $z\in U^{(p)}$, $p\in \mathcal T$, and thus by \eqref{asymzetav21}, the jump matrix $J_R(z)$ for $R(z)$ on $U^{(p)}$, $p\in \mathcal T$, has the form
 \begin{equation}
\label{PN-1v21}
P(z)\mathcal N(z;s\Omega)^{-1}=I+\mathcal O\left(\frac{1}{s(1-v_2)}\right) ,\end{equation} 
as $s\to \infty$,  uniformly under conditions \eqref{partcass} and also uniformly for $z\in U^{(p)}$, $p\in \mathcal T$.

The analysis of $J_R(z)$ on the rest of the jump contour is similar, and we obtain uniformly
for \eqref{partcass} and uniformly on this part of the contour
\begin{equation}
\label{NJN-1v21}
\mathcal N(z;s\Omega)
J_S(z)\mathcal N(z;s\Omega)^{-1}=I+\mathcal O\left(e^{-s(1-v_2)c'(1+|z|)}\right),\qquad c'>0.
\end{equation} 

Thus we have a small norm problem for $R$, and as in the previous section we now obtain 
\begin{equation}\label{extRexp}
R(z)=I+\mathcal O\left(\frac{1}{s(1-v_2)}\right),\qquad 
R'(z)\big|_{z=v_2}=\mathcal O\left(\frac{1}{s(1-v_2)^2}\right)
\end{equation}
uniformly on compact sets under conditions \eqref{partcass}.
Therefore the solution of the $\Phi$-Riemann-Hilbert problem for fixed $v_1$, $v_2$ extends
 to the regime \eqref{partcass}. Note, however, that the error terms are different from those in the previous section.

%A simple symmetry argument $x\rightarrow -x$ shows that the solution is also extended into the regime
%\begin{equation}
%1-v_2\ge 1+v_1,\qquad v_2-v_1\ge\epsilon,\qquad  s(1+v_1)\to \infty,
%\end{equation}
%with $1-v_2$ replaced by $1+v_1$ in the estimates for $R$.
%
%Taken together, the above conclusions imply that the solution to the $\Phi$-RH problem
%extends to the region
% \begin{equation}\label{partcass2}
% v_2-v_1\ge\epsilon,\qquad  s(1+v_1)\to \infty,\qquad s(1-v_2)\to\infty,
%\end{equation}
%with the estimates
%\begin{equation}\label{extRexp2}
%R(z)=I+\mathcal O\left(\frac{1}{s(1+v_1)}\right)+\mathcal O\left(\frac{1}{s(1-v_2)}\right),\qquad 
%R'(z)\big|_{z=v_2}=\mathcal O\left(\frac{1}{s(1+v_1)}\right)+\mathcal O\left(\frac{1}{s(1-v_2)^2}\right),
%\end{equation}
%uniformly on compact sets under conditions \eqref{partcass2}.

\section{Preliminary asymptotic formula for the determinant}\label{secPrelim}
For $\nu=z-v_2$ in a neighborhood of $0$, we write the expansions of $\zeta(z)$,
\begin{equation}\label{expzeta}
\sqrt{\zeta(\nu+v_2)}=s\zeta_0 \sqrt{\nu}(1+\zeta_1\nu+\mathcal O(\nu^2)), 
\qquad
\zeta_0=\frac{2(v_2-x_1)(x_2-v_2)}{\sqrt{(1-v_2^2)(v_2-v_1)}}>0, 
\end{equation}
where $-\pi<\arg\nu<\pi$, and the branch cut is on $(-\infty,0]$.  Similarly, we expand $ \gamma(z)$, $m(z)$, and $u(z)$,
\begin{equation}\label{expgamma}
\begin{aligned} \gamma(\nu+v_2)&=\gamma_0\nu^{-1/4}(1+\gamma_1\nu+\mathcal O(\nu^2)),
\qquad \gamma_0 e^{-\pi i /4}=\left(\frac{(1-v_2)(v_2-v_1)}{1+v_2}\right)^{1/4}>0,\\ 
u(\nu+v_2)&=-u_0\nu^{1/2}(1+u_1\nu+\mathcal O(\nu^{2})),
\qquad
u_0=\frac{1}{I_0\sqrt{(v_2-v_1)(1-v_2^2)}}>0,
\\ m_{jk}(\nu+v_2)&=m_{jk,0}+m_{jk,1}\nu^{1/2}+m_{jk,2}\nu+\mathcal O(\nu^{3/2}),
\end{aligned}
\end{equation}
but with branches chosen such that $0<\arg\nu<2\pi$, and the branch cut on $[0,+\infty)$. 
Here $m_{jk}$ are the matrix elements of $m$. Thus, $\arg\nu$ in \eqref{expzeta} and \eqref{expgamma} are the same for $\Im \nu>0$, but are different for $\Im \nu<0$.

Using the definition of $m$ and the jump conditions \eqref{jumpsm}, we easily obtain
the relations:
 \begin{equation}\begin{aligned}\label{idsmij}
m_{11,0}&=m_{12,0}, & m_{21,0}&=m_{22,0},\\
m_{11,1}&=-m_{12,1}, & m_{21,1}&=-m_{22,1},\\
m_{11,2}&=m_{12,2}, & m_{21,2}&=m_{22,2}.
\end{aligned}\end{equation}
We also find
\begin{equation}
\label{formulamjj}\begin{aligned}
m_{jj,0}&=m_{jj,0}(\omega)=\frac{\theta(0)\theta\left(  \pm\omega+d\right)}{\theta\left( \omega\right)\theta(d)}\\
m_{jj,1}&=-m_{jj,0}u_0 \left( \frac{\theta'\left( 
\pm\omega+d\right)}{\theta\left( 
\pm\omega+d\right)}-\frac{\theta'(d)}{\theta(d)}\right),\\
m_{jj,2}&=\frac{m_{jj,0}u_0^2}{2}\left(
\frac{\theta''(\pm\omega+d)}{\theta(\pm\omega+d)}-\frac{\theta''(d)}{\theta(d)}+2\left(\frac{\theta'(d)}{\theta(d)}\right)^2-2\frac{\theta'(\pm\omega+d)\theta'(d)}{\theta(\pm\omega+d)\theta(d)}
\right),
\end{aligned}
\end{equation}
where $\pm$ means $+$ for $j=1$ and $-$ for $j=2$.

Let 
\begin{equation}
\widehat P(z)=\mathcal N(z;s\Omega)\frac{1}{\sqrt{2}}\begin{pmatrix}0&1\\-1&0\end{pmatrix}\begin{pmatrix}1&-i\\-i&1\end{pmatrix}\left(\pi\zeta^{\frac{1}{2}}\right)^{\frac{1}{2}\sigma_3} \Psi\left(\zeta(z)\right)
\end{equation}
By the definition of $\mathcal S$ in \eqref{def:hatPhi}, $R$ in \eqref{defnR} and $X$ in \eqref{defnhatPsi}, and the fact that $\phi(v_2)=0$, 
\begin{equation}\label{eqninter}
\left[\Phi_+^{-1}(v_2)\Phi_+'(v_2)\right]_{12}=-\left[\widehat P_+^{-1}(v_2)\widehat P_+'(v_2)+\widehat P^{-1}_+(v_2)R^{-1}(v_2)R'(v_2)\widehat P_+(v_2)\right]_{21}.
\end{equation}
With the notation of \eqref{expzeta} and \eqref{expgamma} (where the branches of $\sqrt{\nu}$ coincide for $\Im \nu>0$), it is a straightforward calculation relying on the expansion of $I_0(z)$ in  \eqref{asymI0}, the definition of $\mathcal N$ in \eqref{def:calN}, and the identities for $m_{ij}$ in \eqref{idsmij}, to obtain
\begin{equation} \label{FormulahatP}
\begin{aligned}
\widehat P_+(v_2)&= -\gamma_0 \sqrt{\frac{ \pi s\zeta_0}{2}}\begin{pmatrix}im_{11,0}&*\\m_{22,0}&*\end{pmatrix},\\
\widehat P_+'(v_2) &=-\gamma_0 \sqrt{\frac{ \pi s\zeta_0}{2}}\begin{pmatrix}im_{11,0}\left[\frac{m_{11,2}}{m_{11,0}}+\frac{m_{11,1}}{m_{11,0}}\left(\gamma_0^{-2}-\frac{s\zeta_0}{2}\right)+\frac{\zeta_1}{2}+\frac{s^2\zeta_0^2}{4}+\gamma_1-\gamma_0^{-2}\frac{s\zeta_0}{2}\right]&*\\m_{22,0}\left[\frac{m_{22,2}}{m_{22,0}}+\frac{m_{22,1}}{m_{22,0}}\left(\gamma_0^{-2}+\frac{s\zeta_0}{2}\right)+\frac{\zeta_1}{2}+\frac{s^2\zeta_0^2}{4}+\gamma_1+\gamma_0^{-2}\frac{s\zeta_0}{2}\right]&*\end{pmatrix},
\end{aligned}\end{equation}
where we are uninterested in the entries $*$, and $\omega=s\Omega$ in $m_{jj,k}$.

\begin{comment}
\begin{equation}\label{expanhatP}
\widehat P(\nu+v_2)=- \sqrt{\frac{\pi s\zeta_0}{2}}  \mathcal N(\nu+v_2) 
\begin{pmatrix}
i\left(\nu^{1/4}-\frac{s\zeta_0}{2}\nu^{3/4}+\nu^{5/4}\left(\frac{\zeta_1}{2}+\frac{s^2\zeta_0^2}{4}\right)\right)&\ast\\
\nu^{1/4}+\frac{s\zeta_0}{2}\nu^{3/4}+\nu^{5/4} \left(\frac{\zeta_1}{2}+\frac{s^2\zeta_0^2}{4}\right)&\ast
\end{pmatrix}+\mathcal O(\nu^{7/4}),
\end{equation}
where we are uninterested in the values of the entries denoted by $*$.
By expanding \eqref{def:calN} and using \eqref{idsmij}, and substituting in the the expansion into \eqref{expanhatP},
we obtain that
\begin{multline} \label{expanhatP2}
\widehat P(\nu+v_2)= -\gamma_0 \sqrt{\frac{ \pi s\zeta_0}{2}}(1+\mathcal O(\nu^{3/2}))\\ \times
\begin{pmatrix}
im_{11,0}\left(1+\nu\left[\frac{m_{11,2}}{m_{11,0}}+\frac{m_{11,1}}{m_{11,0}}\left(\gamma_0^{-2}-\frac{s\zeta_0}{2}\right)+\frac{\zeta_1}{2}+\frac{s^2\zeta_0^2}{4}+\gamma_1-\gamma_0^{-2}\frac{s\zeta_0}{2}\right]\right) &\ast
\\ m_{22,0}\left(1+\nu\left[\frac{m_{22,2}}{m_{22,0}}+\frac{m_{22,1}}{m_{22,0}}\left(\gamma_0^{-2}+\frac{s\zeta_0}{2}\right)+\frac{\zeta_1}{2}+\frac{s^2\zeta_0^2}{4}+\gamma_1+\gamma_0^{-2}\frac{s\zeta_0}{2}\right]\right)&\ast
\end{pmatrix}.
\end{multline}
We note that $\det \mathcal N(z)=1$, and by expanding $\det \mathcal N(z)$ as $z\to v_2$ and using \eqref{idsmij}, we obtain
\begin{equation}\label{detid1}
m_{11,0}m_{22,0}\left(1+\frac{\gamma_0^2}{2}\left(\frac{m_{11,1}}{m_{11,0}}+\frac{m_{22,1}}{m_{22,0}}\right)\right)=1. \end{equation}
\end{comment}

We will now make use of the first identity (\ref{id1}) in Lemma \ref{ThmThetaids},
which, by the definitions of $m_{jk,\ell}$, we can write in the form
\begin{equation}\label{id1m}
m_{11,0}m_{22,0}+\frac{\gamma_0^2}{2}(m_{11,0}m_{22,1}+m_{22,0}m_{11,1})=1.
\end{equation}
Using this relation, we obtain by \eqref{eqninter} and \eqref{FormulahatP} for the r.h.s. of the differential identity of Lemma \ref{diffFred},
\begin{multline}\label{exact}
\mathcal F_s(v_1,v_2)=\frac{i}{2\pi}\left[\Phi_+^{-1}(v_2)\Phi_+'(v_2)\right]_{12}=\frac{ s^2\zeta_0^2}{4}-\frac{s\zeta_0}{4}
m_{11,0}m_{22,0}\left(\gamma_0^2 \Gamma_2+\Gamma_1\right) \\
+\frac{i s\zeta_0\gamma_0^2}{4}\begin{pmatrix}
im_{22,0}&m_{11,0}
\end{pmatrix}
R^{-1}(v_2)R'(v_2)\begin{pmatrix}
m_{11,0}\\ -im_{22,0}
\end{pmatrix} \end{multline}
where
\begin{equation}
\Gamma_j=\frac{m_{11,j}}{m_{11,0}}-\frac{m_{22,j}}{m_{22,0}},
\end{equation}
and we take $\omega=s\Omega$ in $m_{jj,k}$. 

Now the more explicit asymptotic expression of \eqref{exact} is different (in the error term)
for fixed $v_1$, $v_2$ (Section \ref{SecR}) and for the double scaling regime of Section \ref{secextend}. 

For fixed  $v_1$, $v_2$,
by \eqref{expandR}, \eqref{R1},
\begin{multline}\label{notexact}
\mathcal F_s(v_1,v_2)=\frac{i}{2\pi}\left[\Phi_+^{-1}(v_2)\Phi_+'(v_2)\right]_{12}=\frac{ s^2\zeta_0^2}{4}-\frac{s\zeta_0}{4}
m_{11,0}m_{22,0}\left(\gamma_0^2 \Gamma_2+\Gamma_1\right) \\
+\frac{i \zeta_0\gamma_0^2}{4}W(s\Omega)+\mathcal O\left(s^{-1}\right), \end{multline}
as $s\to \infty$ (uniformly for $v_1,v_2$ bounded away from each other and $\{-1,1\}$), where
\begin{equation}\label{def:W}
W(\omega)=\begin{pmatrix}
im_{22,0}(\omega)&m_{11,0}(\omega)
\end{pmatrix}\sum_{p\in\{-1,v_1,v_2,1\}}
\int_{\partial U^{(p)}}\frac{ s\Delta_1(z;\omega) dz}{2\pi i(z-v_2)^2}\begin{pmatrix}
m_{11,0}(\omega)\\ -im_{22,0}(\omega)
\end{pmatrix}\end{equation}
with integration in the clockwise direction.

For the regime \eqref{partcass} of Section \ref{secextend}, by \eqref{extRexp} and boundedness of $m_{jk}$,
\begin{equation}
s\begin{pmatrix}
im_{22,0}&m_{11,0}
\end{pmatrix}
R^{-1}(v_2)R'(v_2)\begin{pmatrix}
m_{11,0}\\ -im_{22,0}
\end{pmatrix}=W(s\Omega)+\mathcal O\left(\frac{1}{s(1-v_2)^3}\right),
\end{equation}
and since by \eqref{expzeta} and \eqref{expgamma}, and the formulas for $x_1,x_2$ in \eqref{limx1x2}, we have
\begin{equation}\label{zeta0gamma0v21}\zeta_0\gamma_0^2=\mathcal O(1-v_2),\qquad v_2\to 1,
\end{equation}
equation \eqref{exact} becomes
\begin{equation}\label{notexact2}
\mathcal F_s(v_1,v_2)=\frac{ s^2\zeta_0^2}{4}-\frac{\zeta_0 s}{4}
m_{11,0}m_{22,0}\left(\gamma_0^2 \Gamma_2+\Gamma_1\right) 
+\frac{i \zeta_0\gamma_0^2}{4} W(s\Omega)+\mathcal O\left(\frac{1}{s(1-v_2)^2}\right), 
\end{equation}
uniformly under conditions \eqref{partcass}.

\begin{Prop} \label{PropD}
Let 
\begin{equation}\label{DD}
D(v_1,v_2)=
\frac{ s^2\zeta_0^2}{4}-\frac{s\zeta_0}{4}
m_{11,0}m_{22,0}\left(\gamma_0^2 \Gamma_2+\Gamma_1\right) 
+\frac{i \zeta_0\gamma_0^2}{4}\int_0^1W(\omega)d\omega,\end{equation}
where $\zeta_0$ and $\gamma_0$ are given in \eqref{expzeta}, \eqref{expgamma}, $
\Gamma_j=\frac{m_{11,j}}{m_{11,0}}-\frac{m_{22,j}}{m_{22,0}}$, with $m_{jj,k}=m_{jj,k}(s\Omega)$ from \eqref{formulamjj}, and where $W$ is given in \eqref{def:W} (with $\Delta_1$ defined by \eqref{Delta1} and \eqref{Delta11}). 
\begin{itemize}
\item[(a)]
Let $V\in(0,1)$, and let $\widehat A=(-1,-V)\cup(V,1)$. Let $v_2=-v_1$, and denote $v=v_2$. Fix $\epsilon>0$.
Then
\begin{equation}\nonumber
\log \det (I-K_s)_{\widehat A}-\log \det (I-K_s)_{A_s}=2\int_{1-\frac{2t}{s}}^V D(-v,v)dv+\mathcal O\left(\frac{1}{t}\right),
\end{equation}
as $s\to \infty$, uniformly for $\epsilon\le V\le 1-\frac{2t}{s}$,
where $t(s)\to\infty$, $t\le\frac{1}{2}(\log s)^{1/4}$, and $A_s=(-1,-1+2t/s)\cup (1-2t/s,1)$.

\item[(b)]
 Let $-1<V_1<0$ and $V_2$ be fixed, $V_1<V_2<1$,
 and denote $A=(-1,V_1)\cup(V_2,1)$. Then
\begin{equation}\nonumber
\log \det (I-K_s)_ A-\log \det(I-K_s)_{(-1,V_1)\cup (-V_1,1)}=\int_{-V_1}^{ V_2}D(V_1,v_2)dv_2+\mathcal O\left(\frac{1}{s}\right),
\end{equation}
as $s\to \infty$.
\item[(c)] Let $A=(-1,V_1)\cup(V_2,1)$, 
 and a fixed $\epsilon>0$, and with $-1< V_1<\widehat V_2< 1$ and
\begin{equation}\label{cond}
1-V_2 \le 1-\widehat V_2 \le
1+V_1,\qquad  V_2-V_1\ge\epsilon,\qquad  s(1-V_2)\to \infty.
\end{equation}
Then
\begin{equation}\nonumber
\log \det (I-K_s)_ A-\log \det(I-K_s)_{(-1,V_1)\cup (\widehat V_2,1)}=\int_{\widehat V_2}^{ V_2}D(V_1,v_2)dv_2+\mathcal O\left(\frac{1}{(1-V_2)s}\right),
\end{equation}
as $s\to \infty$, uniformly in the regime \eqref{cond}.
\end{itemize}
\end{Prop}

\begin{remark}
In the proof, considering the effects of averaging w.r.t. $\omega=s\Omega$,
we will show that \eqref{DD} gives the main contribution, and the error terms are as presented.
\end{remark}

\begin{remark}\label{use of Prop}
Part (a) allows us to integrate over symmetric intervals from the position of 2 small ones at $1$ and $-1$ 
(where Lemma \ref{Lemmasep} holds) to  
general symmetric intervals with a fixed $0<V<1$. Part (b) allows then to move the $V_2$ edge to an arbitrary fixed
position $V_1\equiv -V<V_2<1$. Note that the condition $-1<V_1<0$ here is not a loss of generality for 
$\det (I-K_s)_ A$, since we can use the symmetry $x\rightarrow -x$ of the determinant. 

Part (c) allows us to to integrate to reach a scaling limit where $V_2=V_2(s)$ can
approach $1$ provided $s(1-V_2)\to\infty$ and $V_1$ is fixed.

Finally, choose a $V_1(s)=-V_2(s)$ such that $2t=(1+V_1)s\to\infty$ (in this case, Lemma \ref{Lemmasep} still holds
by Remark \ref{Lemmasepremark}),
and then, if needed, move $V_2$ closer to $1$ using Part (c). Then, if needed, one can use the symmetry 
$x\rightarrow -x$, to reach an arbitrary situation with $(1+V_1)s\to\infty$, $(1-V_2)s\to\infty$.
\end{remark}

\begin{proof}
We first prove Part (b) of the proposition, then Part (c), and finally Part (a).
By \eqref{notexact} and the differential identity \eqref{diffid new}, all we need to do for the proof of Part (b) is to show 
that, with $\widehat V_2=-V_1$,
\begin{equation}\label{134}
\int_{\widehat V_2}^{V_2}\zeta_0\gamma_0^2W(s\Omega)dv_2=
\int_{\widehat V_2}^{V_2}\zeta_0\gamma_0^2\int_0^1 W(\omega)d\omega dv_2+\mathcal O(s^{-1}),
\end{equation}
as $s\to \infty$. Denote $f(\omega;v_2,v_1)=\zeta_0\gamma_0^2W(\omega)$. This function is analytic in both $\omega$ and $v_2$ ($v_2$ is bounded away from $v_1$ and $1$). 
Let $f_j$ denote its Fourier coefficients with respect to $\omega$, so that
\begin{equation}f(\omega;v_2,v_1)=\zeta_0\gamma_0^2W(\omega)=
\sum_{j=-\infty}^\infty f_j(v_2,v_1)e^{2\pi ij\omega}.\end{equation}
For $j\neq 0$, it follows by integration by parts that
\begin{multline}\label{fjint}
\left|\int_{\widehat V_2}^{V_2} f_j(v_2,v_1)e^{2\pi ijs\Omega}dv_2\right|\\=\frac{1}{2\pi |j|s}\left|\left[ \frac{f_j(v_2,v_1)e^{2\pi ijs\Omega}}{\frac{\partial}{\partial v_2}\Omega(v_2,v_1)}\right]_{\widehat V_2}^{V_2}-\int_{\widehat V_2}^{V_2} \frac{\partial}{\partial v_2}\left(\frac{f_j(v_2,v_1)}{\frac{\partial}{\partial v_2}\Omega(v_2,v_1)}\right)e^{2\pi ijs\Omega}dv_2\right|.\end{multline}
In Proposition \ref{Propterm2} (b) below we give an explicit formula for $\frac{\partial}{\partial v_2}\Omega(v_2,v_1)$, and in particular it is a strictly positive differentiable function bounded away from zero when $v_2$ is bounded away from $v_1$ and $1$. Thus
\begin{equation}\begin{aligned}
\int_{\widehat V_2}^{V_2} f(s\Omega;v_2,v_1)dv_2&=\sum_{j=-\infty}^\infty \int_{\widehat V_2}^{V_2} f_j(v_2,v_1)e^{2\pi ijs\Omega}dv_2\\
&=\int_{\widehat V_2}^{V_2} f_0(v_2,v_1)dv_2+\mathcal O\left(\frac{1}{s}\right),\qquad  s\to \infty,
\end{aligned}\end{equation}
which yields \eqref{134} since $f_0(v_2,v_1)=\zeta_0\gamma_0^2\int_0^1W(\omega)d\omega$.

\bigskip

We now prove Part (c) of the proposition.

Substituting \eqref{limx1x2} into the expression  \eqref{dOmdv} 
for $\frac{\partial \Omega}{\partial v_2}$ in Proposition \ref{Propterm2} below, 
and also using \eqref{Partial5},
we obtain
\begin{equation}\label{dOmegav21}
\frac{\partial \Omega}{\partial v_2}=\frac{3-v_1}{4\pi \sqrt{2(1-v_1)}}+\mathcal O(1-v_2),\qquad
\frac{\partial^2\Omega}{\partial v_2^2}=\mathcal O((1-v_2)^{-1}),
\end{equation}
and, in particular, $\frac{\partial \Omega}{\partial v_2}$ remains bounded away from $0$. 

We now show that
\begin{equation} \label{fjv21}
f_j(v_2,v_1)=\mathcal O\left(\frac{1}{j(1-v_2)}\right), \qquad \frac{\partial}{\partial v_2}f_j(v_2,v_1)=\mathcal O\left(\frac{1}{j(1-v_2)^2}\right), \end{equation}
as $v_2\to 1$, for $j\neq 0$, uniformly under conditions \eqref{partcass}, which proves Part (c) of the  proposition 
by \eqref{notexact2} and  arguments similar to those we used in the proof of Part (b).

Since
\begin{equation}\label{Fdom} 
\left|f_j(v_2,v_1)\right|=\left|\int_0^1f(\omega;v_2,v_1)e^{-2\pi ij \omega}d\omega\right| 
=\left|\frac{1}{2\pi j}\int_0^1\frac{\partial}{\partial\omega}f(\omega;v_2,v_1)e^{-2\pi i j \omega}d\omega\right|,
\qquad j\neq 0,
\end{equation}
and similarly for $\frac{\partial}{\partial v_2}f_j(v_2,v_1)$, it suffices to show that
\begin{equation}\label{v21suff}
\frac{\partial}{\partial \omega}f(\omega;v_2,v_1)=\mathcal O\left(\frac{1}{1-v_2}\right), \qquad \frac{\partial}{\partial \omega}\frac{\partial}{\partial v_2}f(\omega;v_2,v_1)=\mathcal O\left(\frac{1}{(1-v_2)^2}\right),
\end{equation}
as $v_2\to 1$.

It follows by the definition of $W$ in \eqref{def:W}, \eqref{PN-1v21}, \eqref{zeta0gamma0v21},
and the arguments of the previous section that 
\begin{equation} \label{asymfv21}
f(\omega;v_2,v_1)=\zeta_0\gamma_0^2W(\omega)=
\mathcal O\left(\frac{1}{1-v_2}\right)
\end{equation}
as $v_2\to 1$ under conditions \eqref{partcass}.

We recall that $\frac{\partial m_{ij}(z)}{\partial \omega}$, $i,j\in\{1,2\}$,  are uniformly bounded for $z\in \overline{U^{(p)}}$, $p\in \mathcal T$, and so $\frac{\partial f(\omega;v_2,v_1)}{\partial \omega}$ satisfies the same upper bound as $f(\omega;v_2,v_1)$ given in \eqref{asymfv21}, proving the first bound in  \eqref{v21suff}.

To obtain the second one,
we observe first  that by \eqref{limtau}, 
\begin{equation}\label{dtau}
\frac{\partial }{\partial v_2}\tau=\mathcal O\left(\frac{1}{1-v_2}\right),\qquad v_2\to 1,
\end{equation}  
and by \eqref{uint},
  \begin{equation} \label{partialuv21}
  \frac{\partial u(z)}{\partial v_2}=\mathcal O\left(\frac{1}{1-v_2}\right),\end{equation} as $v_2\to 1$, uniformly for $z\in \partial U^{(p)}$, $p\in \mathcal T$.  
 
It follows by \eqref{uinftyd1} and \eqref{uinftyint} that $\frac{\partial}{\partial v_2}d=\mathcal O\left(\frac{1}{1-v_2}\right)$, as $v_2\to 1$.
  Thus, by \eqref{diffmij},
\begin{equation}\frac{\partial m_{ij}(z;\omega)}{\partial v_2}=\mathcal O\left(\frac{1}{1-v_2}\right),\end{equation}
as $v_2\to 1$, uniformly for $z\in \overline{ U^{(p)}}$, $p\in \mathcal T$. Furthermore, by the definition \eqref{def:gamma},
\begin{equation}
\frac{\partial}{\partial v_2}\gamma(z),\, \frac{\partial}{\partial v_2}\gamma^{-1}(z)=\mathcal O\left(\frac{1}{1-v_2}\right).
\end{equation}
By \eqref{limx1x2} and \eqref{def:zeta},
\begin{equation}\label{asymzetav21'}\frac{\partial}{\partial v_2}\left(\frac{1}{\sqrt{\zeta(z)}}\right)=\mathcal O\left(\frac{1}{s(1-v_2)^2}\right).\end{equation}
The above bounds taken together imply
\begin{equation}\label{v21derivDelta}
s\frac{\partial}{\partial v_2}\Delta_1(z)=\mathcal O\left(\frac{1}{(1-v_2)^2}\right).
\end{equation}
It follows by the definition of $W$ in \eqref{def:W}, \eqref{zeta0gamma0v21}, and boundedness of $m_{jk}$ that
\begin{equation}
\frac{\partial}{\partial v_2}f(\omega;v_2,v_1)=\mathcal O\left(\frac{1}{(1-v_2)^2}\right), \end{equation}
as $v_2\to 1$, uniformly under conditions \eqref{partcass}.
Since $\frac{\partial}{\partial \omega}\frac{\partial m_{ij}(z)}{\partial v_2}=\mathcal O\left(\frac{\partial m_{ij}(z)}{\partial v_2}\right)$, it follows that $\frac{\partial}{\partial \omega}\frac{\partial}{\partial v_2}f(\omega;v_2,v_1)=\mathcal O\left(\frac{\partial}{\partial v_2}f(\omega;v_2,v_1)\right)$, which proves the second bound in \eqref{v21suff}, completing
the proof of Part (c) of the proposition.

\bigskip

To show Part (a), we let $v_2=-v_1=v$, and take the limit $s\to \infty$ such that $\epsilon<v<1-\frac{M}{s}$ for some $\epsilon>0$ and a sufficiently large $M$. By \eqref{diffidsym},
\begin{equation}
\frac{\partial}{\partial v}\det(I-K_s)_{(-1,-v)\cup(v,1)}=2\mathcal F_s(-v,v). 
\end{equation}
We observe that \eqref{notexact2} is valid also for $v_2=-v_1=v$, and all that remains to finish the proof of Part (a) of the proposition is to consider the Fourier coefficients of $f$.
 In place of \eqref{fjint}, we have
\begin{equation}\label{fjint2}
\left|\int_{\widehat V}^{V} f_j(v,-v)e^{2\pi ijs\Omega}dv\right|\\=\frac{1}{2\pi |j|s}\left|\left[ \frac{f_j(v,-v)e^{2\pi ijs\Omega}}{\frac{\partial}{\partial v}\Omega(v,-v)}\right]_{\widehat V}^{V}-\int_{\widehat V}^{V} \frac{\partial}{\partial v}\left(\frac{f_j(v,-v)}{\frac{\partial}{\partial v}\Omega(v,-v)}\right)e^{2\pi ijs\Omega}dv\right|.\end{equation}
By above arguments, it suffices to show that the right hand side of \eqref{fjint2} is of order $\frac{1}{j^2 s (1-v)^2}$.
To do this we need the first bound in \eqref{v21suff}, which holds also for $v_2=-v_1=v$, and additionally we need to prove that
\begin{equation}\label{vsuff}
 \frac{\partial}{\partial \omega}\frac{\partial}{\partial v}f(\omega;v,-v)=\mathcal O\left(\frac{1}{(1-v)^2}\right),
\end{equation}
as $v\to 1$, and that $\frac{d}{dv}\Omega(v,-v)$ remains bounded away from $0$. 
Note that, using contour integration, 
\[
\Omega^{-1}(v_2,v_1)=
I_0=\int_{v_2}^{1}\frac{dx}{\sqrt{|p(x)|}}=\int_{-1}^{v_1}\frac{dx}{\sqrt{|p(x)|}},
\]
and therefore
\begin{equation}\label{Ominlimit}
\Omega(v_2,v_1)=\Omega(-v_1,-v_2),\qquad \frac{\partial}{\partial v}\Omega(v,-v)=2\frac{\partial }{\partial v_2}\Omega(v_2,-v)\Big|_{v_2=v}.
\end{equation}
The last derivative is thus bounded away from $0$ by \eqref{dOmegav21}. 

In order to prove \eqref{vsuff}, we simply observe that the bounds obtained in \eqref{partialuv21}--\eqref{asymzetav21'}  also hold for the derivatives with respect to $v$ instead of $v_2$, which yields 
\begin{equation} \frac{\partial }{\partial v}f(\omega;v,-v)=\mathcal O\left(\frac{1}{(1-v)^2}\right),
\end{equation}
as $v\to 1$. Since $\frac{\partial}{\partial \omega}\frac{\partial m_{ij}(z)}{\partial v}=\mathcal O\left(\frac{\partial m_{ij}(z)}{\partial v}\right)$, it follows that $\frac{\partial}{\partial \omega}\frac{\partial}{\partial v}f(\omega;v,-v)=\mathcal O\left(\frac{\partial}{\partial v}f(\omega;v,-v)\right)$, which proves \eqref{vsuff} and thus Part (a) of the proposition.

\end{proof}

\section{Proof of Theorems \ref{Thm} and \ref{Thmtau1}}\label{secProof1}
In the next 3 sections, we show that \eqref{DD} in Proposition \ref{PropD} can be written as
\begin{equation}\label{newD}
D(v_1,v_2)=\frac{\partial}{\partial v_2}\mathcal G(s;v_1,v_2)
+
\frac{\partial \tau}{\partial v_2}\int_0^1\frac{\partial }{\partial \tau} \log \theta_3(\omega;\tau)d\omega
-\frac{\partial \tau}{\partial v_2}\frac{\partial }{\partial \tau} \log \theta_3(s\Omega;\tau), 
\end{equation}
where
\begin{equation}\label{mathcalG}
\mathcal G=
s^2\left(\frac{I_2-\frac{v_2+v_1}{2}I_1}{I_0}-\frac{(v_2-v_1)^2}{8}\right)
+\log \theta(s\Omega;\tau)-\frac{1}{2}\log I_0-\frac{1}{8}\sum_{y\in\{-1,v_1,v_2,1\}}\log |q(y)|.
\end{equation}

We now use \eqref{newD} to prove Theorems  \ref{Thm} and \ref{Thmtau1}.
First, we show that with $\widehat V_2$ fixed, and in all asymptotic regimes of 
Proposition \ref{PropD},
\begin{equation}\label{averagelogOm}
\int_{\widehat V_2}^{V_2}\left(
\frac{\partial \tau}{\partial v_2}\int_0^1\frac{\partial }{\partial \tau} \log \theta_3(\omega;\tau)d\omega
-\frac{\partial \tau}{\partial v_2}\frac{\partial }{\partial \tau} \log \theta_3(s\Omega;\tau)
\right) dv_2
=
\mathcal O\left(\frac{1}{s(1-V_2)}\right),\qquad
s\to\infty,
\end{equation}
uniformly in integration regimes of Proposition \ref{PropD}, 
and so this part only contributes to the error term.

Using the differential equation \eqref{thdiff} and \eqref{dtau}, we write
\begin{equation}\label{logtauinf1}
\frac{\partial}{\partial\omega}\left(\frac{\partial \tau}{\partial v_2}\frac{\partial }{\partial \tau} \log \theta_3(\omega;\tau)\right)=\frac{1}{4\pi i} 
\frac{\partial \tau}{\partial v_2}
\left(\frac{\theta''_3}{\theta_3}\right)'(\omega)=
\mathcal O\left(\frac{1}{1-v_2}\right).
\end{equation}
Also since by \eqref{limtau},
\begin{equation}
\frac{\partial^2 \tau}{\partial v_2^2}=
\mathcal O\left(\frac{1}{(1-v_2)^2}\right),
\end{equation}
we similarly obtain
\begin{equation}\label{logtauinf2}
\frac{\partial}{\partial\omega}\frac{\partial}{\partial v_2}\left(\frac{\partial \tau}{\partial v_2}\frac{\partial }{\partial \tau} \log \theta_3(\omega;\tau)\right)= \mathcal O\left(\frac{1}{(1-v_2)^2}\right).
\end{equation}
The estimates \eqref{logtauinf1} and \eqref{logtauinf2} imply, 
by similar arguments to \eqref{Fdom}, \eqref{v21suff}, \eqref{fjint}, the estimate \eqref{averagelogOm}.

We now apply Part (a) of Proposition \ref{PropD} to integrate \eqref{newD} from the position of 2 symmetric small
intervals $v=-v_1=v_2=1-\frac{2t}{s}$, $t=\frac{1}{2}\log(s)^{1/4}$,
where Lemma \ref{Lemmasep} can be applied, to the case of $V=-v_1=v_2>0$ {\it fixed}.
If $-v_1=v_2=v$, by symmetry under the exchange $v_2\rightarrow -v_1$, $v_1\rightarrow -v_2$,
\[
2\frac{\partial}{\partial v_2}\mathcal G(s;v_1,v_2)=
\frac{\partial}{\partial v}\mathcal G(s;-v,v).
\]
Thus, applying Part (a) of Proposition \ref{PropD} and using Lemma \ref{Lemmasep},
we obtain
\begin{equation}\label{almostthere}
\log\det(I-K_s)_{\widehat A}=
\mathcal G(s;-V,V)-\mathcal G\left(s;-1+\frac{2t}{s},1-\frac{2t}{s}\right)
-t^2-\frac{1}{2}\log t +2c_0+\mathcal O(1/t).
\end{equation}
To finish the proof of Theorem \ref{Thm} in the symmetric case, we need to estimate 
$\mathcal G\left(s;-1+\frac{2t}{s},1-\frac{2t}{s}\right)$. 
Using formulae \eqref{I02}, \eqref{KEexp}, we obtain 
  in our case $v=1-\frac{2t}{s}$ (recall that $v'^2=1-v^2$) 
\begin{equation}\label{first}
I_0(-v,v)=\frac{\pi}{2}\left(1
+\frac{t}{s}+\frac{5t^2}{4s^2}+
\mathcal O((t/s)^3)\right), \qquad 
\frac{I_2(-v,v)}{I_0(-v,v)}=1
-\frac{2t}{s}+\frac{t^2}{s^2}+
\mathcal O((t/s)^3),
\end{equation}  
and so the term with $s^2$ in $\mathcal G\left(s;-1+\frac{2t}{s},1-\frac{2t}{s}\right)$ becomes
\begin{equation}\label{second}
\frac{I_2(-v,v)}{I_0(-v,v)}-\frac{v^2}{2}=\frac{1}{2}-\frac{t^2}{s^2}+
\mathcal O((t/s)^3).
\end{equation}
The term $\log\theta$ gives a contribution only to the error term, indeed, since by \eqref{KEexp}
\[
J_0(-v,v)=2K(v)=\left(\log\frac{4s}{t}\right)\left(1+\mathcal O(t/s)\right),\qquad
\tau=i\frac{J_0}{I_0}=\frac{2i}{\pi}\left(\log\frac{4s}{t}\right)\left(1+\mathcal O(t/s)\right),
\]
we have that
\begin{equation}\label{third}
\log\theta(s\Omega)=\log\left(1+\mathcal O((t/s)^2)\right)=\mathcal O((t/s)^2),\qquad -v_1=v_2=v=
1-\frac{2t}{s}.
\end{equation}
Finally, in this case 
\begin{equation}
|q(1)|=|q(-1)|=1-\frac{I_2}{I_0}=\frac{2t}{s}\left(1+\mathcal O(t/s)\right),\qquad
|q(-v)|=|q(v)|=\frac{I_2}{I_0}-v^2=\frac{2t}{s}\left(1+\mathcal O(t/s)\right),
\end{equation}
and so
\begin{equation}\label{fourth}
-\frac{1}{8}\sum_{y\in\{-1,v_1,v_2,1\}}\log |q(y)|=
-\frac{1}{2}\log\frac{2t}{s}+\mathcal O(t/s).
\end{equation}

Substituting \eqref{first}, \eqref{second}, \eqref{third}, \eqref{fourth} into the expression
\eqref{mathcalG} for $\mathcal G\left(s;-1+\frac{2t}{s},1-\frac{2t}{s}\right)$, and
that, in turn, into \eqref{almostthere}, we obtain asymptotics \eqref{FormDIZ} with an error term $o(1)$ and with
$\widehat G_1$ and $c_1$ as in \eqref{constant} in the case
$-v_1=v_2=V>0$. We then extend it to the general case of fixed $-1<v_1<v_2<1$ by now a
straightforward application of Part (b) of Proposition \ref{PropD}. (In fact, for $v_1<0$, but the general case follows by a symmetry argument: see Remark \ref{use of Prop}.) 
Now since by \cite{DIZ}, \eqref{FormDIZ} (with the error term $\mathcal O(s^{-1})$)
holds for {\it some} constants $\widehat G_1$, $c_1$, these must be equal to those in \eqref{constant}.
This completes the proof of Theorem \ref{Thm}.

Given Theorem \ref{Thm}, we immediately obtain Theorem \ref{Thmtau1} by applying Part (c) of 
Proposition \ref{PropD} and a symmetry argument as discussed in Remark \ref{use of Prop}.

%K(v)&=\left(\frac{1}{2}\log \frac{1}{2-2v}+2\log 2\right)(1+\mathcal O(1-v)),\\

Thus, all that remains now is to verify \eqref{newD}.
In Section \ref{SecLead} we consider the leading order term in \eqref{DD}, in Section \ref{SecFluc} we consider the term involving $(\gamma_0^2\Gamma_2+\Gamma_1$), which yields the derivative of $\log \theta(s\Omega)$, and in Section \ref{SecConst} we consider the term with $\int_0^1W(\omega)d\omega$, which yields the constant.
Thus we will prove the following 3 lemmata, which taken together imply \eqref{newD}.

\begin{Lemma}\label{leading}
\begin{equation}
\frac{\zeta_0^2}{4}=\frac{\partial}{\partial v_2}\left(\frac{I_2-\frac{v_2+v_1}{2}I_1}{I_0}-\frac{(v_2-v_1)^2}{8}\right).
\end{equation}
\end{Lemma}

\begin{Lemma}\label{subleading}
\begin{equation}
-\frac{s\zeta_0}{4}
m_{11,0}m_{22,0}\left(\gamma_0^2 \Gamma_2+\Gamma_1\right)=
\frac{\partial}{\partial v_2}\log \theta_3(s\Omega;\tau)-\frac{\partial \tau}{\partial v_2}\frac{\partial }{\partial \tau} \log \theta_3(s\Omega;\tau).
\end{equation}
Note that the r.h.s. here equals the partial derivative $
s\frac{\partial\Omega}{\partial v_2}
\frac{\partial}{\partial (s\Omega)}\log \theta_3(s\Omega;\tau)$ with $\tau$ fixed.
\end{Lemma}

\begin{Lemma}\label{constantlemma}
\begin{equation}
\begin{aligned}
\frac{i \zeta_0\gamma_0^2}{4}\int_0^1W(\omega)d\omega=
-\frac{\partial}{\partial v_2}\left(\frac{1}{2}\log I_0+\frac{1}{8}\sum_{y\in\{-1,v_1,v_2,1\}}\log |q(y)|\right)\\
+\frac{\partial \tau}{\partial v_2}\int_0^1\frac{\partial }{\partial \tau} \log \theta_3(\omega;\tau)d\omega,
\end{aligned}
\end{equation}
where $W(\omega)$ is given in \eqref{def:W}.
\end{Lemma}

\section{The leading order term. Proof of Lemma \ref{leading}} \label{SecLead}
Recall from \eqref{IJ} the notation for $I_j$, $J_j$, $j=0,1,2$. 
We will calculate the derivatives $\frac{\partial}{\partial v_2}I_j$, $j=0,1,2$ in terms of the integrals themselves. The crucial identity here
is \eqref{zv-deriv} below.

First, we have
\begin{equation}
\frac{\partial I_j}{\partial v_2}=
\frac{i}{4}\int_{A_1} \frac{z^j}{(z-v_2)\sqrt{p(z)}}dz,\qquad j=0,1,2.
\end{equation}
Therefore,
\begin{equation} \label{Partial2}
\frac{\partial I_1}{\partial v_2}=
\frac{i}{4}\int_{A_1} \frac{z-v_2+v_2}{(z-v_2)\sqrt{p(z)}}dz=
I_0/2+v_2\frac{\partial I_0}{\partial v_2},
\end{equation}
and similarly,
\begin{equation} \label{Partial3}
\frac{\partial I_2}{\partial v_2}=
I_1/2+v_2\frac{\partial I_1}{\partial v_2}.
\end{equation}
The last 2 equations imply
\begin{equation} \label{Partial4}
\frac{\partial I_2}{\partial v_2}=v_2^2 \frac{\partial I_0}{\partial v_2}+I_1/2+v_2 I_0/2.
\end{equation}
From here and \eqref{Partial2},
\begin{equation} \label{Partial6}
\frac{\partial}{\partial v_2}\left(2I_2-(v_2+v_1)I_1\right)=(v_2-v_1)\left[v_2\frac{\partial I_0}{\partial v_2}+I_0/2\right].
\end{equation}

We will now express the derivative $\partial I_0/\partial v_2$ in terms of $I_j$'s.
To this end, observe that 
\begin{equation}\label{zv-deriv}
\begin{aligned}
0=\frac{i}{2} \int_{A_1} \frac{d}{dz}\sqrt{\frac{(z^2-1)(z-v_1)}{z-v_2}}dz &=
-i\frac{v_2-v_1}{4}\int_{A_1} \frac{z^2-1}{(z-v_2)\sqrt{p(z)}}dz+I_2-v_1I_1\\
&=-(v_2-v_1) \frac{\partial}{\partial v_2}\left(I_2-I_0\right)  +I_2-v_1I_1,
\end{aligned}
\end{equation}
so that
\begin{equation} \label{Partial1}
\frac{\partial}{\partial v_2}\left(I_2-I_0\right)=\frac{I_2-v_1I_1}{v_2-v_1}.
\end{equation}
Using this equation and \eqref{Partial4} we have
\begin{equation} \label{Partial5}
\frac{\partial I_0}{\partial v_2}=\frac{-I_2+\frac{v_2+v_1}{2}I_1+\frac{v_2(v_2-v_1)}{2}I_0}{(1-v_2^2)(v_2-v_1)}. 
\end{equation}
This and \eqref{Partial6} imply
\begin{equation}\label{derivativex1x2}
\frac{\partial}{\partial v_2}\left(\frac{I_2-\frac{v_2+v_1}{2}I_1}{I_0}\right)
=\frac{v_2-v_1}{4}+\frac{\left(2I_2-(v_1+v_2)I_1+v_2(v_1-v_2)I_0\right)^2}{4I_0^2(1-v_2^2)(v_2-v_1)}.
\end{equation}
By the formulas for $x_1$ and $x_2$ in \eqref{x1x2eqn1} and \eqref{x1x2eqn2}, and the formula for $\zeta_0$ in \eqref{expzeta}, we finish the proof of Lemma \ref{leading}.

\begin{remark}
We also observe for future reference that the arguments may be copied line for line and applied to the integrals 
$J_j=\int_{v_1}^{v_2}\frac{x^jdx}{\sqrt{|p(x)|}}$
 (by instead considering an integral over a closed loop containing $(v_1,v_2)$ and different branches of the roots), and we obtain the analogues to \eqref{Partial5} and \eqref{Partial6}:
\begin{align}\label{Partial7}
&\frac{\partial J_0}{\partial v_2}=\frac{-J_2+\frac{v_2+v_1}{2}J_1+\frac{v_2(v_2-v_1)}{2}J_0}{(1-v_2^2)(v_2-v_1)},\\
&\frac{\partial}{\partial v_2}\left(2J_2-(v_2+v_1)J_1\right)=(v_2-v_1)\left[v_2\frac{\partial J_0}{\partial v_2}+J_0/2\right]. \label{Partial8}
\end{align}
\end{remark}

\section{The fluctuations. Proof of Lemma \ref{subleading}}\label{SecFluc}
We write the first subleading term in \eqref{DD} in the form, using \eqref{expzeta}, \eqref{expgamma} for $\zeta_0$, $u_0$,
\begin{equation}\label{def:T1}
\frac{ s\zeta_0 u_0}{4}T_1(s\Omega),\qquad
\frac{ \zeta_0 u_0}{4}=\frac{(v_2-x_1)(x_2-v_1)}{2I_0(v_2-v_1)(1-v_2^2)},\qquad
T_1(\omega)=-\frac{m_{11,0}m_{22,0}}{u_0}\left(\gamma_0^2\Gamma_2+\Gamma_1\right).
\end{equation}

Our goal in this section is to prove the following proposition,
of which Lemma \ref{subleading} is an immediate corollary.
\begin{Prop}\label{Propterm2} There hold the identities:
\begin{itemize}
\item[(a)]
\begin{equation}\label{derivativetau}
-i\frac{\partial\tau}{\partial v_2}=\frac{\partial |\tau|}{\partial v_2}=
\frac{\pi}{I_0^2(1-v_2^2)(v_2-v_1)}=\pi u_0^2,
\end{equation}
\item[(b)]
\begin{equation}\label{dOmdv}
\frac{\partial\Omega}{\partial v_2}=\frac{(v_2-x_1)(x_2-v_2)}{I_0(1-v_2^2)(v_2-v_1)},
\end{equation}
\item[(c)] 
\begin{equation}\label{T1th}
T_1(\omega)=2\frac{\theta_3'(\omega)}{\theta_3(\omega)}.
\end{equation}
\end{itemize}
\end{Prop}

\begin{proof}
To show Part (a) note that in the notation of \eqref{IJ}
\[
|\tau|=\frac{J_0}{I_0},
\]
and therefore, using \eqref{Partial5}, \eqref{Partial7}, we have
\begin{equation}
\frac{\partial |\tau|}{\partial v_2}=
\frac{I_2J_0-J_2I_0-\frac{v_1+v_2}{2}(I_1J_0-I_0J_1)}{I_0^2(1-v_2^2)(v_2-v_1)},
\end{equation}
which gives Part (a) of the proposition by Riemann's bilinear relations \eqref{RBL}.

Part (b) follows from \eqref{def:Omega} and \eqref{Partial5} by using \eqref{x1x2eqn1}, \eqref{x1x2eqn2}:
\begin{equation}
\frac{\partial \Omega}{\partial v_2}=-\frac{1}{I_0^2}\frac{\partial I_0}{\partial v_2}
=-\frac{x_1x_2+v_2(v_2-v_1)/2}{I_0 (1-v_2^2)(v_2-v_1)}=
\frac{(v_2-x_1)(x_2-v_2)}{I_0(1-v_2^2)(v_2-v_1)}.
\end{equation}

%By the definition of $\Omega$ in \eqref{def:Omega} and by \eqref{x1x2eqn1}, \eqref{x1x2eqn2},
%\begin{equation}
%\Omega=\frac{1}{\pi}\int_{v_1}^{v_2}\frac{|q(x)|}{\sqrt{|p(x)|}}dx=
%\frac{1}{\pi}(-J_2+(x_1+x_2)J_1-x_1x_2J_0)=
%\frac{1}{\pi}\left[\frac{J_0}{I_0}\left(I_2-\frac{v_1+v_2}{2}I_1\right)-J_2+\frac{v_1+v_2}{2}J_1\right].
%\end{equation}
%Using the differential identities \eqref{Partial6} and \eqref{Partial8}, we obtain
%\begin{equation}
%\frac{\partial \Omega}{\partial v_2}=\frac{1}{\pi}\left[\frac{J_0}{2I_0}(v_2-v_1)\left(v_2\frac{\partial I_0}{\partial v_2}+I_0/2\right)+\frac{\partial}{\partial v_2}\left(\frac{J_0}{I_0}\right)\left(I_2-\frac{v_1+v_2}{2}I_1\right)-\frac{(v_2-v_1)}{2}\left(v_2\frac{\partial J_0}{\partial v_2}+J_0/2\right)\right].
%\end{equation}
%By noting that the terms with $\frac{\partial I_0}{\partial v_2}$ and $\frac{\partial J_0}{\partial v_2}$ here combine to form $\frac{\partial}{\partial v_2}\frac{J_0}{I_0}=
%\frac{d|\tau|}{dv_2}$, 
%\begin{equation}
%\frac{\partial \Omega}{\partial v_2}=\frac{1}{\pi}\left(I_2-\frac{v_1+v_2}{2}I_1-\frac{(v_2-v_1)v_2I_0}{2}\right)\frac{d|\tau|}{dv_2}.
%\end{equation}
%Recalling again the formulas for  $x_1$ and $x_2$ in \eqref{x1x2eqn1} and \eqref{x1x2eqn2}, we obtain
%\begin{equation}\label{PartOm1}
%\frac{\partial \Omega}{\partial v_2}=I_0\frac{(v_2-x_1)(x_2-v_2)}{\pi}\frac{d|\tau|}{dv_2},
%\end{equation}
%which completes the proof of Part (b) in view of Part (a).

We will now show Part (c).  Substituting the definitions of $m_{jj,k}$ and $\Gamma_j$
into $T_1$ in \eqref{def:T1}, and using the identity \eqref{id2} of Lemma \ref{ThmThetaids}, we write $T_1$ in the form
\begin{equation}\label{T1mid}
T_1(\omega)=\gamma_0^2 u_0
\frac{\theta(0)^2\theta(\omega+d)\theta(\omega-d)}{\theta(d)^2\theta(\omega)^2}\left[ \frac{\theta_1'(d)}{\theta_1(d)}
\left(\frac{\theta'(\omega+d)}{\theta(\omega+d)}
+\frac{\theta'(\omega-d)}{\theta(\omega-d)}\right)
-\frac{1}{2}
\left(\frac{\theta''(\omega+d)}{\theta(\omega+d)}
-\frac{\theta''(\omega-d)}{\theta(\omega-d)}\right)
\right].
\end{equation}

We now show that $T_1(\omega)$ has the same behavior as $2\theta'(\omega)/\theta(\omega)$
under the shift $\omega\to \omega+\tau$, and therefore their difference is an elliptic function.
We obtain using \eqref{almost_elliptic}
\[
T_1(\omega+\tau)=T_1(\omega)+f(\omega),
\]
where
\begin{equation}\label{def: of f}
f(\omega)=2\pi i \gamma_0^2 u_0
\frac{\theta(0)^2\theta(\omega+d)\theta(\omega-d)}{\theta(d)^2\theta(\omega)^2}\left[ 
\frac{\theta'(\omega+d)}{\theta(\omega+d)}
-\frac{\theta'(\omega-d)}{\theta(\omega-d)}-
2\frac{\theta_1'(d)}{\theta_1(d)}\right].
\end{equation}

It is easily seen that $f(\omega)=f(\omega+\tau)=f(\omega+1)$, so that $f$ is elliptic.
Furthermore, at the zero $(1+\tau)/2$ of $\theta(\omega)$, by \eqref{theta'/theta},
\[
\frac{\theta'(\omega+d)}{\theta(\omega+d)}
-\frac{\theta'(\omega-d)}{\theta(\omega-d)}=
\frac{\theta_1'(d+\nu)}{\theta_1(d+\nu)}
+\frac{\theta_1'(d-\nu)}{\theta_1(d-\nu)}=2\frac{\theta_1'(d)}{\theta_1(d)}+\mathcal{O}(\nu),
\qquad \nu=\omega-\frac{1+\tau}{2},
\]
and thus the expression in the square brackets in \eqref{def: of f} vanishes
as $\omega\to (1+\tau)/2$. So the pole of $f$ at $(1+\tau)/2$ cannot have the order larger then 1.
Thus $f$ is an elliptic function with at most single first-order pole
modulo the lattice, which means $f$ is a constant. At $\omega=0$,
\[
f(0)=4\pi i \gamma_0^2 u_0 \left(\frac{\theta'(d)}{\theta(d)}-\frac{\theta_1'(d)}{\theta_1(d)}\right)=-4\pi i
\]
by \eqref{id2} of Lemma \ref{ThmThetaids}.
Thus
\[
f(\omega)\equiv -4\pi i.
\]
This immediately implies that the function
\[
T_1(\omega)-2\frac{\theta'(\omega)}{\theta(\omega)}
\]
is elliptic. To analyze its behavior at the pole, it is convenient to write
$T_1$ in terms of $\theta_1$ by \eqref{theta'/theta}, \eqref{theta''/theta}
with $\nu=\omega-\frac{1+\tau}{2}$:
\begin{equation}
\begin{aligned}
T_1(\omega)=&-\gamma_0^2 u_0
\frac{\theta(0)^2\theta_1(d+\nu)\theta_1(d-\nu)}{\theta(d)^2\theta_1(\nu)^2}\left[ \frac{\theta_1'(d)}{\theta_1(d)}
\left(\frac{\theta_1'(d+\nu)}{\theta_1(d+\nu)}
-\frac{\theta_1'(d-\nu)}{\theta_1(d-\nu)}\right)\right.\\
&\left.
-\frac{1}{2}
\left(\frac{\theta_1''(d+\nu)}{\theta_1(d+\nu)}
-\frac{\theta_1''(d-\nu)}{\theta_1(d-\nu)}\right)
-2\pi i \frac{\theta_1'(d)}{\theta_1(d)} +\pi i 
\left(\frac{\theta_1'(d+\nu)}{\theta_1(d+\nu)}
+\frac{\theta_1'(d-\nu)}{\theta_1(d-\nu)}\right)
\right].
\end{aligned}
\end{equation}
It is obvious from this representation that the expression in the square brackets vanishes
at $\nu=0$, and therefore the order of the pole of $T_1$ at $\nu=0$ is no larger than 1.
Since the same is true for $\theta'(\omega)/\theta(\omega)=\theta_1'(\nu)/\theta_1(\nu)-i\pi$, 
\[
T_1(\omega)-2\frac{\theta'(\omega)}{\theta(\omega)}\equiv\mathrm{const}.
\]
The value of this constant is easy to obtain by setting $\omega=0$:
since both $T_1(0)=0$ (see \eqref{T1mid}) and $\theta'(0)=0$, this value is 0, which proves Part (c).

\end{proof}

\section{The constant. Proof of Lemma \ref{constantlemma}}\label{SecConst}
Recalling \eqref{def:W}, we write the term with $W$ in \eqref{DD}
\begin{equation}\label{T2T3}
\frac{i \zeta_0\gamma_0^2}{4}\int_0^1W(\omega)d\omega=\frac{i \zeta_0\gamma_0^2}{4}\int_0^1\left( T_2(\omega)+
T_3(\omega)\right)d\omega,
\end{equation}
where 
\begin{equation}\label{defT2}
\begin{aligned}
 T_2&=\begin{pmatrix}
im_{22,0}&m_{11,0}
\end{pmatrix}\sum_{p\in \{-1,v_1,1\}}\int_{\partial U^{(p)}}\frac{ s\Delta_1(z) dz}{2\pi i(z-v_2)^2}\begin{pmatrix}
m_{11,0}\\ -im_{22,0}
\end{pmatrix}, 
\\ T_3&=\begin{pmatrix}
im_{22,0}&m_{11,0}
\end{pmatrix}\int_{\partial U^{(v_2)}}\frac{ s\Delta_1(z) dz}{2\pi i(z-v_2)^2}\begin{pmatrix}
m_{11,0}\\ -im_{22,0}
\end{pmatrix},\end{aligned}
\end{equation}
with the integrals traversed \textit{clockwise}.

In this section we show (in subsection \ref{91}) that
\begin{equation}\label{avgT2}
 \frac{i\gamma_0^2\zeta_0}{4}\int_0^1 T_2(\omega)d\omega=\frac{1}{8}\sum_{y\in\{-1,v_1,1\}}
 \frac{\partial}{\partial v_2}\left[-\log|q(y)|+\frac{1}{2}\log|(y-v_2)|\right],
 \end{equation}
and (in subsection \ref{92}) that
 \begin{multline}\label{SolnT3}
  \frac{i\gamma_0^2\zeta_0}{4}\int_0^1T_3(\omega)d\omega-\frac{\partial \tau}{\partial v_2}\int_0^1
  \frac{\partial }{\partial \tau} \log \theta(\omega;\tau)d\omega
  \\=-\left(\frac{1}{16}\frac{\partial}{\partial v_2}
 \log \left[(1-v_2^2)(v_2-v_1)\right]+\frac{1}{2}\frac{\partial}{\partial v_2}\log I_0+\frac{1}{8}\frac{\partial}{\partial v_2}\log |q(v_2)|\right).
  \end{multline}

Substituting the last 2 equations into  \eqref{T2T3}, we prove Lemma \ref{constantlemma}.

\subsection{Evaluation of $T_2$}\label{91}
Our goal in this section is to obtain \eqref{avgT2}. We first compute $T_2(\omega)$.
By the definition of $\mathcal N$ in \eqref{def:calN} and by  \eqref{uz0} and \eqref{th124}, with $\omega=\pi\Omega$,
\begin{equation}\begin{aligned}\mathcal N(z)e^{-i\pi\omega\sigma_3}&=\frac{\gamma(z)\theta_3}{2\theta_3(\omega)}\begin{pmatrix}\frac{\theta_1(\omega+d)}{\theta_1(d)}&i\frac{\theta_1(\omega+d)}{\theta_1(d)}\\
-i\frac{\theta_1(d-\omega)}{\theta_1(d)}&\frac{\theta_1(d-\omega)}{\theta_1(d)}\end{pmatrix}+o(1),\qquad z\to -1\\
\mathcal N(z)e^{-i\pi\omega\sigma_3}&=\frac{\gamma(z)^{-1}\theta_3}{2\theta_3(\omega)}\begin{pmatrix}\frac{\theta_2(\omega+d)}{\theta_2(d)}&-i\frac{\theta_2(\omega+d)}{\theta_2(d)}\\
i\frac{\theta_2(d-\omega)}{\theta_2(d)}&\frac{\theta_2(d-\omega)}{\theta_2(d)}\end{pmatrix}+o(1),\qquad z\to v_1\\
\mathcal N(z)&=\frac{\gamma(z)^{-1}\theta_3}{2\theta_3(\omega)}\begin{pmatrix}\frac{\theta_4(\omega+d)}{\theta_4(d)}&-i\frac{\theta_4(\omega+d)}{\theta_4(d)}\\
i\frac{\theta_4(d-\omega)}{\theta_4(d)}&\frac{\theta_4(d-\omega)}{\theta_4(d)}\end{pmatrix}+o(1),\qquad z\to 1
\end{aligned}\end{equation}
(Note that $\theta_j(d)\neq 0$, $j=1,2,3,4$, by the argument following \eqref{def:d}. Moreover, $\theta_3(\omega)\neq 0$ for $\omega\in\mathbb R$.)
Thus, by \eqref{pointsphi}, \eqref{Delta1}, \eqref{Delta11}, and the definition of $m_{jj,0}$ in \eqref{formulamjj}, a straightforward calculation yields
\begin{equation}\label{resDelta1}
\begin{pmatrix}
im_{22,0}&m_{11,0}
\end{pmatrix}s\Delta_1(z)
\begin{pmatrix}
m_{11,0}\\ -im_{22,0}
\end{pmatrix}
=\begin{cases}\pm\frac{is\gamma(z)^2}{16\sqrt{\zeta(z)}}F_1(\omega)+o(1),&\textrm{ as }z\to -1,\\
\mp\frac{is\gamma(z)^{-2}}{16\sqrt{\zeta(z)}}F_2(\omega)+o(1),&\textrm{ as }z\to v_1,\\
\mp\frac{is\gamma(z)^{-2}}{16\sqrt{\zeta(z)}}F_4(\omega)+o(1),&\textrm{ as }z\to 1,
\end{cases}
\end{equation}
 where the upper sign is taken if $\Im z >0$, the lower if $\Im z <0$, and
$F_j$ is given by
\begin{equation}\label{defFj}
F_j(\omega)=
\frac{\theta_3^4\left[
\theta_j(\omega+d)\theta_3(\omega-d)+\theta_j(\omega-d)\theta_3(\omega+d)\right]^2}{\theta_3(\omega)^4\theta_3(d)^2\theta_j(d)^2},\qquad j=1,2,4.
\end{equation}

To compute the residue of \eqref{resDelta1}
at $-1$, we need to analyze $\pm\frac{\gamma(z)^2}{\sqrt{\zeta(z)}}$ at $-1$. It is meromorphic, and we need to determine
the sign of its residue (the absolute value follows straightforwardly from the expansions of $\gamma(z)$ and $\zeta(z)$).  
Let $x\in U^{(-1)}$, with $x=-1+\epsilon$, $\epsilon>0$, and $x$ lying on the positive side of the cut. For such $x$, 
$\gamma(x)^2=i\left|\gamma(x)^2\right|$ by \eqref{def:gamma}, and  by the expansion \eqref{expzeta2}, and in particular the fact that $\widetilde \zeta_0$ is positive, we have that $\sqrt{\zeta(x)}$ is positive. Thus
 $\frac{\gamma(x)^2}{\sqrt{\zeta(x)}}=i\left|\frac{\gamma(x)^2}{\sqrt{\zeta(x)}}\right|$, and by \eqref{def:gamma} and \eqref{expzeta2}, 
\begin{equation}\label{resDelta2}
\pm\frac{is\gamma(z)^2}{\sqrt{\zeta(z)}}=-\frac{1}{z+1}\frac{1+v_1}{|q(-1)|}(1+\varepsilon_1(z))
\end{equation}
in a neighborhood of $-1$, where $\varepsilon_1(z)$ is an analytic function uniformly $o(1)$ as $z\to -1$.

Similar analysis in the neighborhoods $U^{(v_1)}$, $U^{(1)}$ yields
\begin{equation}\label{resDelta3}
\mp \frac{is\gamma(z)^{-2}}{\sqrt{\zeta(z)}}=\begin{cases}
\frac{1}{z-v_1}\frac{(v_2-v_1)(1+v_1)}{2 |q(v_1)|}(1+\varepsilon_2(z))&\textrm{for } z\in U^{(v_1)},\\
\frac{1}{z-1}\frac{1-v_2}{|q(1)|}(1+\varepsilon_4(z))&\textrm{for } z\in U^{(1)},\end{cases}
\end{equation}
where $\varepsilon_2(z)$,  $\varepsilon_4(z)$ are analytic function uniformly $o(1)$ as $z\to v_1,1$, respectively.

Thus by the definition of $T_2$ in \eqref{defT2}, computing the residue by \eqref{resDelta1} (note the negative orientation of the contours), we obtain
\begin{equation}\label{formulaT2}
 T_2=\frac{1+v_1}{16 (1+v_2)^2|q(-1)|}F_1(\omega)
-\frac{1+v_1}{32 (v_2-v_1)|q(v_1)|}F_2(\omega)-
\frac{1}{16 (1-v_2)|q(1)|}F_4(\omega).
\end{equation}

We now evaluate $\int_0^1F_j(\omega)d\omega$. 
It is easily seen that $F_j(\omega)$, $j=1,2,4$, are elliptic functions.
We start with $F_1$. 
Note first that since $\theta_3$ is even and $\theta_1$ is odd, we have $F_1(0)=0$. By the definition of $\theta_1$ and $\theta_3$, we have
\begin{equation}
\frac{\left(\theta_3(\omega-d)\theta_1(\omega+d)+\theta_3(\omega+d)\theta_1(\omega-d)\right)^2}{\theta_3(\omega)^4}=
-\frac{\left(\theta_1(\nu-d)\theta_3(\nu+d)+\theta_1(\nu+d)\theta_3(\nu-d)\right)^2}{\theta_1(\nu)^4},
\end{equation}
where $\nu=\omega-\frac{1+\tau}{2}$. Thus, as $\nu\to 0$, the r.h.s. of this equation becomes
\begin{equation}
-4\frac{\left(\theta_1'(d)\theta_3(d)-\theta_1(d)\theta_3'(d)\right)^2}{\left(\theta_1'\right)^4\nu^2}+\mathcal O\left(\nu^{-1}\right).
\end{equation}
Thus we can apply Lemma \ref{LemmaComp} in Appendix \ref{App1} to $F_1$, which gives
\begin{equation}
\int_0^1F_1(\omega)d\omega=-4\left(\frac{\theta_3}{\theta_1'}\right)^4\frac{\theta_3''}{\theta_3}\left(\frac{\theta_1'(d)}{\theta_1(d)}-\frac{\theta_3'(d)}{\theta_3(d)}\right)^2.
\end{equation}
Using here the identity \eqref{id2} of Lemma \ref{ThmThetaids}, and then the equation $\theta_1'=\pi\theta_2\theta_3\theta_4$, we finally obtain
\begin{equation}\label{intF1}
\int_0^1F_1(\omega)d\omega=4\left(\frac{\theta_3}{\theta_1'}\right)^4\frac{\theta_3''}{\theta_3}I_0^2(1+v_2)^2=4\left(\frac{1}{\pi\theta_2\theta_4}\right)^4\frac{\theta_3''}{\theta_3}I_0^2(1+v_2)^2.
\end{equation}

We now evaluate the integrals of $F_2$ and $F_4$.
Applying  the summation formulae \eqref{theta23} and \eqref{theta34} to the definition of  $F_2$ and $F_4$, respectively, in \eqref{defFj}, we obtain
\begin{equation}\label{F2F4}
F_2(\omega)=\frac{4\theta_3^2}{\theta_2^2}\frac{\theta_2(\omega)^2}{\theta_3(\omega)^2},\qquad
F_4(\omega)=\frac{4\theta_3^2}{\theta_4^2}\frac{\theta_4(\omega)^2}{\theta_3(\omega)^2}.
\end{equation}
 By the definitions of $\theta_j$ for $j=1,2,3,4$, we have with $\nu=\omega-\frac{1+\tau}{2}$
\begin{equation}\begin{aligned}
\frac{\theta_2(\omega)^2}{\theta_3(\omega)^2}&=\frac{\theta_4(\nu)^2}{\theta_1(\nu)^2}=
\frac{\theta_4^2}{\left(\theta_1'\right)^2}\nu^{-2}+\mathcal O\left(\nu^{-1}\right),\qquad \nu\to 0,\\
\frac{\theta_4(\omega)^2}{\theta_3(\omega)^2}&=-\frac{\theta_2(\nu)^2}{\theta_1(\nu)^2}
=-\frac{\theta_2^2}{\left(\theta_1'\right)^2}\nu^{-2}+\mathcal O\left(\nu^{-1}\right),\qquad \nu\to 0; \\
\end{aligned}\end{equation}
and applying Lemma \ref{LemmaComp},  we obtain
\begin{equation}
\int_0^1\frac{\theta_2(\omega)^2}{\theta_3(\omega)^2}d\omega=\frac{\theta_4^2}{\left(\theta_1'\right)^2}\frac{\theta_3''}{\theta_3}+\frac{\theta_2^2}{\theta_3^2},\qquad
\int_0^1\frac{\theta_4(\omega)^2}{\theta_3(\omega)^2}d\omega=-\frac{\theta_2^2}{\left(\theta_1'\right)^2}\frac{\theta_3''}{\theta_3}+\frac{\theta_4^2}{\theta_3^2}.
\end{equation}
From here, by \eqref{F2F4} and the equation $\theta_1'=\pi\theta_2\theta_3\theta_4$,
\begin{equation}\label{intF2F4}
\int_0^1F_2(\omega)d\omega=4\left(\frac{1}{\pi^2\theta_2^4}\frac{\theta_3''}{\theta_3}+1\right),\qquad
\int_0^1F_4(\omega)d\omega=4\left(-\frac{1}{\pi^2\theta_4^4}\frac{\theta_3''}{\theta_3}+1\right).
\end{equation}

Integrating \eqref{formulaT2} by \eqref{intF1}, \eqref{intF2F4}, we obtain
\begin{multline}\label{formulaT22}
 \int_0^1T_2(\omega)d\omega=\frac{(1+v_1)I_0^2}{4|q(-1)|}\left(\frac{1}{\pi\theta_2\theta_4}\right)^4\frac{\theta_3''}{\theta_3}
-\frac{1+v_1}{8 (v_2-v_1)|q(v_1)|}\left(\frac{1}{\pi^2\theta_2^4}\frac{\theta_3''}{\theta_3}+1\right)\\-
\frac{1}{4 (1-v_2)|q(1)|}\left(-\frac{1}{\pi^2\theta_4^4}\frac{\theta_3''}{\theta_3}+1\right).
\end{multline}

We now express all the $\theta$-constants here in terms of elliptic integrals. For $\theta_2^4$, $\theta_2^4$,
this was already done in \eqref{idth2}, \eqref{idth4} of Lemma \ref{ThmThetaids}.
To obtain an expression for $\frac{\theta_3''}{\theta_3}$, we first note that by the differential equation \eqref{thdiff} satisfied by $\theta$-functions, and then by \eqref{derivativetau} and \eqref{idth3},
\begin{equation}
\frac{\theta_3''}{\theta_3}=4\pi i \frac{\partial}{\partial \tau}\log \theta_3=\pi i \frac{1}{\partial \tau/\partial v_2}\frac{\partial}{\partial v_2}\log \theta_3^4=I_0^2(1-v_2)^2(v_2-v_1)\frac{\partial}{\partial v_2}\log (I_0^2(1+v_2)).
\end{equation}
We now use \eqref{Partial5}, \eqref{x1x2eqn2}, and then the expression 
$|q(v_2)|=(v_2-x_1)(x_2-v_2)=-x_1x_2-v_2(v_2-v_1)/2$,
 to obtain from here
\begin{equation}\label{dtheta33}
\frac{\theta_3''}{\theta_3}=2I_0^2\left(x_1x_2+\frac{v_2-v_1}{2}\right)=
2I_0^2\left(-|q(v_2)|+\frac{(1-v_2)(v_2-v_1)}{2}\right).
\end{equation}
Substituting this expression as well as \eqref{idth2}, \eqref{idth4} into \eqref{formulaT22},
and using the fact that by \eqref{expgamma}, \eqref{expzeta},
\[
\frac{i\gamma_0^2\zeta_0}{4}=-\frac{|q(v_2)|}{2(1+v_2)},
\]
we obtain
\begin{multline}\label{formulaT23}
 \frac{i\gamma_0^2\zeta_0}{4}\int_0^1T_2(\omega)d\omega=
 \frac{1}{8}\frac{q(v_2)^2}{(1-v_2^2)(v_2-v_1)}\left(\frac{1}{|q(-1)|}-\frac{1}{|q(v_1)|}+\frac{1}{|q(1)|}\right)\\
 +\frac{1}{16}\frac{|q(v_2)|}{(1-v_2^2)(v_2-v_1)}\left(-\frac{(1-v_2)(v_2-v_1)}{|q(-1)|}+\frac{1-v_2^2}{|q(v_1)|}+
 \frac{(1+v_2)(v_2-v_1)}{|q(1)|}\right).
\end{multline}
In the last 3 terms here, we express $|q(v_2)|$ by $|q(-1)|$, $|q(v_1)|$, $|q(1)|$, respectively, e.g. for the last term
we write (recall \eqref{x1x2eqn1})
\begin{equation}\label{qv2}
|q(v_2)|=-x_1x_2-v_2(v_2-v_1)/2=-|q(1)|+1-(v_1+v_2)/2-v_2(v_2-v_1)/2.
\end{equation}
This allows us to write \eqref{formulaT23} in the form
\begin{multline}\label{formulaT24}
 \frac{i\gamma_0^2\zeta_0}{4}\int_0^1T_2(\omega)d\omega=
 \frac{1}{8}\frac{q(v_2)^2}{(1-v_2^2)(v_2-v_1)}\left(\frac{1}{|q(-1)|}-\frac{1}{|q(v_1)|}+\frac{1}{|q(1)|}\right)\\
 +\frac{1}{16}\left(\frac{1}{1+v_2}+\frac{1}{v_2-v_1}-\frac{1}{1-v_2}
 -\frac{2-(v_2-v_1)}{2|q(-1)|}-\frac{v_2+v_1}{2|q(v_1)|}+
 \frac{2+v_2-v_1}{2|q(1)|}\right).
\end{multline}

On the other hand, by \eqref{x1x2eqn2} and \eqref{derivativex1x2},
\begin{equation}
\frac{\partial}{\partial v_2}|q(-1)|=\frac{\partial}{\partial v_2}\left(1+x_1x_2+\frac{v_1+v_2}{2}\right)=
-\frac{v_2-v_1}{4}-\frac{q(v_2)^2}{(1-v_2^2)(v_2-v_1)}+\frac{1}{2},
\end{equation}
and
\begin{equation}\label{qv11}
\begin{aligned}
\frac{\partial}{\partial v_2}|q(v_1)|&=
\frac{v_2-v_1}{4}+\frac{q(v_2)^2}{(1-v_2^2)(v_2-v_1)}+\frac{v_1}{2},\\
\frac{\partial}{\partial v_2}|q(1)|&=
-\frac{v_2-v_1}{4}-\frac{q(v_2)^2}{(1-v_2^2)(v_2-v_1)}-\frac{1}{2}.
\end{aligned}
\end{equation}
We therefore easily obtain the expression for $\frac{\partial}{\partial v_2}\log|q(-1)q(v_1)q(1)|$. Comparing it 
with \eqref{formulaT24} shows \eqref{avgT2}.

  \subsection{Evaluation of $T_3$}\label{92}
 Now consider $T_3$. Our aim in this section is to prove \eqref{SolnT3}.
We write  $\mathcal N$ in \eqref{def:calN} in the form
\begin{equation}\label{defB}\begin{aligned}
\mathcal N(z)&=A(z;s\Omega)+B(z;s\Omega),\qquad
A=\frac{1}{2}\begin{pmatrix}
A_1&iA_1\\
-iA_2&A_2
\end{pmatrix}, \qquad 
B=\frac{1}{2}\begin{pmatrix}
B_1&-iB_1\\
iB_2&B_2
\end{pmatrix},\\
A_j(z;\omega)&=\frac{\theta_3}{2\theta_3(\omega)}\left[\left(\gamma(z)+\gamma(z)^{-1}\right)\frac{\theta_3(u(z)\pm \omega+d)}{\theta_3(u(z)+d)}
+\left(\gamma(z)-\gamma(z)^{-1}\right)\frac{\theta_3(-u(z)\pm \omega+d)}{\theta_3(-u(z)+d)}\right],\\
B_j(z;\omega)&=\frac{\theta_3}{2\theta_3(\omega)}\left[\left(\gamma(z)+\gamma(z)^{-1}\right)\frac{\theta_3(u(z)\pm \omega+d)}{\theta_3(u(z)+d)}
-\left(\gamma(z)-\gamma(z)^{-1}\right)\frac{\theta_3(-u(z)\pm \omega+d)}{\theta_3(-u(z)+d)}\right],
\end{aligned}
\end{equation}
where $\pm$ means $+$ for $j=1$ and $-$ for $j=2$.
Using the jump conditions \eqref{jumpsgamma}, \eqref{jumpsu}, we observe that
 $(z-v_2)^{1/4}A(z)$ and $(z-v_2)^{-1/4}B(z)$ are analytic on $U^{(v_2)}$. 

Since $\Delta_1(z)$ in \eqref{Delta1} for $p=v_2$ is meromorphic on $U^{(v_2)}$, all 
odd powers of roots $(z-v_2)^{1/2}$ in the expansion of 
\eqref{Delta1} cancel, and it follows that for $z\in U^{(v_2)}$ and $\Im z>0$,
\begin{equation}\begin{aligned}
\Delta_1(z)=-\frac{1}{32\sqrt{\zeta(z)}}\left[
\begin{pmatrix}
A_1&iA_1\\
-iA_2&A_2
\end{pmatrix}
\begin{pmatrix}
-1&- 2i\\ -2i &1
\end{pmatrix} 
\begin{pmatrix}
A_2&-iA_1\\iA_2&A_1
\end{pmatrix}\right.\\
\left.
+
\begin{pmatrix}
B_1&-iB_1\\
iB_2&B_2
\end{pmatrix}
\begin{pmatrix}
-1&-2i\\ - 2i &1
\end{pmatrix} 
\begin{pmatrix}
B_2&iB_1\\-iB_2&B_1
\end{pmatrix}
\right].\end{aligned}
\end{equation}
Therefore
\begin{equation}\nonumber
\begin{pmatrix}
im_{22,0}&m_{11,0}
\end{pmatrix}
\Delta_1(z)
\begin{pmatrix}
m_{11,0}\\-im_{22,0}
\end{pmatrix}=\frac{i}{16\sqrt{\zeta(z)}}\left[
(m_{22,0}A_1-m_{11,0}A_2)^2+3(m_{22,0}B_1+m_{11,0}B_2)^2
\right].\end{equation}
Expanding $A_1(z)$ and $A_2(z)$ as $z\to v_2$, we obtain using \eqref{expzeta}, \eqref{expgamma}, and \eqref{idsmij},
\begin{equation}m_{22,0}A_1(z)-m_{11,0}A_2(z)=-\gamma_0^{-1}u_0 T_1(\omega)(z-v_2)^{3/4}+\mathcal O\left((z-v_2)^{5/4}\right)
\end{equation}
with $T_1(\omega)$ as defined in \eqref{def:T1}. By \eqref{T1th} 
in Proposition \ref{Propterm2}, this equals
\begin{equation}-\frac{2u_0}{\gamma_0}\frac{\theta_3'(\omega)}{\theta_3(\omega)}(z-v_2)^{3/4}+\mathcal O\left((z-v_2)^{5/4}\right).
\end{equation}
%\begin{multline}\label{formulaC1}
%C_1=m_{11,0}^2m_{22,0}^2(\gamma_0^2\Gamma_2+\Gamma_1)^2
%-24\gamma_1
%\\-12m_{11,0}m_{22,0}
%\Bigg(\Sigma_2+\gamma_0^2 \frac{m_{11,2}m_{22,1}+m_{22,2}m_{11,1}}{m_{11,0}m_{22,0}}+2\frac{m_{11,1}m_{22,1}}{m_{11,0}m_{22,0}}+\gamma_0^{-2}\Sigma_1
%\Bigg).
%\end{multline}
So that by the definition of $T_3$ in \eqref{defT2}, we obtain computing the residue for the first term,
\begin{equation}\label{formulaT31}
T_3(\omega)=-\frac{iu_0^2}{4\gamma_0^2\zeta_0}\left(\frac{\theta_3'(\omega)}{\theta_3(\omega)}\right)^2+\int_{\partial U^{(v_2)}}\frac{3i[m_{22,0}B_1(z)+m_{11,0}B_2(z)]^2}{16(z-v_2)^2\sqrt{\zeta(z)}}\frac{dz}{2\pi i},
\end{equation}
where the integration is in the negative direction around $v_2$, and where $\sqrt{\zeta} $ and $B_1,B_2$ are understood to be the analytic continuation from $\Im z>0$.  

We now write the average
\begin{equation}  \label{SolnT32}
\frac{i\gamma_0^2\zeta_0}{4}\int_0^1T_3(\omega)d\omega
  =\frac{u_0^2}{16}\int_0^1\left(\frac{\theta_3'(\omega)}{\theta_3(\omega)}\right)^2d\omega+\frac{i\gamma_0^2\zeta_0Q}{4},
  \end{equation}
  where
  \begin{equation}\label{def:Q}
Q=\int_0^1d\omega \int_{\partial U^{(v_2)}}\frac{3is[m_{22,0}B_1(z;\omega)+m_{11,0}B_2(z;\omega)]^2}{16(z-v_2)^2\sqrt{\zeta(z)}}\frac{dz}{2\pi i}.
\end{equation}

To compare with Lemma \ref{subleading}, we will now single out a contribution from
\begin{equation}
\delta=\frac{\partial \tau}{\partial v_2}\int_0^1 \frac{\partial }{\partial \tau}\log\theta_3(\omega)d\omega.
\end{equation}
Using the differential equation \eqref{thdiff} and the fact that
\[
0=\int_0^1\left(\frac{\theta_3'(\omega)}{\theta_3(\omega)}\right)'d\omega=\int_0^1\left[ \frac{\theta_3''(\omega)}{\theta_3(\omega)}-\left(\frac{\theta_3'(\omega)}{\theta_3(\omega)}\right)^2
\right]
d\omega,
\]
we can write
\begin{equation}
\delta=i\frac{\partial |\tau|}{\partial v_2}\int_0^1 \frac{\theta_3''(\omega) }{\theta_3(\omega)}\frac{d\omega}{4\pi i}=
\frac{\partial |\tau|}{\partial v_2}\int_0^1 \left(\frac{\theta_3'(\omega) }{\theta_3(\omega)}\right)^2\frac{d\omega}{4\pi}.
\end{equation}
Since, by \eqref{derivativetau}, $\pi u_0^2=\frac{\partial |\tau|}{\partial v_2}$, we can rewrite \eqref{SolnT32}
in the form
\begin{equation}  \label{SolnT33}
\frac{i\gamma_0^2\zeta_0}{4}\int_0^1T_3(\omega)d\omega
  =-\frac{3}{16\pi}\frac{\partial |\tau|}{\partial v_2}
  \int_0^1\left(\frac{\theta_3'(\omega)}{\theta_3(\omega)}\right)^2d\omega+\frac{i\gamma_0^2\zeta_0Q}{4}+\delta.
  \end{equation}

Now by \eqref{intlogder4},
\[
\int_0^1 \left(\frac{\theta_3'(\omega)}{\theta_3(\omega)}\right)^2 d\omega=\frac{\pi^2}{3}+\frac{\theta_1'''}{3\theta_1'}.
\]
Using the identity $\theta_1'=\pi\theta_2\theta_3\theta_4$, and the identities \eqref{idth4}--\eqref{idth3}
of Lemma \ref{ThmThetaids}, we write here
\begin{equation}
\frac{\theta_1'''}{\theta_1'}=4\pi i \frac{\partial}{\partial \tau} \log \left(\theta_1'\right)=\frac{\pi i}{
\frac{\partial \tau}{\partial v_2}}\frac{\partial}{\partial v_2}\log \left(\theta_1'\right)^4=
\frac{\pi}{
\frac{\partial |\tau|}{\partial v_2}}\frac{\partial}{\partial v_2}\log \left(\theta_2\theta_3\theta_4\right)^4=
\frac{\pi}{
\frac{\partial |\tau|}{\partial v_2}}
\frac{\partial}{\partial v_2}\log\left[I_0^6(1-v_2^2)(v_2-v_1)\right],
\end{equation}
so that we can rewrite \eqref{SolnT33} in the form
\begin{equation}  \label{SolnT34}
\frac{i\gamma_0^2\zeta_0}{4}\int_0^1T_3(\omega)d\omega
  =-\frac{1}{16}\left(
  \pi \frac{\partial |\tau|}{\partial v_2}+
  \frac{\partial}{\partial v_2}\log\left[I_0^6(1-v_2^2)(v_2-v_1)\right]\right)
  +\frac{i\gamma_0^2\zeta_0Q}{4}+\delta.
  \end{equation}

It remains to evaluate $Q$ defined in \eqref{def:Q}. To simplify the computations, 
we first do the averaging over $\omega$ and only then compute the residue in this case.

As above for $A(z)$, we expand $B(z)$ to obtain
\begin{equation}
\gamma(z)[m_{22,0}B_1(z;\omega)+m_{11,0}B_2(z;\omega)]- 2 =\mathcal O(z-v_2),\qquad z\to v_2,
\end{equation}
and therefore
\begin{equation}
\gamma(z)^2[m_{22,0}B_1(z;\omega)+m_{11,0}B_2(z;\omega)]^2=-4+4\gamma(z)[m_{22,0}B_1(z;\omega)+m_{11,0}B_2(z;\omega)]+\mathcal O\left((z-v_2)^2\right),
\end{equation}
as $z\to v_2$. Thus, upon changing the order of integration,
\begin{equation}\label{ChangeT5}
Q=\int_{\partial U^{(v_2)}}\frac{dz}{2\pi i} \frac{3i s}{4\gamma^{2}(z)(z-v_2)^2\sqrt{\zeta(z)}}\left[-1+\gamma(z)\int_0^1d\omega\left[m_{22,0}B_1(z;\omega)+m_{11,0}B_2(z;\omega)\right]\right].
\end{equation}
By the definition of $B_1$ and $B_2$ in \eqref{defB} and the formula for $m_{11,0}$ and $m_{22,0}$ in \eqref{formulamjj}, we have
\begin{equation}\label{formulaT37}
\int_0^1d\omega\left[m_{22,0}B_1(z;\omega)+m_{11,0}B_2(z;\omega)\right]=\int_0^1
\left(\widetilde q(\omega)+\widetilde q(-\omega)\right)d\omega,
\end{equation}
where
\begin{multline}\label{formulaT35}
\widetilde q(\omega)=\frac{\theta_3^2\theta(-\omega+d)}{2\theta(d)\theta(\omega)^2}\Bigg(\left(\gamma(z)+\gamma(z)^{-1}\right)\frac{\theta(u(z)+\omega+d)}{\theta(u(z)+d)}\\-\left(\gamma(z)-\gamma(z)^{-1}\right)\frac{\theta(-u(z)+\omega+d)}{\theta(-u(z)+d)}\Bigg).\end{multline}
Since $\widetilde q(-\omega)=\widetilde q(1-\omega)$, we have that
\begin{equation}\label{formulaT36}
\int_0^1\widetilde q(-\omega)d\omega=\int_0^1\widetilde q(\omega)d\omega,
\end{equation}

Applying \eqref{thint2} to evaluate $\int_0^1\widetilde q(\omega)d\omega$, we obtain:
\begin{multline}\label{formulaT38}
\gamma(z)\int_0^1d\omega\left[m_{22,0}B_1(z;\omega)+m_{11,0}B_2(z;\omega)\right]=
\frac{\pi \theta_3^2 g(d)}{\left(\theta_1'\right)^2\sin(\pi u)}\\
\times\left\{ \left(\gamma(z)^2+1\right) 
g(d+u)[f(d)-f(d+u)]
+\left(\gamma(z)^2-1\right)g(d-u)[f(d)-f(d-u)]
\right\},
\end{multline}
where
\begin{equation}\label{formulaT39}
g(x)=\frac{\theta_1(x)}{\theta_3(x)},\qquad
f(x)=\frac{\theta_1'(x)}{\theta_1(x)}.
\end{equation}
Note that \eqref{ellfuns5} gives for the derivative of $f(z)$
\begin{equation}
f'(x)=-\left(\frac{\theta_1'}{\theta_3}\right)^2\frac{1}{g(x)^2}+\frac{\theta_3''}{\theta_3}.
\end{equation} 
Using this, we have, in particular, as $z\to v_2$, i.e. $u\to 0$,
\begin{multline}\label{formulaT39a}
 f(d)-f(d\pm u)=\left(\frac{\theta_1'}{\theta_3}\right)^2\Bigg[\pm u\left(\frac{1}{g(d)^2}-\frac{\theta_3''}{\theta_3}\left(\frac{\theta_3}{\theta_1'}\right)^2
 \right)-u^2\frac{g'(d)}{g(d)^3}\\ \pm \frac{u^3}{3}\left(-\frac{g''(d)}{g(d)^3}+3\frac{g'(d)^2}{g(d)^4}\right)+\frac{u^4}{12}\left(
-\frac{g'''(d)}{g(d)^3}+\frac{9g''(d)g'(d)}{g(d)^4}-\frac{12g'(d)^3}{g(d)^5}\right)\Bigg]+\mathcal O \left(u^5\right).
\end{multline}
Expanding also the other terms in \eqref{formulaT38}, and also expanding $u$ by \eqref{expgamma},
we  obtain that, as $z\to v_2$,
\begin{multline}\label{T39c}
\gamma(z)\int_0^1d\omega\left[m_{22,0}B_1(z;\omega)+m_{11,0}B_2(z;\omega)\right]=
g(d)\left(1+\frac{\pi^2}{6}u_0^2(z-v_2)+ \mathcal O((z-v_2)^2)\right)\\ \times\left[H_0+u_0\gamma_0^2(1+(z-v_2)(u_1+2\gamma_1))H_1+(z-v_2)(u_0^2H_2-u_0^3\gamma_0^2H_3)+\mathcal O((z-v_2)^{3/2})\right],
\end{multline}
 where
\begin{equation}\begin{aligned}
H_0&=2g(d)\left(\frac{1}{g(d)^2}-\frac{\theta_3''}{\theta_3}\left(\frac{\theta_3}{\theta_1'}\right)^2\right),&
H_1&=2g'(d)\frac{\theta_3''}{\theta_3}\left(\frac{\theta_3}{\theta_1'}\right)^2,\\
H_2&=g''(d)\left(\frac{1}{3g(d)^2}-\frac{\theta_3''}{\theta_3}\left(\frac{\theta_3}{\theta_1'}\right)^2\right),&
H_3&=\frac{g'''(d)}{6}\left(\frac{1}{g(d)^2}-2\frac{\theta_3''}{\theta_3}\left(\frac{\theta_3}{\theta_1'}\right)^2\right)-\frac{1}{6}\frac{g''(d)g'(d)}{g(d)^3}.
\end{aligned}\end{equation}
By applying 
\eqref{id20} and \eqref{id3} 
of Proposition \ref{ThmThetaids}, we simplify the combinations of the $H_j$ as follows:
\begin{equation} 
H_0+u_0\gamma_0^2H_1=\frac{2}{g(d)},\qquad u_0^2H_2-u_0^3\gamma_0^2H_3=\frac{2\gamma_1+u_1}{g(d)}\left(1-2g(d)^2\frac{\theta_3''}{\theta_3}\left(\frac{\theta_3}{\theta_1'}\right)^2\right),
\end{equation}
which allows us to write \eqref{T39c} in the form
\begin{equation}
\gamma(z)\int_0^1d\omega\left[m_{22,0}B_1(z;\omega)+m_{11,0}B_2(z;\omega)\right]=
2+\left(\frac{\pi^2}{3}u_0^2+(2\gamma_1+u_1)\right)(z-v_2)+\mathcal O((z-v_2)^{3/2}).
\end{equation}
Substituting this into \eqref{ChangeT5} and calculating the residue, we obtain
\begin{equation} \label{FormulaQ}
Q=\frac{3i}{4\gamma_0^2\zeta_0}\left(\zeta_1-u_1-\frac{\pi^2}{3}u_0^2\right)=
\frac{3i}{4\gamma_0^2\zeta_0}\left(\zeta_1-u_1-\frac{\pi}{3}\frac{\partial|\tau|}{\partial v_2}\right).
\end{equation}

For the coefficients $\zeta_1$, $u_1$ in expansions  \eqref{expzeta}  and \eqref{expgamma}
we easily obtain:
\begin{equation}\begin{aligned}
\zeta_1&=\frac{1}{3}\frac{d}{dx}\log q(x)|_{x=v_2}-
\frac{1}{6}\frac{\partial}{\partial v_2}\log (v_2^2-1)(v_2-v_1),\\
u_1&=-\frac{1}{6}\frac{\partial}{\partial v_2}\log (v_2^2-1)(v_2-v_1),
\end{aligned} \end{equation}
so that
\[
\zeta_1-u_1=\frac{1}{3}\frac{d}{dx}\log q(x)|_{x=v_2}=\frac{2v_2-(v_1+v_2)/2}{-3|q(v_2)|}.
\]
On the other hand, by \eqref{qv2} and \eqref{qv11},
\begin{equation}
\frac{\partial}{\partial v_2}|q(v_2)|=-\frac{3}{4}v_2+\frac{v_1}{4}+\frac{q(v_2)^2}{(1-v_2^2)(v_2-v_1)},
\end{equation}
and by \eqref{Partial5}, \eqref{x1x2eqn2},
\begin{equation}
\frac{\partial}{\partial v_2}\log I_0=
-\frac{|q(v_2)|}{(1-v_2^2)(v_2-v_1)}.
\end{equation}
These equations imply
\begin{equation}
\zeta_1-u_1=\frac{2}{3}\frac{\partial}{\partial v_2}\log (|q(v_2)| I_0).
\end{equation}
Substituting this into \eqref{FormulaQ} for $Q$, and that, in turn, into \eqref{SolnT34},
we obtain \eqref{SolnT3}.

\section{Slow merging of gaps. Proof of Theorem \ref{Thmtau0}}\label{sec-tau0}

\subsection{Solution of the $\Phi$-RH problem as $v_2-v_1\to 0$.}

We consider the asymptotics of the $\Phi$-RH problem 
%
%as $\nu=\frac{v_2-v_1}{2}\to 0$ ($v_2$ and $v_1$ remaining bounded away from $-1$ and $1$) and $s\to \infty$, such that $\nu>\frac{1}{2s^{5/4}}$, with the objective of proving Lemma \ref{Lemtau0} below, which after applying Theorem \ref{Thm} with fixed $v_1$, $v_2$, 
%proves the first statement of Theorem \ref{Thmtau0}, while  the second statement \eqref{limThm1tau0}  is proven at the end of the section.
%
%
%
%First, we will show that the solution of the  $\Phi$-RH problem can be extended (with a worse error term) to the
%
in the double-scaling regime where $\nu=\frac{v_2-v_1}{2}$ can approach zero with $s\to\infty$ at a rate such that 
$2\nu>\frac{1}{s^{2-\varepsilon}}$, for any fixed $\varepsilon>0$.

Let
\begin{equation}\label{def:alphabeta}
-\alpha=1+\frac{v_2+v_1}{2}>0, \qquad \beta=1-\frac{v_2+v_1}{2}>0,\qquad \gamma=
\frac{\beta^{-1}+|\alpha|^{-1}}{8}.
\end{equation}
We need to evaluate the integrals $I_j$ in the limit $\nu\to 0$. To do this (and to make a comparison with \cite{FKduke}
easier), we first change integration variable $x=t+\frac{v_1+v_2}{2}$, which maps $(v_2,1)$ to $(\nu,\beta)$; we then
split this interval into $(\nu,\sqrt{\nu})\cup [\sqrt{\nu},\beta)$ and use a change of variable $y=t/\sqrt{\nu}$ for integration over the first one. We then obtain:\footnote{Cf. equations (278)--(280) in \cite{FKduke}.}
\begin{align}\label{I2tau0}
I_2-\frac{v_2+v_1}{2}I_1&=\sqrt{|\alpha\beta|}+\mathcal O\left(\nu^2\log \nu^{-1}\right),\\
\label{I0tau0}
I_0&=\frac{\log \left(\gamma\nu\right)^{-1}}{\sqrt{|\alpha \beta|}}+\mathcal O\left(\nu^2\log \nu^{-1}\right).
\end{align}
Hence, by \eqref{x1x2eqn2},
\begin{equation}\label{tau0x1x2}
x_1x_2=\left(-I_2+\frac{v_1+v_2}{2}I_1\right)\frac{1}{I_0}=
-\frac{|\alpha\beta|}{\log \left(\gamma\nu\right)^{-1}}+\mathcal O\left(\nu^2\right).
\end{equation}

Let the neighborhoods $U^{(v_1)}$,  $U^{(v_2)}$ have radius $\nu/3$; they will be therefore contracting as 
$\nu\to 0$.
 We now evaluate the jumps $J_S(z)$ of $S$ on the edges of the lenses $\Gamma_{\Phi,L}\cup \Gamma_{\Phi,U}$. Recall from \eqref{smallJhat} that these jumps were exponentially close to the identity, in the case where $v_1$ and $v_2$ were fixed. For $z\in\Gamma_{\Phi,L}\cup \Gamma_{\Phi,U}$ and $z$ bounded away from the points $v_1,v_2$, it is clear that the jumps are still exponentially close to the identity so that \eqref{smallJhat} holds, and we consider the case where $z\to \frac{v_1+v_2}{2}$ along the edges of the lenses. We substitute \eqref{tau0x1x2} into the definition of $\phi$ in \eqref{def:phi} and obtain (taking $u=\frac{z-v_2}{v_2-v_1}$)
 \begin{equation}
 \phi(z)=\frac{\pm i\sqrt{|\alpha\beta|}}{\log \left(\gamma\nu\right)^{-1}}
 \int_{0}^{\frac{z-v_2}{2\nu}}\frac{du}{\sqrt{u(u+1)}}\left(1+\mathcal O\left(z-\frac{v_1+v_2}{2}\right)\right),
 \end{equation}
 as $\nu \to 0$,  and $z\to \frac{v_1+v_2}{2}$. Here $'+'$ sign is taken on $\Gamma_{\Phi,U}$, and 
 $'-'$ sign is taken on $\Gamma_{\Phi,L}$, and thus
 $\Im \phi(z)<0$, $\Im \phi(z)>0$ on $\Gamma_{\Phi,L}$ and $\Gamma_{\Phi,U}$ respectively. 
Worsening somewhat the error term, we have that 
\begin{equation}\label{JS0}
J_S(z)=I+\mathcal O\left(e^{-c\sqrt{s}(|z|+1)}\right),\qquad c>0,
\end{equation} 
 as $s\to \infty$, uniformly for $2\nu>s^{-2+\varepsilon}$ and $z\in\Gamma_{\Phi,L}\cup \Gamma_{\Phi,U}$.

 Next we consider the jumps of $R$ on the boundary $\partial U^{(p)}$ for $p\in \mathcal T=\{-1,v_1,v_2,1\}$. 
 Estimating $\phi(z)$ as above but now in the definition of $\zeta$ in \eqref{def:zeta}, we obtain that as $s\to \infty$, uniformly for $2\nu>s^{-2+\varepsilon}$,
\begin{equation}\label{tau0zeta}
\frac{1}{\zeta(z)^{1/2}}= \begin{cases}
\mathcal O\left(\frac{\log \nu^{-1}}{s}\right), &\textrm{uniformly on $\partial U^{(v_1)}$ and $\partial U^{(v_2)}$,}\\
\mathcal O\left(\frac{1}{s}\right), &\textrm{uniformly on $\partial U^{(1)}$ and $\partial U^{(-1)}$.}
\end{cases}
\end{equation}
To estimate $\Delta(z)$, we need to consider $\mathcal N$. We first observe that by the definition \eqref{def:gamma}, $\gamma(z), \gamma(z)^{-1}=\mathcal O(1)$ uniformly on $\partial U^{(p)}$ for $p\in \mathcal T$ as $\nu\to0$ . 
Using \eqref{I0tau0} and a simpler expansion for $J_0$, we obtain
\begin{equation}
\tau=i\frac{J_0}{I_0}=\frac{\pi i}{\log \left(\gamma\nu\right)^{-1}}(1+\mathcal O\left(\nu^2\right)),
\qquad \nu \to 0, \label{tau0tau}
\end{equation}
and define
\begin{equation}\label{asymkappa}
\kappa=e^{-\pi i/\tau}=\left[\gamma\nu\right]^{1+\mathcal O(\nu^2)}.\end{equation}
By the inversion formula \eqref{tau1overtau} for $\theta$-functions,
\begin{equation}\label{tau0asymtheta}
\theta(z)=\frac{1}{\sqrt{-i\tau}}\sum_k e^{-\frac{i\pi}{\tau}(k-z)^2}=
\frac{\kappa^{\langle z\rangle^2}}{\sqrt{-i\tau}}\left(1+\kappa^{1-2\langle z\rangle}+\kappa^{1+2\langle z\rangle}\right)+\mathcal O\left(\frac{\kappa^{9/4}}{\sqrt{|\tau|}}\right),
\end{equation}
where
\begin{equation}
z=j+\langle z \rangle , \qquad-1/2<\langle \Re z \rangle \leq 1/2, \quad j\in \mathbb Z.
\end{equation}

We now show, in \eqref{LargeDelta} below,  that $\Delta(z)$, which enters the jump matrix for $R$, 
may be too large for certain parameter sets, which makes it necessary to modify the solution of the RH problem.
First, a simple analysis of \eqref{def:d} shows that $d\rightarrow -1/2$ as $\nu\to 0$.
On the boundary of $U^{(v_1)}$, $U^{(v_2)}$, we have  $|u(z)|\rightarrow 0$, uniformly in $z$.
Therefore, using the boundedness of $\gamma,\gamma^{-1}$ on $\partial U^{(p)}$ for $p\in \mathcal T$, and
applying \eqref{tau0asymtheta}, we have for the $11$ element of $\mathcal N$ on $\partial U^{(v_1)}$
if $\langle\omega\rangle>0$ (and thus $u(z)+d+\langle\omega\rangle=\langle u(z)+d+
\langle\omega\rangle\rangle$ for $\nu$ sufficiently small):
\begin{equation}\label{N0N}
|\mathcal N_{11}|\le C\left|\frac{\theta(0)}{\theta (\omega)}
\frac{\theta(u(z)+d+\omega)}{\theta (u(z)+d)}\right|\le
C_1\frac{1}{\kappa^{\langle \omega\rangle^2}}\frac{\kappa^{\langle \omega\rangle^2-\langle \omega\rangle+1/4}}
{\kappa^{1/4}}=C_1\kappa^{-\langle \omega\rangle}\le C_2\nu^{-\langle \omega\rangle},\qquad \omega=s\Omega,
\end{equation}
for some constants $C,C_1,C_2>0$.
Similarly, we analyze the behavior of $\mathcal N_{11}$ for $\langle\omega\rangle<0$, the behavior of other matrix elements of $\mathcal N$ on $\partial U^{(v_1)}$, as well as the behavior of $\mathcal N$ on
$\partial U^{(v_2)}$ and $\partial U^{(\pm 1)}$. We find that the estimate \eqref{N0N} is the worst
(note that, in fact, the estimates for $\mathcal N$ on $\partial U^{(\pm 1)}$ are much better), and thus
recalling \eqref{tau0zeta}, we have 
\begin{equation}\label{LargeDelta}
\begin{aligned}
\Delta(z)&=\mathcal N(z)\mathcal O\left(\frac{\log \nu^{-1}}{s}\right)\mathcal N(z)^{-1}\\
&=\mathcal O\left(\frac{\log \nu^{-1}}{s}\nu^{-2|\langle s\Omega\rangle|}\right),
\end{aligned}
\end{equation}
as $s\to \infty$ and $\nu \to 0$, for $z\in\partial U^{(v_p)}$. Thus if, for example, $\nu = \frac{1}{s}$ and $|\langle s\Omega\rangle |=1/2$ (which is a case we need to deal with since the splitting of the gap regime described in \cite{FKduke} breaks down in this limit), we cannot say that $\Delta$ is small, and 
so the corresponding jump of $R$ is not guaranteed to be close to the identity, and so we cannot claim solvability of the $R$-RH problem. However, it was shown in \cite{FKduke} for the case of the RH problem of \cite{DIZ}
that we can modify the solution to ensure solvability for the range
$2\nu>s^{-5/4}$. We now provide more details of that construction in the present case, and apply it for all values of 
$\langle s\Omega\rangle $. 

Let 
\begin{equation}
t=\langle s\Omega\rangle +k/2, \end{equation}
where $k=\pm 1$ is chosen such that $-1/2< t\le 1/2$. Consider the following function:
\begin{equation}\label{def:tildeN}
\begin{aligned}
\widetilde {\mathcal N}(z)&=\begin{pmatrix}
\frac{ \delta+ \delta^{-1}}{2} \widetilde m_{11}&
\frac{ \delta- \delta^{-1}}{2i}\widetilde m_{12}\\
- \frac{\delta-\delta^{-1}}{2i}\widetilde m_{21}&
\frac{\delta+\delta^{-1}}{2} \widetilde m_{22}
\end{pmatrix},\\
 \widetilde m(z)&=
\begin{pmatrix}\frac{\theta(u(z_-)+d')}{\theta(u(z_-)+t+d')}&0\\ 0& \frac{\theta(u(z_-)+d')}{\theta(u(z_-)-t+d')}\end{pmatrix}
\begin{pmatrix}
\frac{\theta(u(z)+t+d')}{\theta (u(z)+d')}&\frac{\theta(u(z)-t-d')}{\theta (u(z)-d')}\\
\frac{\theta(u(z)+t-d')}{\theta (u(z)-d')}&\frac{\theta(u(z)-t+d')}{\theta (u(z)+d')}
\end{pmatrix},
\end{aligned}
\end{equation}
where the constant $d'$ will be fixed later on, and
we now take 
\begin{equation}\label{def:delta}
\delta(z)=\nu^{-1/4}\left(\frac{(z-v_1)(z-v_2)}{z^2-1}\right)^{1/4},
\end{equation}
with branch cuts on $(-1,v_1)\cup (v_2,1)$, and
positive as $z\to \infty$ on the first sheet of the Riemann surface $\Sigma$.
We have
\[
\delta(z)_+=\begin{cases} i\delta(z)_-&\mathrm{on}\quad  (-1,v_1)\\
-i\delta(z)_-&\mathrm{on}\quad  (v_2,1)\end{cases}
\]
It is easy to verify that $\widetilde {\mathcal N}(z)$ satisfies the same jump conditions as $\mathcal N$:
\begin{equation}
\begin{aligned}
\widetilde{\mathcal N}_+(z)&=\widetilde{\mathcal N}_-(z)\begin{pmatrix}
0&-1\\1&0
\end{pmatrix}&\textrm{for }z\in (v_2,1),\\
\widetilde{\mathcal N}_+(z)&=\widetilde{\mathcal N}_-(z)\begin{pmatrix}
0& e^{-2\pi i(s\Omega+k/2)}\\-e^{2\pi i(s\Omega+k/2)}&0
\end{pmatrix}=
\widetilde{\mathcal N}_-(z)
\begin{pmatrix}
0&-e^{-2\pi is\Omega}\\e^{2\pi is\Omega}&0
\end{pmatrix}
&\textrm{for }z\in (-1,v_1).
\end{aligned}
\end{equation}

Furthermore, one verifies that
$\delta(z)-\delta(z)^{-1}$ has two zeros at $z_+$, $z_-$ located on the first sheet and such that
$\delta(z_+)=\delta(z_-)=1$ and
\begin{equation}
z_{\pm}=\frac{v_1+v_2}{2}\pm i\sqrt{\nu |\alpha \beta|}+\mathcal O(\nu),\qquad  \nu \to 0.
\end{equation}

Set 
\[
d'=u(z_+)+1/2+\tau/2,
\]
then it follows by the properties of the Abel map $u(z)$ \eqref{def:u}
that $\theta(u(z)-d')$ has a single zero at $z_+$, and $\theta(u(z)+d')$ has no zeros on 
the first sheet
$\mathbb C\setminus A$. 
Thus $\widetilde {\mathcal N}(z_-)=I$, and since $\det \widetilde {\mathcal N}$ extends to an entire function, $\det \widetilde {\mathcal N}(z)=1$ for $z\in \mathbb C$.
Considering the zeros and poles of the meromorphic function $\delta^{-2}-1$ on $\Sigma$,
and using the Abel theorem, we have
\begin{equation}
u(v_1)+u(v_2)-u(z_-)-u(z_+)\equiv 0,
\end{equation}
modulo the lattice.
Since $u(v_1)+u(v_2)\equiv u(v_1)\equiv\frac{\tau}{2}$, 
\begin{align}u(z_+)+u(z_-)&\equiv -\tau/2, \label{u(z+)+u(z-)} \\ u(z_-)+d'&\equiv 1/2.\label{u(z-)+d'}\end{align} 

Using the change of integration variable $x=t+\frac{v_1+v_2}{2}$ as above, we obtain (from now on always on the first sheet, so modulo $\mathbb{Z}$)
\begin{equation}\begin{aligned}
u(z_+)=-\frac{i}{2I_0}\int_{v_2}^{z_+}\frac{dx}{p(x)^{1/2}}&=
-\frac{1}{2I_0\sqrt{|\alpha\beta|}}\int_{\nu}^{i\sqrt{|\alpha\beta|\nu}}\frac{dt}{(t^2-\nu^2)^{1/2}}\left(1+\mathcal O(\sqrt{\nu})\right)\\
&=-\frac{1+\tau}{4}-\frac{\hat \epsilon}{2}+\mathcal O(\sqrt \nu),\end{aligned}
\end{equation}
as $\nu \to 0$, where $\hat \epsilon$ is real,  satisfying $\hat \epsilon \to 0$ as $\nu \to 0$.
Similarly,
\begin{equation}
u(z_-)=\frac{1-\tau}{4}+\frac{\hat \epsilon}{2}+\mathcal O(\sqrt \nu).
\end{equation}
Therefore by the definition of $d'$, 
\begin{equation}d'=\frac{1+\tau}{4}-\frac{\hat \epsilon}{2}+\mathcal O(\sqrt \nu), \label{tau0asymd'}
\end{equation}
and 
\begin{equation}\nonumber
u(z_-)+d'=1/2.
\end{equation}
Thus, since $\theta$ is an even function,
\begin{equation}
 \widetilde m(z)=\frac{\theta(1/2)}{\theta(t+1/2)}\begin{pmatrix}
\frac{\theta(u(z)+t+d')}{\theta (u(z)+d')}&\frac{\theta(-u(z)+t+d')}{\theta (-u(z)+d')}\\
\frac{\theta(u(z)+t-d')}{\theta (u(z)-d')}&\frac{\theta(-u(z)+t-d')}{\theta (-u(z)-d')}
\end{pmatrix}.
\end{equation}
By \eqref{asymkappa}, $\nu \to 0$ corresponds to  $\kappa \to 0$.  By \eqref{tau0asymtheta},
\begin{equation}\frac{\theta(1/2)}{\theta(1/2+t)}=\frac{2\kappa^{|t|-|t|^2}}{1+\kappa^{2|t|}}\left(1+\mathcal O(\kappa)\right)=
\mathcal O\left(\nu^{|t|-|t|^2}\right),\qquad \nu \to 0. \end{equation}
As $\nu \to 0$, we have $u(z)\to 0$ uniformly for $z$ in the closure of $ U^{(v_1)}\cup U^{(v_2)}$, and
by \eqref{tau0asymd'},
\begin{equation}d'\pm u(z)= \langle d' \pm u(z)\rangle \to 1/4.
\end{equation}

Consider first the case $0<t\le 1/4$. 
Pick $0<\epsilon<\varepsilon/8$.
Then, uniformly on the closure of $U^{(v_1)}\cup U^{(v_2)}$,
\[
\frac{\theta(1/2)\theta(d'\pm u(z)+t)}{\theta(1/2+t)\theta(d'\pm u(z))}=
\mathcal O\left(\kappa^{t-t^2}\kappa^{t^2+2t(\pm u(z)+d')}\right)=
\mathcal O\left(\nu^{3t/2-\epsilon}\right),
\]
which is the asymptotics of $\widetilde m(z)_{11}$, $\widetilde m(z)_{12}$. Moreover,
\[
\frac{\theta(1/2)\theta(-d'\pm u(z)+t)}{\theta(1/2+t)\theta(-d'\pm u(z))}=
\mathcal O\left(\kappa^{t-t^2}\kappa^{t^2+2t(\pm u(z)-d')}\right)=
\mathcal O\left(\nu^{t/2-\epsilon}\right),
\]
which is the asymptotics of $\widetilde m(z)_{21}$, $\widetilde m(z)_{22}$.

For $1/4<t\le 1/2$, we have $\langle\pm u(z)+t+d'\rangle=\pm u(z)+t+d'-1$ so that
\[
\frac{\theta(1/2)\theta(d'\pm u(z)+t)}{\theta(1/2+t)\theta(d'\pm u(z))}=
\mathcal O\left(\kappa^{t-t^2}\kappa^{t^2+2t(\pm u(z)+d'-1)+(3/4)^2-(1/4)^2-\epsilon}\right)=
\mathcal O\left(\nu^{(1-t)/2-\epsilon}\right),
\]
which is the asymptotics of $\widetilde m(z)_{11}$, $\widetilde m(z)_{12}$, and finally
\[
\frac{\theta(1/2)\theta(d'\pm u(z)+t)}{\theta(1/2+t)\theta(d'\pm u(z))}=
\mathcal O\left(\nu^{t/2-\epsilon}\right),
\]
which is the asymptotics of $\widetilde m(z)_{21}$, $\widetilde m(z)_{22}$.

Similarly, we analyze the case of $-1/2<t\le 0$. Collecting the results together, we obtain
\begin{equation}\label{mest}
\widetilde m(z)=\mathcal O\left(\nu^{|t|/2-\epsilon}\right)+\mathcal O\left(\nu^{(1-|t|)/2-\epsilon}\right)=
\mathcal O\left(\nu^{-\epsilon}\right),\qquad  \nu\to 0,
\end{equation}
uniformly on the closure of $U^{(v_1)}\cup U^{(v_2)}$. By similar arguments, we obtain the same estimate also 
on the closure of $U^{(1)}\cup U^{(-1)}$ (in this case, $|u(z)+1/2|\le\epsilon'$, $\epsilon'>0$.)

On the other hand, the definition of $\delta$ gives
\begin{equation}\label{tau0delta1}
\delta(z)+\delta(z)^{-1}, \,\, \delta(z)-\delta(z)^{-1}=\mathcal O\left(\nu^{-1/4}\right), 
\end{equation} 
uniformly
for $z\in \partial U^{(p)}$  as $\nu\to 0$, for $p\in \mathcal T=\{-1,v_1,v_2,1\}$.

Thus, 
\begin{equation}\label{tau0orderN}
\widetilde {\mathcal N}(z), \widetilde {\mathcal N}(z)^{-1}=\mathcal O\left(\frac{1}{\nu^{1/4+\epsilon}}\right),
\end{equation}
as $\nu\to 0$, uniformly on $\partial U^{(p)}$ for $p\in \mathcal T$.

Since the solution to the RH problem for $\mathcal N$ is unique, we have
\begin{equation}\label{NtildeN}
 \mathcal N(z)=\widetilde{\mathcal N}(\infty)^{-1}\widetilde {\mathcal N}(z).\end{equation}

Define the new local parametrices by
\begin{equation}\label{NtildeP}
\widetilde P(z)=\widetilde{\mathcal N}(\infty)P(z),
\end{equation}
and let 
\begin{equation}
\widetilde R(z)=\begin{cases}\widetilde{\mathcal N}(\infty)S(z)\widetilde {\mathcal N}(z)^{-1} &z\in \mathbb C \setminus \cup_{p\in \mathcal T}U^{(p)},\\
\widetilde{\mathcal N}(\infty)S(z)\widetilde P(z)^{-1} &z\in \cup_{p\in \mathcal T}U^{(p)},\end{cases}
\end{equation}
Then $\widetilde R(z)\to 1$ as $z\to\infty$; and $\widetilde R(z)$ has jumps on $\Sigma_R$, see Figure \ref{ContR}.
By \eqref{PE} and the expansion of $\zeta$ in \eqref{tau0zeta}, 
the jumps of $\widetilde R(z)$ on $\partial U^{(p)}$ have the form
\begin{equation}\label{tau0PN-1}
\widetilde P(z)\widetilde{\mathcal N}^{-1}(z)=I+\widetilde \Delta(z), \qquad \widetilde \Delta(z)=\mathcal O\left(\frac{\log \nu^{-1}}{s\nu^{1/2+2\epsilon}}\right),\end{equation}
uniformly for $z\in U^{(p)}$ as $s\to \infty$ for $2\nu>s^{-2+\varepsilon}$.

For the proof of Lemma \ref{Lemtau0} below, we will also require the finer estimate 
\begin{equation}\label{tau0Delta}
\widetilde \Delta(z)=\widetilde \Delta_1(z)+\widetilde{\mathcal N}(z)\mathcal O\left(\frac{\left(\log \nu^{-1}\right)^2}{s^2}\right)\widetilde {\mathcal N}(z)^{-1},
\end{equation}
where 
\begin{equation}\label{tau0Delta1}
\begin{aligned}
\widetilde \Delta_1(z)
&=\frac{\mp 1}{8\sqrt{\zeta(z)}}\widetilde{\mathcal N}(z)e^{is\phi(p)\sigma_3}\begin{pmatrix}
-1&- 2i\\ - 2i &1
\end{pmatrix}e^{-is\phi(p)\sigma_3}\widetilde{\mathcal N}^{-1}(z),\qquad p=-1,v_2,\\
\widetilde \Delta_1(z)
&=\frac{\mp 1}{8\sqrt{\zeta(z)}}\widetilde{\mathcal N}(z)e^{is\phi(p)\sigma_3}\begin{pmatrix}
-1& 2i\\  2i &1
\end{pmatrix}e^{-is\phi(p)\sigma_3}\widetilde{\mathcal N}^{-1}(z),\qquad p=v_1,1,
\end{aligned}
\end{equation}
where $\mp$ means $+$ for $ \Im z<0$ and $-$ for $\Im z>0$.

By \eqref{JS0}, 
the jumps of $\widetilde R(z)$ on the rest of the contour  are estimated as follows (we descrease
$c>0$ somewhat)
\begin{equation}\label{tau0PN-inf}
\widetilde{\mathcal N}(z)J_S(z)\widetilde{\mathcal N}(z)^{-1}=I+
\mathcal O\left(e^{-c\sqrt{s}(|z|+1)}\right),\qquad c>0,
\end{equation} 
 as $s\to \infty$, uniformly for $2\nu>s^{-2+\varepsilon}$ and for $z\in\Gamma_{R,L}\cup \Gamma_{R,U}$.
Thus $\widetilde R$ satisfies a small-norm problem and therefore has a solution
for $s$ sufficiently large and $2\nu>s^{-2+\varepsilon}$, and 
\begin{equation} \label{smallRtilde}
\widetilde R(z)=I+\mathcal O\left(\frac{\log \nu^{-1}}{s\nu^{1/2+2\epsilon}}\right),\end{equation}
as $s\to \infty$, uniformly for $2\nu>s^{-2+\varepsilon}$, and uniformly for $z\in \mathbb C \setminus \Sigma_R$.

Since the RH problem for $\widetilde R$ has a unique solution, the RH problem for $S$ 
(and hence for $\Phi$) has a unique solution obtained by tracing back the transformations.

\subsection{Integration of the differential identity}

We now prove

\begin{Lemma} \label{Lemtau0}
Let $-1<V_1<\widehat V_2<1$ be fixed, and $V_1<V_2<\widehat V_2$ be such that $|V_2-V_1|>s^{-5/4}$.
Then, uniformly for such $V_2$ as $s\to \infty$,
\begin{equation}
\log \det (I-K_s)_ A-\log \det(I-K_s)_{(-1,V_1)\cup (\widehat V_2,1)}=\int_{\widehat V_2}^{ V_2}D(V_1,v_2)dv_2+\mathcal O(s^{-1/9}),
\end{equation}
where $D$ is defined in \eqref{DD} of Proposition \ref{PropD}.
\end{Lemma}

\begin{proof}
In this proof, $\epsilon$ stands for a sufficiently small positive constant whose value
may vary from line to line.

In the previous section, we obtained the asymptotic solution of the $S$-RH problem in the regime
$s\to\infty$, $2\nu>s^{-2+\epsilon}$. By 
\eqref{defnR}, $R$ is also well defined in this regime, 
\begin{equation}  R(z)=\widetilde N(\infty)^{-1}\widetilde R(z)\widetilde N(\infty),\label{RtildeR}
\end{equation}
and thus \eqref{exact} holds. We now aim to prove the analogue of \eqref{notexact}, namely  
\begin{equation}\label{tau0aim}
\mathcal F_s(v_1,v_2)=\frac{ s^2\zeta_0^2}{4}-\frac{\zeta_0 s}{4}
m_{11,0}m_{22,0}\left(\gamma_0^2 \Gamma_2+\Gamma_1\right) 
+\frac{i \zeta_0\gamma_0^2}{4}W(s\Omega)+
\mathcal O\left(\frac{1}{s\nu^{3/2+\epsilon}}+\frac{1}{s^2\nu^{5/2+\epsilon}}\right), \end{equation}
as $s\to\infty$, uniformly for $2\nu >s^{-5/4}$, with the same notation as in
\eqref{exact}, \eqref{notexact}.

 By \eqref{def:calN} and \eqref{def:tildeN}, using \eqref{idsmij} and similar identities for
 $\widetilde m_{jk}$, we obtain
 \begin{align} \label{limv2N1}
\widetilde{\mathcal N}(z)&=\frac{\delta^{-1}(z)}{2}\begin{pmatrix}\widetilde m_{11}(v_2) &i\widetilde m_{11}(v_2)\\
-i\widetilde m_{22}(v_2)&\widetilde m_{22}(v_2)\end{pmatrix}+\mathcal O\left((z-v_2)^{1/4}\right),\\
\mathcal N(z)&=\frac{\gamma(z)}{2}\begin{pmatrix} m_{11}(v_2) &i m_{11}(v_2)\\
-i m_{22}(v_2)& m_{22}(v_2)\end{pmatrix}+\mathcal O\left((z-v_2)^{1/4}\right), \label{limv2N2}
\end{align}
as $z\to v_2$.

Thus, substituting \eqref{limv2N1} and \eqref{limv2N2} into  \eqref{NtildeN} and taking the limit $z\to v_2$, we obtain
\begin{equation}\label{mtildem}\begin{aligned}
 \begin{pmatrix} m_{11}(v_2)\\ -im_{22}(v_2)\end{pmatrix}
&=\left(\lim_{z\to v_2}\frac{1}{\gamma(z)\delta(z)}\right)\widetilde {\mathcal N} (\infty)^{-1}
\begin{pmatrix} \widetilde m_{11}(v_2)\\ -i \widetilde m_{22}(v_2)\end{pmatrix},\\
\begin{pmatrix}
i  m_{22}(v_2)& m_{11}(v_2)
\end{pmatrix}&=\left(\lim_{z\to v_2}\frac{1}{\gamma(z)\delta(z)}\right)
\begin{pmatrix}
i  \widetilde m_{22}(v_2)&\widetilde  m_{11}(v_2)
\end{pmatrix}\widetilde N(\infty).
 \end{aligned}\end{equation}

By the definition of $\gamma$ and $\delta$ in \eqref{def:gamma} and \eqref{def:delta}, $\lim_{z\to v_2} \gamma(z)\delta(z)= \sqrt{2}\nu^{1/4}/\sqrt{(v_2+1)}$. Thus, by \eqref{RtildeR}, the third term on the right hand side of \eqref{exact} is given by
\begin{multline}\label{consttilde}
\frac{i s\zeta_0\gamma_0^2}{4}\begin{pmatrix}
i  m_{22,0}& m_{11,0}
\end{pmatrix}
 R^{-1}(v_2) R\, '(v_2)\begin{pmatrix}
 m_{11,0}\\ -i m_{22,0}
\end{pmatrix}\\ =
\frac{i s\zeta_0\gamma_0^2(1+v_2)}{8\nu^{1/2}}\begin{pmatrix}
i \widetilde m_{22}(v_2)&\widetilde m_{11}(v_2)
\end{pmatrix}
\widetilde R^{-1}(v_2)\widetilde R\, '(v_2)\begin{pmatrix}
\widetilde m_{11}(v_2)\\ -i\widetilde m_{22}(v_2)
\end{pmatrix} ,
\end{multline}
which we now evaluate. By \eqref{tau0x1x2}, \eqref{expzeta}, \eqref{expgamma},
\begin{equation} \label{tau0gamzet}
 \zeta_0\gamma_0^2=\mathcal O\left(\frac{1}{\log \nu^{-1}}\right),
\end{equation}
as $\nu\to 0$.

By the definition of $\Delta_1$, $\widetilde \Delta_1$, and by \eqref{NtildeN},
\begin{equation} \widetilde \Delta_1(z)=\widetilde N(\infty) \Delta_1(z)\widetilde N(\infty)^{-1}, \end{equation}
and thus, by \eqref{mtildem},   and \eqref{def:W},
\begin{equation}\label{def:W2}
W(\omega)=\frac{(1+v_2)}{2\sqrt{\nu}}
\begin{pmatrix}
i \widetilde m_{22}(v_2;\omega)&\widetilde m_{11}(v_2;\omega)
\end{pmatrix}
\sum_{p\in \mathcal T} \int_{\partial U^{(p)}} \frac{s\widetilde \Delta_1(z;\omega)}{(z-v_2)^2}\frac{dz}{2\pi i}\begin{pmatrix}
\widetilde m_{11}(v_2;\omega)\\ -i\widetilde m_{22}(v_2;\omega)
\end{pmatrix}.
\end{equation}

%and thus
%a preliminary investigation relying on \eqref{errortheta2} would reveal that the third term is of order  $\nu^{-2}$ (barring logarithmic terms) which is  unsatisfactory, but we now proceed with a closer inspection which will reveal that the main order terms cancel.

Note that $\widetilde R$ satisfies (we denote the jump of $\widetilde R$ on $\Sigma_R$ by $I+\widetilde \Delta(z)$)
\begin{equation} \label{Rtildeint}
\widetilde R(z)=I+\int_{\Sigma_R} \frac{\widetilde R_-(\xi)\widetilde \Delta(\xi)}{\xi-z}\frac{d\xi}{2\pi i}.
\end{equation}
 By \eqref{Rtildeint}, \eqref{tau0PN-1}, \eqref{tau0PN-inf},
 \eqref{smallRtilde}, and the fact that $U^{(v_1)}$ and $U^{(v_2)}$ have radius $\nu/3$,
\begin{equation}\begin{aligned}
\widetilde R\, '(v_2)&=\int_{\Sigma_R}\left(I+\int_{\Sigma_R}\frac{\widetilde R_-(u)\widetilde \Delta(u)}{u-\xi_-}\frac{du}{2\pi i} \right) \frac{\widetilde \Delta(\xi)}{(\xi-v_2)^2}\frac{d\xi}{2\pi i}\\
&=\int_{\partial U^{(v_1)}\cup \partial U^{(v_2)}}\left(I+\int_{\Sigma_R} \frac{\widetilde \Delta(u)}{u-\xi_-} du+\mathcal O\left(\frac{1}{s^2\nu^{1+4\epsilon}}\right)\right) \frac{\widetilde \Delta(\xi)}{(\xi-v_2)^2} \frac{d\xi}{2\pi i}\\
&\quad+\int_{\partial U^{(1)}\cup \partial U^{(-1)}} \frac{\widetilde \Delta(\xi)}{(\xi-v_2)^2} \frac{d\xi}{2\pi i}+\mathcal O\left(\frac{1}{s^2\nu^{1+4\epsilon}}\right),
\\
\widetilde R(v_2)^{-1}&=I-\int_{\Sigma_R}\frac{\widetilde \Delta(u)}{u-v_2}\frac{du}{2\pi i}+\mathcal O\left(\frac{1}{s^2\nu^{1+4\epsilon}}\right),
\end{aligned}\end{equation}
as $s\to \infty$, uniformly for $z\in \mathbb C\setminus \Sigma_R$ and for $2\nu>s^{-5/4}$. Thus,
\begin{multline}\nonumber
\widetilde R(v_2)^{-1}\widetilde R\, '(v_2)=\int_{\partial U^{(v_1)}\cup \partial U^{(v_2)}}\left(I+\int_{\Sigma_R} \widetilde \Delta(u)\left(\frac{1}{u-\xi_-}-\frac{1}{u-v_2}\right) du+\mathcal O\left(\frac{1}{s^2\nu^{1+\epsilon}}\right)\right) \\ \times \frac{\widetilde \Delta(\xi)}{(\xi-v_2)^2} \frac{d\xi}{2\pi i}
+\int_{\partial U^{(1)}\cup \partial U^{(-1)}} \frac{\widetilde \Delta(\xi)}{(\xi-v_2)^2} \frac{d\xi}{2\pi i}+\mathcal O\left(\frac{1}{s^2\nu^{1+\epsilon}}\right),\end{multline}
in the same limit. Since $\frac{1}{u-\xi_-}-\frac{1}{u-v_2}=\mathcal O(\nu)$ when $u\in \partial U^{(1)}\cup \partial U^{(-1)}$ and $\xi_-\in \partial U^{(v_1)}\cup \partial U^{(v_2)}$, we obtain
\begin{multline}\label{tau0RR}
\widetilde R(v_2)^{-1}\widetilde R\, '(v_2)=\int_{\partial U^{(v_1)}\cup \partial U^{(v_2)}}\Bigg(I+\int_{\partial U^{(v_1)}\cup \partial U^{(v_2)}} \widetilde \Delta(u)\left(\frac{1}{u-\xi_-}-\frac{1}{u-v_2}\right) du+\\ \mathcal O\left(\frac{1}{s^2\nu^{1+\epsilon}}+\frac{\nu^{1/2-\epsilon}}{s}\right)\Bigg)  \frac{\widetilde \Delta(\xi)}{(\xi-v_2)^2} \frac{d\xi}{2\pi i}
+\int_{\partial U^{(1)}\cup \partial U^{(-1)}} \frac{\widetilde \Delta(\xi)}{(\xi-v_2)^2} \frac{d\xi}{2\pi i}+\mathcal O\left(\frac{1}{s^2\nu^{1+\epsilon}}\right).\end{multline}

We will now estimate \eqref{consttilde}.
For estimates on $\partial U^{(-1)}\cup  \partial U^{(1)}$,
recall that by \eqref{mest}, $\widetilde m(v_2)$ is of order $\nu^{-\epsilon}$. 
For estimates on $\partial U^{(v_1)}\cup  \partial U^{(v_2)}$
we need more precise information:
note that by \eqref{limv2N1},
\begin{equation}
\begin{pmatrix}1\\
0\end{pmatrix}=
\widetilde{\mathcal N}(z)^{-1}\widetilde{\mathcal N}(z)
\begin{pmatrix}1\\
0\end{pmatrix}=
\frac{\delta^{-1}(z)}{2}
\widetilde{\mathcal N}(z)^{-1}
\begin{pmatrix}\widetilde m_{11}(v_2)\\
-i\widetilde m_{22}(v_2)\end{pmatrix}
+\mathcal O\left(\nu^{-1/4-\epsilon}\delta(z)\right),
\end{equation}
on $\partial U^{(v_1)}\cup  \partial U^{(v_2)}$, and therefore
\begin{equation}\label{Ntildem2}
 \widetilde {\mathcal N}(z)^{-1}\begin{pmatrix}
\widetilde m_{11}(v_2)\\ -i\widetilde m_{22}(v_2)
\end{pmatrix}=\mathcal O\left(\nu^{1/4-\epsilon}\right),
\end{equation}
as $\nu \to 0$ for $z\in \partial U^{(v_1)}\cup \partial U^{(v_2)}$. Similarly,
\begin{equation}\label{Ntildem3}
\begin{pmatrix} i\widetilde m_{22}(v_2)&\widetilde m_{11}(v_2)\end{pmatrix}\widetilde N(z)=\mathcal O\left(\nu^{1/4-\epsilon}\right).
\end{equation}

Estimates \eqref{tau0RR}, and \eqref{Ntildem2}, \eqref{Ntildem3} on
$\partial U^{(v_1)}\cup  \partial U^{(v_2)}$, and $\widetilde m(v_2)=\mathcal O(\nu^{-\epsilon})$, $\mathcal N(z)=\mathcal O(\nu^{-1/4-\epsilon})$ on 
$\partial U^{(-1)}\cup  \partial U^{(1)}$ imply that
\eqref{consttilde} can be written as
\begin{multline}\label{consttilde2}
\frac{i s\zeta_0\gamma_0^2}{4}\begin{pmatrix}
i  m_{22,0}& m_{11,0}
\end{pmatrix}
 R^{-1}(v_2) R\, '(v_2)\begin{pmatrix}
 m_{11,0}\\ -i m_{22,0}
\end{pmatrix}\\ =
\frac{i\zeta_0\gamma_0^2}{4}
\frac{s(1+v_2)}{2\nu^{1/2}}\begin{pmatrix}
i \widetilde m_{22}(v_2)&\widetilde m_{11}(v_2)
\end{pmatrix}
\widetilde R^{-1}(v_2)\widetilde R\, '(v_2)\begin{pmatrix}
\widetilde m_{11}(v_2)\\ -i\widetilde m_{22}(v_2)
\end{pmatrix}\\ =
\frac{i\zeta_0\gamma_0^2}{4}
W(s\Omega)+
\mathcal O\left(\frac{1}{s\nu^{3/2+\epsilon}}+\frac{1}{s^2\nu^{5/2+\epsilon}}\right).
\end{multline}

Thus we obtained \eqref{tau0aim}. After integration, the error term here yields the
one not larger than that of the statement of the lemma, $\mathcal O(s^{-1/9})$.

\bigskip

To finish the proof of the lemma we need to estimate the error of replacing $W$
with its average value. 
From the definition \eqref{def:W2} and the estimates above, we deduce
\begin{equation} 
f(\omega)=\zeta_0\gamma_0^2
W(\omega)=\mathcal O\left(\frac{1}{\nu^{1+\epsilon}}\right),\qquad\nu\to 0. 
\end{equation}
By \eqref{I0tau0},
\begin{equation}\label{Omegatau0}
\Omega=\frac{1}{I_0}=
\frac{\sqrt{|\alpha\beta|}}{\log \left(\gamma\nu\right)^{-1}}\left(1+\mathcal O(\nu^2)\right),\qquad 
\frac{\partial\Omega}{\partial v_2}=\mathcal O\left(\frac{1}{\nu\left(\log \nu^{-1}\right)^2}\right), \qquad \frac{\partial^2\Omega}{\partial v_2^2}=\mathcal O\left(\frac{1}{\nu^2 \left(\log \nu^{-1}\right)^2}\right),
\end{equation} 
as $\nu \to 0$. 

First, we have $f=\mathcal O\left(\nu^{-1-\epsilon}\right)$ and $\frac{\partial}{\partial v_2}f=\mathcal O\left(\nu^{-2-\epsilon}\right)$.
By the analysis leading to \eqref{mest}, $\frac{\partial}{\partial\omega}\widetilde m(v_2)=
\mathcal O(\nu^{-\epsilon}\log\nu)$,
$\omega=s\Omega$, and therefore, adjusting $\epsilon$, we also have
 $\frac{\partial}{\partial\omega}f=\mathcal O\left(\nu^{-1-\epsilon}\right)$ and $\frac{\partial}{\partial\omega}\frac{\partial}{\partial v_2}f=\mathcal O\left(\nu^{-2-\epsilon}\right)$. Thus, by \eqref{Fdom} and a similar expression for 
 $\frac{\partial}{\partial v_2}f_j$,
the right hand side of \eqref{fjint} is of order $\frac{1}{j^2 s\nu^{\epsilon}}$, and we obtain 
\begin{equation}\label{Fouriertau0}
\begin{aligned}
\int_{\widehat V_2}^{V_2}f(s\Omega;v_2,v_1)dv_2&=\sum_{j=-\infty}^\infty \int_{\widehat V_2}^{V_2} f_j(v_2,v_1)e^{2\pi ijs\Omega}dv_2\\
&=\int_{\widehat V_2}^{V_2}f_0(v_2,v_1)dv_2+\mathcal O\left(\frac{1}{s\nu^{\epsilon}}\right), \end{aligned} \end{equation}
as $s\to \infty$, uniformly for $2\nu>s^{-5/4}$. The error term here is better than the one of the statement of the lemma. Thus the lemma is proved.
\end{proof}

\subsection{Proof of Theorem \ref{Thmtau0}}

 By \eqref{newD} and Lemma \ref{Lemtau0}, we see that to show that the expansion \eqref{FormDIZ} holds in the asymptotic regime of Theorem \ref{Thmtau0} (with the error term $\mathcal O(s^{-1/9})$) it remains to 
 prove that
\begin{equation}\label{Fourier2tau0}
\int_{\widehat V_2}^{V_2}\left(\frac{\partial \tau}{\partial v_2}\int_0^1\frac{\partial }{\partial \tau} \log \theta_3(\omega;\tau)d\omega \right)
-\left(\frac{\partial \tau}{\partial v_2}\frac{\partial }{\partial \tau} \log \theta_3(s\Omega;\tau)\right) dv_2
=\mathcal O\left(\frac{1}{s\nu^{\epsilon}}\right). \end{equation}
Since by \eqref{derivativetau}, \eqref{I0tau0},
\begin{equation}
\frac{\partial \tau}{\partial v_2}=\frac{i\pi}{I_0^2(1-v_2^2)(v_2-v_1)}=
\mathcal O \left(\frac{1}{\nu \log^2 (\gamma\nu)^{-1}}\right),
\end{equation}
and by \eqref{tau0asymtheta}, \eqref{asymkappa},
\[
\frac{1}{\theta(\omega)}\frac{d^k}{d\omega^k}\theta(\omega)=
\mathcal O\left(\log^k(\gamma\nu)^{-1}\right),
\]
we obtain
\begin{equation}\label{logtau01}
\frac{\partial}{\partial\omega}\left(\frac{\partial \tau}{\partial v_2}\frac{\partial }{\partial \tau} \log \theta_3(\omega;\tau)\right)=\frac{1}{4\pi i} 
\frac{\partial \tau}{\partial v_2}
\left(\frac{\theta''_3}{\theta_3}\right)'(\omega)=
 \mathcal O\left(\frac{\log(\gamma\nu)^{-1}}{\nu}\right). 
\end{equation}
Also since by \eqref{Partial5},
\begin{equation}
\frac{\partial^2 \tau}{\partial v_2^2}=
\mathcal O \left(\frac{1}{\nu^2\log^2 (\gamma\nu)^{-1}}\right),
\end{equation}
we similarly obtain
\begin{equation}\label{logtau02}
\frac{\partial}{\partial\omega}\frac{\partial}{\partial v_2}\left(\frac{\partial \tau}{\partial v_2}\frac{\partial }{\partial \tau} \log \theta_3(\omega;\tau)\right)= \mathcal O\left(\frac{\log(\gamma\nu)^{-1}}{\nu^2}\right). 
\end{equation}
The estimates \eqref{logtau01} and \eqref{logtau02} imply, as in the proof of \eqref{Fouriertau0},
the estimate \eqref{Fourier2tau0}.
Thus, we have proven the first statement of Theorem \ref{Thmtau0}.

Since we have proven the uniformity of Theorem \ref{Thm} for $2\nu>s^{-5/4}$, all that remains to show \eqref{limThm1tau0} is to expand $G_0$, $\log \theta_3(s\Omega;\tau)$, and $c_1$ as $\nu \to 0$.

By \eqref{G0} and \eqref{tau0x1x2},
\begin{equation} G_0=\frac{1}{2}-\frac{|\alpha \beta|}{\log (\gamma \nu)^{-1}}+\mathcal O\left(\nu^2\right),
\end{equation}
as $\nu \to 0$. 

By the formula for $\Omega $ in \eqref{Omegatau0}, $\theta$ in  \eqref{tau0asymtheta}, $\kappa$ in \eqref{asymkappa}, $\tau$ in \eqref{tau0tau},
\begin{equation} \log\theta_3(s\Omega;\tau)=\frac{1}{2}\log \log (\gamma \nu)^{-1}
-\langle \omega_0 \rangle^2 \log (\gamma \nu)^{-1}
 +\log \left(1+ (\gamma\nu)^{1-2|\langle \omega_0 \rangle |}\right)-\frac{1}{2}\log \pi+o(1),
 \end{equation}
as $s\nu\to 0$, where 
\[
s\Omega=\omega_0+o(1),
\]
with $\omega_0$ given by \eqref{omegaintro}.

By the asymptotics for  $I_0$ in  \eqref{I0tau0} and $x_1x_2$ in \eqref{tau0x1x2}, and by \eqref{x1x2eqn1},
 \begin{equation}
 c_1=-\frac{1}{4}\log \log (\gamma \nu)^{-1}-\frac{1}{8}\log |\alpha \beta|+\frac{1}{2}\log \pi+2c_0+ o(1),
 \end{equation}
 as $\nu \to 0$.
Thus we obtain \eqref{limThm1tau0} if $s\nu\to 0$.

\appendix
\numberwithin{equation}{section}
\renewcommand{\theequation}{\thesection.\arabic{equation}}

\section{$\theta$-functions and elliptic integrals}\label{App1}

Here we collect the properties of Jacobian $\theta$-functions and elliptic integrals
we need in the main text. For more information on the topic, see \cite{WW,GR,Spr}. 

The third Jacobian $\theta$-function is defined by a series\footnote{
$\theta$-functions are defined in \cite{WW} with argument $z/\pi$.}:
\begin{equation} \label{def:theta}
\theta_3(z;\tau)\equiv\theta_3(z)\equiv\theta(z)=\sum_{m\in \mathbb Z} e^{2\pi i zm+\pi i \tau m^2},\qquad \Im\tau>0.
\end{equation}
The function
$\theta(z)$ satisfies the periodicity properties:
\begin{equation}\label{theta3 period}
\theta(z)=\theta(z+1),\qquad
\theta(z\pm \tau)=e^{\mp 2\pi i z-\pi i \tau} \theta(z). 
\end{equation}
It is an entire function which is even, $\theta(z)=\theta(-z)$. Furthermore, $\theta(z)$ has a single zero modulo the lattice $(\mathbb Z,\tau\mathbb Z)$ at $\frac{1+\tau}{2}$, and at the zero the derivative $\theta'(z)$ is non-zero.

The first, second, and fourth $\theta$-functions are then defined as follows:
\begin{equation}\label{th124}
\begin{aligned}
\theta_1(z)&=ie^{-\pi iz+\frac{\pi i\tau}{4}}\theta_3\left(z-\frac{\tau+1}{2}\right),\\
\theta_2(z)&=\theta_1(z+1/2)=e^{-\pi i z+\pi i \tau/4}\theta_3\left(z-\frac{\tau}{2}\right), \qquad \theta_4(z)=\theta_3(z+1/2). 
\end{aligned}
\end{equation}
The function $\theta_1(z)$ is odd, while $\theta_2(z)$, $\theta_4(z)$ are even. 
The unique zeros (modulo the lattice) of $\theta_1$, $\theta_2$ and $\theta_4$ are at $0$,$1/2$ and $\tau/2$ respectively, and we have the periodicity properties:
\begin{equation}\label{theta period}
\begin{aligned}
\theta_1(z+1)&=-\theta_1(z), \qquad \theta_1(z+\tau)=-e^{-2\pi iz-\pi i \tau}\theta_1(z),\\
\theta_2(z+1)&=-\theta_2(z), \qquad \theta_2(z+\tau)=e^{-2\pi iz-\pi i \tau}\theta_2(z),\\
\theta_4(z+1)&=\theta_4(z), \qquad \theta_4(z+\tau)=e^{-2\pi iz-\pi i \tau}\theta_4(z).
\end{aligned}
\end{equation}
From the periodicity properties we have
\begin{equation}\label{almost_elliptic}\begin{aligned}
\frac{\theta'_j(z+1)}{\theta_j(z+1)}&=\frac{\theta_j'(z)}{\theta_j(z)},\qquad 
\frac{\theta_j'(z+\tau)}{\theta_j(z+\tau)}=\frac{\theta_j'(z)}{\theta_j(z)}-2\pi i,\\
\frac{\theta''_j(z+1)}{\theta_j(z+1)}&=\frac{\theta_j''(z)}{\theta_j(z)},\qquad 
\frac{\theta_j''(z+\tau)}{\theta_j(z+\tau)}=\frac{\theta_j''(z)}{\theta_j(z)}-4\pi i\frac{\theta_j'(z)}{\theta_j(z)}-4\pi^2,\qquad j=1,2,3,4.
\end{aligned}
\end{equation}

We denote $\theta_j=\theta_j(0)$, and the derivatives at zero $\theta'_j=\theta'_j(0)$,
etc. In particular, we have expansions at zero: $\theta_3(z)=\theta_3+\frac{z^2}{2}\theta_3''+\cdots$, $\theta_1(z)=z\theta_1'+z^3\frac{\theta_1'''}{6}+\cdots$.

We will use representations of $\theta_3$ in terms of $\theta_1$. By \eqref{th124},
%\begin{equation}\label{EllDIZ}
%\frac{\theta_3\left( z+\omega\right)\theta_3\left( z-\omega\right)}{\theta_3\left( z\right)^2}=\frac{\theta_1( \nu+\omega)\theta_1( \nu-\omega)}{\theta_1( \nu)^2},\qquad 
%\nu=z-\frac{1+\tau}{2}.
%\end{equation}
%Similarly, we obtain 
\begin{equation}\label{theta'/theta}
\frac{\theta_3'(z)}{\theta_3(z)}=\frac{\theta_1'(\nu)}{\theta_1(\nu)}-\pi i, \qquad \nu=z-\frac{1+\tau}{2},
  \end{equation}
and
\begin{equation}\label{theta''/theta}
\frac{\theta_3''(z)}{\theta_3(z)}=\frac{\theta_1''(\nu)}{\theta_1(\nu)}-2\pi i\frac{\theta_1'(\nu)}{\theta_1(\nu)}-\pi^2,\qquad \nu=z-\frac{1+\tau}{2}.
\end{equation}

$\theta$-functions satisfy Jacobian addition relations, of which we will make use 
of the following two:
\begin{align}\label{theta23}
\theta_2(x+y)\theta_3(x-y)+\theta_2(x-y)\theta_3(x+y)&=\frac{2}{\theta_2\theta_3}\theta_2(x)\theta_2(y)\theta_3(x)\theta_3(y),\\
\label{theta34}
\theta_4(x+y)\theta_3(x-y)+\theta_4(x-y)\theta_3(x+y)&=\frac{2}{\theta_4\theta_3}\theta_4(x)\theta_4(y)\theta_3(x)\theta_3(y).
\end{align}

$\theta$-functions satisfy the differential equation
\begin{equation}\label{thdiff}
\theta''_j(z)=4\pi i \frac{\partial}{\partial\tau}\theta_j(z),\qquad j=1,2,3,4,
\end{equation}
some useful for us well-known identities for the values at zero:
\[
\theta_1'=\pi\theta_2\theta_3\theta_4,\qquad \theta_3^4=\theta_2^4+\theta_4^4,
\] 
and the following transformation formula for $\tau\to 1/\tau$,
\begin{equation}\label{tau1overtau}
\theta_3(z)=\frac{1}{\sqrt{-i\tau}}\sum_k e^{-\frac{i\pi}{\tau}(k-z)^2}.
\end{equation}

We will also need the following identity:
\begin{equation}\label{ellfuns4}
 \left(\frac{\theta_3'(z)}{\theta_3(z)}\right)'=\left(\frac{\theta_1'}{\theta_3}\right)^2\frac{\theta_1(z)^2}{\theta_3(z)^2}+\frac{\theta_3''}{\theta_3}.
\end{equation}
To show it, we first observe that both sides of the equation are elliptic functions
(i.e. they satisfy the periodicity relations $f(z+1)=f(z)$, $f(z+\tau)=f(z)$) with
second order pole at $z=(1+\tau)/2$. Considering the expansions of these functions
at the pole, we obtain that the difference of these functions has a pole of order
at most 1, and is therefore a constant. This constant is then evaluated setting $z=0$. 

Changing variable $z=\nu+\frac{1+\tau}{2}$
in \eqref{ellfuns4}, we also obtain 
\begin{equation}\label{ellfuns5}
 \left(\frac{\theta_1'(\nu)}{\theta_1(\nu)}\right)'=-\left(\frac{\theta_1'}{\theta_3}\right)^2\frac{\theta_3(\nu)^2}{\theta_1(\nu)^2}+\frac{\theta_3''}{\theta_3}.
\end{equation}

We further have 
\begin{Lemma}\label{LemmaComp}
If $g(z)$  is an elliptic function with a single pole modulo the lattice, located at $z=\frac{1+\tau}{2}$,  and
\begin{equation}g\left(\nu+\frac{1+\tau}{2}\right)=c_1\nu^{-2}+\mathcal O \left(\nu^{-1}\right),\end{equation}
as $\nu \to 0$,
then 
\begin{equation}\label{LemAsymg}
g(z)=-c_1\left[\left(\frac{\theta_3'(z)}{\theta_3(z)}\right)'-\frac{\theta_3''}{\theta_3}\right] +g(0),
\end{equation}
and furthermore
\begin{equation}\label{Lemmaintg}\int_0^1g(z)dz=c_1\frac{\theta_3''}{\theta_3}+g(0).\end{equation}
\end{Lemma}
\begin{proof} The second part of the lemma, \eqref{Lemmaintg}, follows directly from \eqref{LemAsymg}.

To show  \eqref{LemAsymg} note first that
since $\theta_3(z)$ has a zero of order 1 at $\frac{1+\tau}{2}$,  
\begin{equation}\label{Asymtheta3}
 \frac{\theta_3'(z)}{\theta_3(z)}=\frac{1}{z-\frac{1+\tau}{2}}+\mathcal O(1),\end{equation}
as $z\to \frac{1+\tau}{2}$.
By the fact that $\left(\frac{\theta_3'(z)}{\theta_3(z)}\right)' $ is elliptic and the hypothesis of the theorem,
\begin{equation}\label{LemAsymg2}
g(z)+c_1\left(\frac{\theta_3'(z)}{\theta_3(z)}\right)' 
\end{equation}
is an elliptic function with a single simple pole modulo the lattice, and therefore is a
constant. By \eqref{ellfuns4}, this constant is $g(0)+c_1\frac{\theta_3''}{\theta_3}$.
This shows \eqref{LemAsymg}.
\end{proof}

\begin{Lemma}\label{2integrals}
We have
\begin{equation}\label{intlogder4}
\int_0^1 \left(\frac{\theta_3'(z)}{\theta_3(z)}\right)^2 dz=\frac{\pi^2}{3}+\frac{\theta_1'''}{3\theta_1'},
\end{equation}
and, for any $d$, $u$,
\begin{equation}\label{thint2}
\int_0^1\frac{\theta_3(z-d)\theta_3(z+u+d)}{\theta_3(z)^2}dz=\frac{\pi \left[\theta_1'(d)\theta_1(u+d)-\theta_1(d)\theta_1'(u+d)\right]}{\left(\theta_1'\right)^2  \sin(\pi u) }.
\end{equation}
\end{Lemma}

\begin{proof}
 Since 
\begin{equation}\label{logder0}\int_0^1\frac{\theta_3'(\omega)}{\theta_3(\omega)}d\omega=0,\end{equation} 
we have by the relation between the logarithmic derivatives of $\theta_1$ and $\theta_3$ in \eqref{theta'/theta},
\begin{equation}\label{intlogder01}
\int_0^1 \left(\frac{\theta_3'(z)}{\theta_3(z)}\right)^2 dz=
\pi^2+\int_0^1\left(\frac{\theta_3'(z)}{\theta_3(z)}+\pi i \right)^2 dz
=\pi^2+\int_J\left(\frac{\theta_1'(\nu)}{\theta_1(\nu)}\right)^2d\nu,
\end{equation}
where 
\begin{equation}J=\left\{\nu=z-\frac{1+\tau}{2}:z\in(0,1)\right\}. \end{equation}
Let $\widetilde \Gamma$ be the rectangle with corners $\pm 1/2 \pm\tau/2$, with positive orientation. 
Writing the integral around the contour and using
the periodicity relation of $\theta_1'/\theta_1$ in \eqref{almost_elliptic}, we obtain
\begin{equation}\int_{\widetilde \Gamma}\left(\frac{\theta_1'(\nu)}{\theta_1(\nu)}\right)^3d\nu=
6 \pi i \int_J\left(\frac{\theta_1'(\nu)}{\theta_1(\nu)}\right)^2d\nu+12\pi^2\int_J\frac{\theta_1'(\nu)}{\theta_1(\nu)}d\nu-8\pi^3i.
\end{equation}
By \eqref{theta'/theta}, and \eqref{logder0}, $\int_J\frac{\theta_1'(\nu)}{\theta_1(\nu)}d\nu=\pi i$, and therefore
\begin{equation} \label{intlogder2}
\int_J\left(\frac{\theta_1'(\nu)}{\theta_1(\nu)}\right)^2d\nu
=-\frac{2\pi^3}{3}+\frac{1}{6\pi i}\int_{\widetilde \Gamma}\left(\frac{\theta_1'(\nu)}{\theta_1(\nu)}\right)^3d\nu.\end{equation}
Since $\theta_1$ has a single zero modulo the lattice located at $0$, and since $\theta_1''(0)=0$, we obtain
\begin{equation}\label{intlogder3}
\int_{\widetilde \Gamma}\left(\frac{\theta_1'(\nu)}{\theta_1(\nu)}\right)^3d\nu=2\pi i \frac{\theta_1'''}{\theta_1'}\end{equation}
by evaluating the residue of $\left(\frac{\theta_1'(\nu)}{\theta_1(\nu)}\right)^3$ at $0$. 
Combining \eqref{intlogder01}, \eqref{intlogder2}, and \eqref{intlogder3},
we obtain \eqref{intlogder4}.

To obtain \eqref{thint2}, we first observe that by \eqref{th124}, \eqref{theta period},
\begin{equation}\label{formulaT34}
\int_0^1\frac{\theta(z-d)\theta(z+u+d)}{\theta(z)^2}dz=e^{-\pi i u}\int_J\frac{\theta_1(\nu-d)\theta_1(u+\nu+d)}{\theta_1(\nu)^2}d\nu,\end{equation}
where again $J=\{\nu=z-\frac{1+\tau}{2}, z\in(0,1)\}$. With $\widetilde \Gamma$ as above, we have by periodicity properties that
\begin{equation}\label{formulaT33}
\int_{\widetilde \Gamma}\frac{\theta_1(\nu-d)\theta_1(u+\nu+d)}{\theta_1(\nu)^2}d\nu=\left(1-e^{-2\pi i u}\right)\int_J \frac{\theta_1(\nu-d)\theta_1(u+\nu+d)}{\theta_1(\nu)^2}d\nu.
\end{equation}
On the other hand, computing the residue, we obtain
\begin{equation}\label{formulaT32}
\int_{\widetilde \Gamma}\frac{\theta_1(\nu-d)\theta_1(u+\nu+d)}{\theta_1(\nu)^2}d\nu=\frac{2\pi i}{(\theta_1')^2}\left(\theta_1'(d)\theta_1(u+d)-\theta_1(d)\theta_1'(u+d)\right).
\end{equation}
The last 3 equations give \eqref{thint2}.

\end{proof}

\bigskip

Recall the definition of the elliptic integrals $I_j=I_j(v_1,v_2)$, $J_j=J_j(v_1,v_2)$ from \eqref{IJ}.

\begin{Lemma}
There holds a Riemann's period relation:
\begin{equation}\label{RBL}
\left(I_2-\frac{v_1+v_2}{2}I_1\right)J_0-
I_0\left(J_2-\frac{v_1+v_2}{2}J_1\right)=\pi.
\end{equation}
\end{Lemma}

\begin{proof}
We cut the Riemann surface $\Sigma$ along the loops $A_1$, $B_1$, which yields a 4-gon $\gamma$ with the sides
$A_1$, $B_1$, $A_1^{-1}$, $B_1^{-1}$ (the side $A_1$ is identified with $A_1^{-1}$ on the surface, 
the same with  $B_1$, $B_1^{-1}$). The standard Riemann period relation between meromorphic differentials $\lambda$, $\mu$ on $\Sigma$
is as follows:
\begin{equation}
\int_{\gamma} \Lambda\mu=
\int_{A_1}\lambda\int_{B_1}\mu-\int_{A_1}\mu\int_{B_1}\lambda,\qquad \Lambda(x)=\int_{x_0}^x\lambda,\qquad x\in\Sigma,
\end{equation}
where $\gamma$ is traversed in the positive direction, and where $x_0$ is a fixed point on the surface away from the cuts.

Now taking $\lambda=\frac{x^2-x(v_1+v_2)/2}{p(x)^{1/2}}dx$, $\mu=\frac{dx}{p(x)^{1/2}}$, we have in the local variable
$\xi=1/z$, $\lambda=\mp(1+\mathcal O(\xi^2))\frac{d\xi}{\xi^2}$, $\mu=\mp(1+\mathcal O(\xi))d\xi$, as
$\xi\to 0$. Here the upper sign is taken on the first sheet, and the lower one on the second. 
Computing the residue at $z$-infinity
(at 2 points on $\Sigma$ corresponding to it) of $\Lambda\mu$, we obtain \eqref{RBL}.
\end{proof}

The complete elliptic integrals of first and second kind, respectively, are defined
as follows:
\begin{equation}\label{EllInt}
K(v)=\int_0^1\frac{dt}{\sqrt{(1-t^2)(1-v^2t^2)}},\qquad E(v)=\int_0^1 \sqrt{\frac{1-v^2t^2}{1-t^2}}dt.
\end{equation}
Moreover, let
\begin{equation}
K'(v)=\int_1^{1/v}\frac{dt}{\sqrt{(t^2-1)(1-v^2t^2)}},\qquad \widehat E(v)=\int_1^{1/v} \sqrt{\frac{1-v^2t^2}{t^2-1}}dt.
\end{equation}
It is well-known that
\begin{equation}\label{KandK'}
K'(v)=K(v'),\qquad  v'=\sqrt{1-v^2}.
\end{equation}
%and
%\begin{equation}\label{Standard}
%E(v)K(v')+E(v')K(v)-K(v)K'(v)=\pi/2.
%\end{equation}
By integrating the derivative of $t\sqrt{\frac{1-t^2}{1-v'^2t^2}}$, we also obtain that
\begin{equation}\label{HatE}
 \widehat E(v)=K(v')-E(v').
\end{equation}

As $v\to 1$ (and therefore $v'\to 0$), we have the expansions:
\begin{equation}\label{KEexp}
\begin{aligned}
K(v)&=\left(\frac{1}{2}\log \frac{1}{2-2v}+2\log 2\right)(1+\mathcal O(1-v)),\\
K(v')&=\frac{\pi}{2}\left(1
+\frac{v'^2}{4}+\frac{9v'^4}{64}+
\mathcal O(v'^6)\right), \qquad 
E(v')=\frac{\pi}{2}\left(1
-\frac{v'^2}{4}-\frac{3v'^4}{64}+
\mathcal O(v'^6)\right).
\end{aligned}
\end{equation}

Now consider the case symmetric intervals $-v_1=v_2\equiv v$.
By the change of variable $x=vy$ and by using \eqref{HatE}, we see that
\begin{equation}\label{I02}
I_0(-v,v)=K(v'),\qquad \frac{I_2(-v,v)}{I_0(-v,v)}=1-\frac{\widehat E(v)}{K(v')}=\frac{E(v')}{K(v')},\qquad J_0(-v,v)=2K(v).
\end{equation}
%and, using also \eqref{Standard} for the last formula,
%\begin{equation}\label{JOmega}
%J_0(-v,v)=2K(v),\qquad J_2=2(K(v)-E(v)),\qquad  \Omega=\frac{1}{\pi}\int_{-v}^v\frac{I_2/I_0-x^2}{\sqrt{|p(x)|}}dx=\frac{1}{K(v')}.
%\end{equation}

%The following formulae are obtained by reducing elliptic integrals to complete ones (see, e.g., \cite{WW})
%and then applying Landen's transform (cf. \cite{GR}, equation (3.147)):
%\begin{align}
%\label{intv21K'}I_0&=\int_{v_2}^{1}\frac{dx}{\sqrt{|p(x)|}}
%=\frac{2}{\sqrt{(1-v_1)(1+v_2)}}K(r),\\
%\label{intv21K}J_0&=\int_{v_1}^{v_2}\frac{dx}{\sqrt{|p(x)|}}=
%\frac{2}{\sqrt{(1-v_1)(1+v_2)}}K'(r),
%\end{align}
 % where
%\begin{equation}
%\label{def:k}
%k&=\frac{-v_1\sqrt{1-v_2^2}+v_2\sqrt{1-v_1^2}}{\sqrt{1-v_2^2}+\sqrt{1-v_1^2}},\\
%\label{def:beta}\beta&=\sqrt{\left| \frac{2+2v_1v_2+2\sqrt{(1-v_1^2)(1-v_2^2)}}{2-v_1^2-%v_2^2+2\sqrt{(1-v_1^2)(1-v_2^2)}}\right|},\\
%\label{def:r}
%r=\sqrt{\frac{(1-v_2)(1+v_1)}{(1-v_1)(v_2+1)}}.
%\end{equation}
%
%Elliptic integrals are related to $\theta$-functions via $\tau$, in particular,
%\begin{equation}
%\theta_3^2(0;\tau)=\frac{2K(r)}{\pi}, 
%\qquad
%\tau=i\frac{J_0}{I_0}=i\frac{K'(r)}{K(r)},
%\qquad r=\frac{\theta_2^2(0;\tau)}{\theta_3^2(0;\tau)}.
%\label{Link}\end{equation}
%Thus,
%\begin{equation}\label{Link2}
%\theta_3^2(0;\tau)=\frac{I_0\sqrt{(1-v_1)(1+v_2)}}{\pi},\qquad \tau=i\frac{J_0}{I_0}.
%\end{equation}

\section{Prefactor of $\log s$}\label{App2}

Here we show that the constant $\widehat G_1$ in \eqref{FormDIZ} obtained in \cite{DIZ} is equal to $-1/2$.
Let 
\begin{equation} u(z)=-\frac{i}{2I_0}\int_{v_2}^z\frac{d\xi}{p(\xi)^{1/2}}, \end{equation}
and define
\begin{equation}\nonumber
\rho(z,\omega)= \frac{\theta^2(0) \theta(u(z)+\omega-u(\infty))\theta(u(z)-\omega-u(\infty))}{\theta^2(\omega)\theta^2(u(z)-u(\infty))},\qquad d=-u(\infty).
\end{equation}
It is easily verified that $\rho$ as a function of $\omega$ is elliptic: $\rho(\omega)=\rho(\omega+1)=
\rho(\omega+\tau)$.
Here we use our definitions of $u(z)$ \eqref{def:u} and $d$ (which has the property \eqref{uinftyd}) from Section \ref{secPhi}. However, it is straightforward to verify that $\rho$ is exactly the function (1.30) in \cite{DIZ} for $n=1$ with
$x=\omega/\Omega$, $V=\Omega$.

Let 
\begin{equation}
h(z)=(z-1)(z-v_1)+(z-v_2)(z+1),
\end{equation}
and consider the function $G_1$ given by (1.33) in \cite{DIZ}, which in our case of $n=1$ becomes
\begin{equation}\nonumber
G_1(t)=-\frac{1}{16}\sum_{y=\{-1,v_1,v_2,1\}}\rho(y,t\Omega)\frac{h(y)}{q(y)}.
\end{equation}
It was shown in \cite{DIZ} that the coefficient $\widehat G_1$ in \eqref{FormDIZ} is given by
\begin{equation}\nonumber
\widehat G_1=\lim_{x\to \infty}\frac{1}{x}\int_{x_0}^x G_1(t)dt,
\end{equation}
for some fixed large $x_0$.

By ellipticity of $\rho$, this can be written in the form
\begin{equation}
\widehat G_1=-\frac{1}{16}\sum_{y=\{-1,v_1,v_2,1\}}\frac{h(y)}{q(y)}
\int_0^1\rho(y,\omega)d\omega.
\end{equation}
To compute the integral, note first that by \eqref{th124}
\begin{equation}
\begin{aligned}
\rho \left(y, \nu+\frac{1+\tau}{2}\right)&=\frac{\theta_3^2}{\theta_3^2(u(z)+d)}
\frac{\theta_1(u(z)+d+\nu)\theta_1(-u(z)-d+\nu)}{\theta_1^2(\nu)}\\
&=-\frac{\theta_3^2}{\theta_3^2(u(z)+d)}
\frac{\theta_1^2(u(z)+d)}{(\theta_1')^2 \nu^2}+\mathcal O\left(\nu^{-1}\right),\qquad \nu\to 0.
\end{aligned}
\end{equation}
Using Lemma \ref{LemmaComp} in Appendix \ref{App1}, we compute the integral $\int_0^1\rho(y,\omega)d\omega$
and obtain
 \begin{equation}
\widehat G_1=-\frac{1}{16}\sum_{y\in\{-1,v_1,v_2,1\}}\frac{h(y)}{q(y)}\left(1-
\frac{\theta_3\theta_3''}{(\theta_1')^2}\frac{\theta_1^2(u(y)+d)}{\theta_3^2(u(y)+d)}\right).
\end{equation}
By applying the identities  \eqref{idd} of Proposition \ref{ThmThetaids} (d), 
\begin{equation}\label{FormulaG2}
\widehat G_1=-\frac{1}{16}\sum_{y\in\{-1,v_1,v_2,1\}}\frac{1}{q(y)}\left(h(y)+\frac{\theta_3''}{\theta_3 I_0^2}\right).
\end{equation}
By \eqref{dtheta33}, 
\begin{equation}
\frac{\theta_3''}{\theta_3 I_0^2}=2q(v_2)-h(v_2),
\end{equation}
and therefore the term with $y=v_2$ in \eqref{FormulaG2} is
\[
\frac{1}{q(v_2)}\left(h(v_2)+\frac{\theta_3''}{\theta_3 I_0^2}\right)=2.
\]
Now note (recall \eqref{x1x2eqn1}) that
\begin{equation}
2q(v_2)-h(v_2)=2q(v_1)-h(v_1)=2q(1)-h(1)=2q(-1)-h(-1)=v_2-v_1+2x_1x_2,
\end{equation}
so that all the other terms in the sum in \eqref{FormulaG2} are also equal $2$.
Therefore 
\begin{equation}
\widehat G_1=-\frac{1}{16}(2+2+2+2)=-\frac{1}{2}.
\end{equation}

\section*{Acknowledgements}
The work of the authors was partly supported by the Leverhulme Trust research project grant
RPG-2018-260.

\end{document}